\documentclass[a4paper]{amsart}

\usepackage[british]{babel}
\usepackage[OT2,T2A,T1]{fontenc}
\usepackage[utf8]{inputenc}
\usepackage{eulervm}
\usepackage[scr=dutchcal, bb=boondox]{mathalfa}
\usepackage{physics}
\usepackage{amsmath,amsthm,amssymb,dsfont,extarrows,amscd,amsfonts,mathtools,bm,physics}
\usepackage[mathscr]{euscript}
\usepackage{MnSymbol}

\usepackage{hyperref}
\usepackage[capitalize]{cleveref}
\usepackage{ multirow} 
\usepackage{fullpage, color}
\usepackage[labelformat=empty]{subfig}
\usepackage{booktabs}
\usepackage{enumerate, longtable}
\usepackage{makecell}
\usepackage{IEEEtrantools}
\usepackage{todonotes}

\usepackage[alphabetic, initials]{amsrefs}

 \usepackage{relsize}
\usepackage{tikz-cd}
\usepackage{xcolor}

\usepackage[hmargin=2.6cm,vmargin=3.2cm]{geometry}
\usepackage{marginnote, mparhack}

\usepackage{tikz}
\usetikzlibrary{decorations.pathreplacing,angles,quotes,decorations.markings}

\tikzset{
  midarrow/.style={
    postaction={decorate},
    decoration={markings, mark=at position #1 with {\arrow{>}}}
  },
  midarrow/.default=0.55
}

\tikzset{
  contour arrows/.style 2 args={
    postaction={decorate},
    decoration={
      markings,
      mark=at position #1 with {\arrow{>}},
      mark=at position #2 with {\arrow{>}}
    }
  },
  contour antipodal/.style={
    contour arrows={0.20}{0.70} % change 0.20 to slide both together
  },
  contour antipodal reversed/.style={
    postaction={decorate},
    decoration={
      markings,
      mark=at position 0.20 with {\arrowreversed{>}},
      mark=at position 0.70 with {\arrowreversed{>}}
    }
  }
}

\crefname{equation}{Eq.}{Eqs.}
\crefname{eqnarray}{Eq.}{Eqs.}
\crefname{algo}{Algorithm}{Algorithms}
\crefname{conj}{Conjecture}{Conjectures}
\crefname{lem}{Lemma}{Lemmas}
\crefname{thm}{Theorem}{Theorems}
\crefname{claim}{Claim}{Claims}
\crefname{rmk}{Remark}{Remarks}
\crefname{prop}{Proposition}{Propositions}
\crefname{section}{Section}{Sections}
\crefname{appendix}{Appendix}{Appendices}
\crefname{cor}{Corollary}{Corollaries}
\crefname{figure}{Figure}{Figures}
\crefname{table}{Table}{Tables}
\crefname{example}{Example}{Examples}
\crefname{prob}{Problem}{Problems}
\crefname{assm}{Assumption}{Assumptions}
\crefname{defn}{Definition}{Definitions}

\DeclareMathOperator{\diag}{diag}

\def\bary{\begin{array}} 
\def\eary{\end{array}} 
\def\ben{\begin{enumerate}} 
\def\een{\end{enumerate}}
\def\bit{\begin{itemize}} 
\def\eit{\end{itemize}}

\theoremstyle{plain}

\newtheorem{thm}{Theorem}[section]
\newtheorem{lem}{Lemma}[thm]

\newtheorem{prop}{Proposition}[section]

\newtheorem*{conj*}{Conjecture}
\newtheorem{cor}{Corollary}[thm]
\newtheorem*{cor*}{Corollary}

\theoremstyle{definition}

\newtheorem{example}{Example}[section]

\newcommand{\GITl}[1]{\backslash \!\! \backslash _{\kern-.2em #1 \kern0.1em}}
\newcommand{\GIT}[1]{/\!\!/_{\kern-.2em #1 \kern0.1em}}

\newtheorem{theorem}{Theorem}[section]
\newtheorem{lemma}[theorem]{Lemma}
\newtheorem{lemma-definition}[theorem]{Lemma-Definition}
\newtheorem{proposition}[theorem]{Proposition}

\theoremstyle{definition}

\newtheorem{defn}[theorem]{Definition}

\theoremstyle{remark}
\newtheorem{remark}{Remark}[section]

\numberwithin{equation}{section}
\numberwithin{figure}{section}

\DeclareMathOperator{\Hw}{\mathcal{H}}
\DeclareMathOperator{\CC}{\mathbb{C}}
\DeclareMathOperator{\ZZ}{\mathbb{Z}}

\DeclareMathOperator{\RS}{\mathbb{P}^1}
\DeclareMathOperator{\Hol}{\mathscr{O}}
\DeclareMathOperator{\Rat}{\mathrm{Rat}}
\DeclareMathOperator{\PSL}{\mathbf{PSL}}
\DeclareMathOperator{\Div}{\mathrm{Div}}
\DeclareMathOperator{\Cr}{\mathrm{Cr}}
\newcommand{\ttop}{\mathrm{(top)}}
\newcommand{\out}{\mathrm{(out)}}
\newcommand{\inn}{\mathrm{(in)}}
\newcommand{\low}{\mathrm{(low)}}
\newcommand{\Sym}{\mathfrak{S}}
\newcommand{\Roots}{\mathfrak{R}}
\newcommand{\X}{\mathfrak{X}}
\newcommand{\Lie}{\mathscr{L}}

\usepackage{titlesec}
\titleformat{\section}
	{\centering\Large\scshape\bfseries}{\thesection.}{0.5em}{} 
\titleformat{\subsection}
 	{\normalfont\bfseries\scshape}{\thesubsection.}{.5em}{}
\titleformat*{\paragraph}{\bfseries}

  \newenvironment{acknowledgements}{%
  
  \begin{abstract}
}{%
  \end{abstract}
}
\numberwithin{equation}{section}

\makeatletter
\renewcommand{\l@section}{\@dottedtocline{1}{0em}{1.5em}}
\makeatother

\title{Landau--Ginzburg models of generalised Dubrovin--Zhang form and  pole collision: Dynkin-type A}
\author[A.~Proserpio]{Alessandro Proserpio}
\address[A.~Proserpio]{
School of Mathematics and Statistics,
University of Glasgow, Glasgow, G12 8QQ, United Kingdom.}
\email{a.proserpio.1 [at] research.gla.ac.uk}

\author[K.~van Gemst]{Karoline van Gemst}
\address[K.~van Gemst]{Dipartimento di Matematica e Applicazioni, Universit\`a di Milano-Bicocca, \newline
Via Roberto Cozzi 55, I-20125 Milano, Italy and INFN sezione di Milano-Bicocca}
\email{karoline.vangemst [at] unimib.it}

\begin{document}

\begin{abstract}
  In \cites{Bri20, BvG22}, the authors derive one-dimensional Landau--Ginzburg mirrors of Dubrovin--Zhang Frobenius manifolds constructed on regular orbit spaces of an extension of affine Weyl groups. We generalise the method employed, and classify the resulting Frobenius manifold structures in Dynkin type $A\,$. We interpret our results in terms of a stratification on the Hurwitz space boundary, and develop a pole-collision framework to compare the Frobenius structures within different strata. With this, we can prove  a structural result at the level of the prepotential, for arbitrary rank and dimension, as a suitable renormalised limit of the formulae in \cite{PS26}. As a corollary, a conjecture of \cite{MZ24} regarding the form of prepotentials related to doubly-extended affine Weyl groups is proven. 
\end{abstract}

\maketitle
\setcounter{tocdepth}{1}
\tableofcontents

\section{Introduction}\label{sec:intro}
Frobenius manifolds were introduced by Dubrovin as a geometric framework for
the WDVV (Witten--Dijkgraaf--Verlinde--Verlinde) equations arising in two-dimensional topological field theory
\cite{Dub96}. They encode this system of equations in terms of a
commutative, associative multiplication on the tangent bundle of a manifold, compatible with
a flat metric, a unit field and an Euler vector field. It turns out that this 
geometric structure appears in a range of mathematical and physical contexts,
from quantum cohomology and singularity theory to integrable systems and
supersymmetric field theories. In algebraic geometry, Frobenius manifolds arise in quantum cohomology, where
the Frobenius multiplication is the quantum product and the prepotential (the associated solution to the WDVV equations) is the generating function of the
genus-zero Gromov--Witten invariants. In singularity theory, they appear on the
base spaces of miniversal unfoldings of simple singularities  through Saito's theory of primitive forms \cite{Sai81}.
In the theory of integrable systems, Frobenius manifolds give rise to principal
hierarchies of hydrodynamic type and provide a natural geometric setting for
flat pencils of metrics, bi-Hamiltonian structures and dispersionless limits of
integrable hierarchies. In physics, they encode the genus-zero sector of
topological field theories and are closely related to mirror symmetry. 

A variety of generalisations of Frobenius manifolds have been developed, such as (flat, bi-flat) $F$-manifolds \cites{HM99,Man05,AL13}, almost Frobenius manifolds \cite{Dub04}, degenerate Frobenius manifolds \cite{Str99}, and generalised Frobenius manifolds (without flat unit) \cite{LQZ25}. Nevertheless, 
the semisimple, or massive, Frobenius manifold case remains the most extensively studied. The setting of semisimplicity allows one to use canonical coordinates, in which the tangent
algebra diagonalises, and is the natural setting for reconstruction
results (Givental--Teleman reconstruction \cites{Giv01, Tel12}). It is
moreover the setting for Landau--Ginzburg realisations of
Frobenius manifolds via Hurwitz spaces, where the canonical coordinates are
given by the critical values of a meromorphic function. More specifically, a Landau--Ginzburg model of a Frobenius manifold consists of a pair $(\lambda, \phi)$, where $\lambda$ is a meromorphic function called the \textit{Landau--Ginzburg superpotential}, and $\phi$ is the \textit{primary differential}. Such differentials were classified into five types in \cite{Dub96}.

The benefit of having a Landau--Ginzburg description of a Frobenius manifold is two-fold.
Firstly, it provides explicit superpotentials from which the Saito metric, 
multiplication, intersection form, as well as associated flat coordinates and Hamiltonian densities, may be computed in a straightforward way by residue
formulae. Secondly, it links Frobenius manifold theory with adjacent areas, such as mirror symmetry formulations (the Landau--Ginzburg model can be viewed as equivalent to the usual mirror symmetric B-model), as well as giving access to the Chekhov--Eynard--Orantin topological recursion procedure \cites{EO09, DOSS14, DNOPS19}.

A particularly important class of Frobenius manifolds is given by those constructed
on the regular orbit spaces of extended affine Weyl groups, denoted by $\widetilde{W}^{(\bar k)}(\Roots)$ for a crystallographic root system $\Roots\,$, as constructed by
Dubrovin and Zhang \cite{DZ98}. Their construction extends Dubrovin's earlier
work on orbit spaces of finite Coxeter groups \cite{Dub98}. For simply laced Dynkin types, these structures are related to the
quantum cohomology of orbifold projective lines \cite{Ros10}. A Landau--Ginzburg, or $B$-model, description of Dubrovin--Zhang Frobenius
manifolds was originally constructed in some cases \cites{DZ98,DSZZ19}, and was later obtained in full generality by Brini and  Brini--van Gemst from a degeneration of spectral curves of
the relativistic Toda chain \cites{Bri20,BvG22}. 

The appearance of relativistic Toda spectral curves in this construction has
two complementary sources.  On the one hand, it is motivated by
Seiberg--Witten theory: the odd periods of (almost dual of) polynomial Frobenius manifolds
associated with finite Coxeter groups are known to reproduce the quantum
periods of the Seiberg--Witten curves for four-dimensional
$\mathcal N=2$ pure Yang--Mills theory.  Nekrasov's formulation of the
corresponding five-dimensional $\mathcal N=1$ theory compactified on a circle
replaces the finite Weyl group geometry by its extended affine counterpart,
which suggested that Dubrovin--Zhang Frobenius manifolds should play the analogous
role in the five-dimensional theory. On the other hand, relativistic Toda spectral curves also arise from the
large-$N$ Chern--Simons/topological-string perspective.  Through the
Gopakumar--Ooguri--Vafa correspondence, Chern--Simons theory is related to
topological string theory; for spherical Seifert manifolds, the resulting
large-$N$ B-model curves are identified with Toda-type spectral curves of ADE
type.  Thus both the Seiberg--Witten and the Chern--Simons/topological-string
pictures point to the same relativistic Toda spectral curve geometry. 

The method to explicitly produce such Landau--Ginzburg models, as described in \cites{Bri20, BvG22}, consists of starting from
the characteristic equations of Lax operators, which may be classified into  Dynkin types $\Roots \in \{A_\ell, B_\ell, C_\ell, D_\ell, E_6, E_7, E_8, F_4, G_2\}$, and writing them as polynomials in fundamental characters (in a chosen representation), via the correspondence between elementary symmetric polynomials, Young diagrams and wedge products.  The spectral parameter (the LG-superpotential), and the parameter encoding the extension of the affine Weyl group, are then introduced via shifting a canonical fundamental character. 

The aim of the present paper is to generalise this method, and do an in-depth study in Dynkin type $A$.  The remaining Dynkin types are to be treated in a separate publication \cite{PvGprep}. More precisely,  instead of shifting only the
canonical fundamental character, as in the Dubrovin--Zhang case, we allow
shifts of an arbitrary number of fundamental characters. We then ask which such modified spectral curves still define
semisimple Frobenius manifolds, and study the resulting geometry.  Note that throughout this paper, unless stated otherwise, Einstein summation convention will be employed. 

\subsection{Main results}

The main results of this paper are contained in two parts. 

\vspace{1em}

In the first part (Section \ref{sec:Ashifts}), we generalise the method of \cites{Bri20, BvG22}.   This is done by 
modifying the relativistic Toda spectral curves of \cites{Bri20,BvG22} through
shifts of fundamental characters of the form
\begin{equation}
    w_i \longmapsto w_i + f_i \frac{\lambda}{w_0}, \quad \text{ for } i \in I \subset \{1, \cdots, \ell\},
    \label{eq:introshift}
\end{equation}
where $f_i$, for $i \in I$, are generic parameters, and $\ell$ is the rank of the associated complex simple Lie algebra. We classify the Frobenius manifolds of type $A$  by proving the following theorem. 
\begin{theorem}[= Theorem \ref{thm:AMain}]\label{thm:Amainintro}
    Let $\pi = (m_1,\ldots,m_{\ell(\pi)})$ be a partition of $p$ of length $\ell(\pi)$, and let $Q_{{\ell}}$ be the characteristic equation of the Lax-operator of type $A_\ell$ as it appears in \cite{BvG22} (and \eqref{eq:detA}). Then, all Frobenius manifolds obtained from performing a shift of the form \eqref{eq:introshift} on $Q_\ell$, are semisimple $\ell+1+\ell(\pi)$-dimensional Frobenius manifolds of charge $d=1$, with a Landau--Ginzburg model of form:
    \begin{equation}
    \lambda^{(I)}_{A_\ell}(\mu;\underline{w})
    =
    \frac{w_0\, Q_\ell(\mu;w_1, \cdots, w_\ell)}
    {\mu^{\bar k}
    \prod_{\alpha=1}^{\ell(\pi)}(\mu-w_{-\alpha})^{m_\alpha}}, \qquad \phi =\dd\log\mu\,,
    \label{eq:LGintro}
\end{equation}
where $I = \{\bar{k}, \bar{k}+1, \cdots, \bar{k}+p\}$ for some $\bar{k} \in \{1, \cdots, \ell-p\}$, and $w_{i}$, for $i = -p, \cdots, \ell$ are to be considered coordinates on the Frobenius manifolds.
\end{theorem}

 In particular, when ordering the index set $I=\{i_1, \cdots i_{|I|}\}$ such that $i_j < i_k$ if and only if $j<k$, $I$ must be of consecutive form (i.e. $i_{j} = i_{j-1}+1$ for all $j \in \{2, \cdots, |I|\}$) to obtain the structure of a Frobenius manifold. It can be shown that if this is not the case, then the Saito metric degenerates (Proposition \ref{prop:Anonconsec}). Moreover,  all possible Frobenius structures of type $A$ have LG-models of generalised DZ-form. The canonical case of \cites{Bri20, BvG22} is recovered by setting $I = \{\bar{k}\}$.  To obtain this form, the shift functions $f_i$, for $i \in I$, in \eqref{eq:introshift} are elementary symmetric polynomials of $w_{-p}, \cdots, w_{-1}$ (Remark \ref{rmk:Afepsilon}). 

The Landau--Ginzburg models of
\cites{PS26,Zuo20,MZ24} are recovered as special cases, and our construction
produces new, more general, solutions of the WDVV equations. As an application, we derive explicit closed-form expressions for Saito flat coordinates and prepotentials for several non-trivial examples. 

\vspace{1em}

In the second part (Section \ref{sec:polecollision}), we interpret the partition data as a stratification of the Hurwitz
    boundary obtained by coalescence of movable poles. Each boundary stratum is known to carry a Frobenius manifold given by a choice of a primary differential compatible with all the superpotentials of the form in \cref{thm:Amainintro}. We relate the Frobenius
    structures on different strata via an explicit pole-collision formalism. The open stratum corresponds to the case in which all
additional movable poles are simple; lower dimensional strata are obtained by
allowing these poles to coalesce, and are therefore indexed by partitions obtained by coarsening the longest one. The simple poles that collapse to produce a higher order pole will be said to be part of a \textit{collision cluster} (or block). We
develop a pole-collision formalism that relates the residue metric and the
Frobenius three-tensor on different strata, giving an effective procedure for
passing between Frobenius structures on a stratum and its boundary.

 Using this framework, we introduce a natural `clusterisation' of flat coordinates in each collision block, and prove a general structural form for the prepotentials of Frobenius structures associated with the relativistic Toda chain, in type $A$. That is,
 \begin{prop}\label{prop:flatcoordinatesintro}
       Let $\hat{\ell},\bar{k}\in\ZZ_{\geq 0}$, and $\underline{t}$ be flat coordinates for the $\widetilde{W}^{(\bar{k})}(A_{\hat{\ell}+\bar{k}})$ Dubrovin--Zhang manifold. Fix $p\in\ZZ_{\geq 0} $ and let $\pi=(m_1,\dots, m_{\ell(\pi)})$ be a partition of $p\,$, and consider the associated Landau--Ginzburg model of the form \eqref{eq:LGintro} for $\ell:=\hat{\ell}+\bar{k}+p\,$. 
    One can find coordinates of the form $(\underline{t},\underline{\beta},\underline{c})$ -- called collision coordinates -- on the associated Frobenius manifold such that:
          \begin{enumerate}
              \item $\underline{\beta}\in\CC^{\ell(\pi)}$ are flat coordinates logarithmic in the movable poles;
              \item $\underline{c}=\bigl(\underline{c}_1,\dots, \underline{c}_{\ell(\pi)}\bigr)\in\CC^{m_1}\times \dots\times\CC^{m_{\ell(\pi)}}$ are (in general curved) natural collision quantities that are well-behaved in the limit (\cref{eq:collisionchartexp,prop:polecoalescence}). In each collision cluster, these are related to flat coordinates in a controllable, block-diagonal, universal way that only depends on the size (\cref{lem:rootofcmflat}).
          \end{enumerate}
 \end{prop}
\begin{theorem}[= Theorem \ref{thm:prepotential}]
Let $\hat{\ell},\bar{k},p,\pi$ and $\underline{t}$ be as in \cref{prop:flatcoordinatesintro}, and denote by $F^{\mathrm{DZ}(\hat{\ell},\bar{k})}(\underline{t})\in \mathcal{A}^{(\bar{k})}_{\hat{\ell}+\bar{k}}$ the corresponding Dubrovin--Zhang prepotential (\cref{thm:DZ:reconstruction}). The prepotential associated to the Landau--Ginzburg model of the form in \cref{eq:LGintro} can be written as:
 \begin{equation*}
        \begin{aligned}
F_\pi(\underline{t},\underline{\beta},\underline{c})&=F^{\mathrm{DZ}(\hat{\ell},\bar{k})}(\underline{t})+\sum_{\alpha=1}^{\ell(\pi)}F^{(m_\alpha)}_{\mathrm{ren}}(\beta_\alpha,\underline{c}_\alpha)+\\&\quad +\sum_{\alpha=1}^{\ell(\pi)}\sum_{i=1}^{m_\alpha}\Lambda_i(\underline{t},\beta_\alpha)\,c_{i\alpha}+\sum_{1\leq \alpha<\gamma\leq \ell(\pi)}F_{\mathrm{int}}^{\alpha,\gamma}(\beta_\alpha,\beta_\gamma,\underline{c}_\alpha,\underline{c}_\gamma)\,,
\end{aligned}
    \end{equation*}
  where $F^{(m)}_{\mathrm{ren}}$ is a renormalised limit (\cref{prop:prepotentiallimit,prop:renormalisedlimitexplicit}) of the logarithmic sector of the top-stratum prepotential for the collapse of $m$ simple poles, $\Lambda_i(\underline{t},\nu)\in \mathcal{A}^{(\bar{k})}_{\hat{\ell}+\bar{k}}[\nu,e^{\pm \nu}]$ and the pairwise interaction term between two collision clusters is of the form:
  \[
  \begin{aligned}
      F_{\mathrm{int}}^{\alpha,\gamma}(\beta_\alpha,\beta_\gamma,\underline{c}_\alpha,\underline{c}_\gamma)=c_{1,\alpha}c_{1,\gamma}\log(e^{\beta_\alpha}-e^{\beta_\gamma})+\sum_{k= 1}^{m_\alpha+m_\gamma-2}\frac{1}{(e^{\beta_\alpha}-e^{\beta_\gamma})^k}P^{\alpha,\gamma}_k(e^{\beta_\alpha},e^{\beta_\gamma},\underline{c}_\alpha,\underline{c}_\gamma)\,,
  \end{aligned}
  \]
 for some families of polynomials $\bigl\{P^{\alpha,\gamma}_k\bigr\}\,$.     In particular:
    \begin{enumerate}
        \item $F_\pi$ depends polynomially on $\{c_{1,\alpha},\dots, c_{m_\alpha-1,\alpha}\}_{1\leq \alpha\leq \ell(\pi)}\,$ and it is a Laurent polynomial of $c_{m_\alpha,\alpha}$ for any pole of order at least three ($m_\alpha\geq 3$), up to a logarithmic term of the form $c_{1,\alpha}^2\log c_{m_\alpha,\alpha}$ in each collision cluster.
        \item It is invariant under the action of the Young subgroup $\Sym_\pi$ of the partition $\pi$ acting by permutation of collision blocks giving poles of the same order.
    \end{enumerate}
\end{theorem}
In fact, the coordinate transformation in \cref{lem:rootofcmflat}, ensures that property (1) will still hold for $F_\pi\,$, when written as a function of flat coordinates. As a corollary, we confirm the conjecture of \cite{MZ24}, for prepotentials related to orbit spaces of doubly-extended affine Weyl groups, which corresponds to the case $\pi=(p)\,$. 

\subsection{Organisation of the paper}
The paper is organised as follows. In Section \ref{Section:DZ} we provide the necessary generalities on Frobenius manifolds and spectral curves. In Section \ref{sec:Ashifts}, we introduce the main method employed for producing  spectral curves and prove a classification theorem of the possible Frobenius manifold structures obtained.  Section \ref{sec:Ex} consists of several non-trivial examples, for which we derive explicit closed-form expressions for flat frames and prepotentials. In Section \ref{sec:polecollision},  we perform an in-depth analysis of the boundary of the Hurwitz spaces associated to the Landau--Ginzburg models derived in Section \ref{sec:Ashifts} and, using this, we compare and relate the different Frobenius manifold structures. We also prove a general form of the associated prepotentials. Finally, we conclude the paper in Section \ref{sec:conc} and discuss some potential future directions.

\section{Background}\label{Section:DZ}
We start by recalling the basic definitions from the theory of Frobenius manifolds. The reader is referred to \cite{Dub96} for further details.

\begin{defn}
A (complex, holomorphic) Frobenius manifold of charge $d\in\CC$ is a 5-tuple $\mathcal{M}=(M,  \cdot, \eta, e, E)$, where $M$ is a finite dimensional complex manifold equipped with the following geometric data:
\begin{itemize}
    \item an $\Hol_M$-bilinear multiplication of holomorphic vector fields $\cdot:\X_M\otimes_{\Hol_M}\X_M\to \X_M\,$, where $\X_M$ denotes the sheaf of holomorphic vector fields on $M\,$;
    \item a symmetric, non-degenerate bilinear form $\eta:\X_M\otimes_{\Hol_M}\X_M\to\Hol_M\,$;
    \item two holomorphic vector fields $e,E\in\X_M(M)\,$, respectively the identity and Euler vector field;
\end{itemize}
such that:
\begin{enumerate}
    \item [(FM1)] At each point $p \in M$, the fibre $T_pM$ of the holomorphic tangent bundle at $p$ has the structure of a unital associative commutative Frobenius algebra with multiplication $\cdot_p\,$, pairing $\eta_p$ and identity element $e_p\,$;
    \item [(FM2)] The bilinear form $\eta$ is flat;
    \item [(FM3)] The unit vector field is covariantly constant,
    $\nabla e = 0\,$,
with respect to the Levi-Civita connection $\nabla$ associated to $\eta\,$;
\item [(FM4)] If $c$ denotes the $(0,3)$-tensor $c(X,Y,Z):=\eta(X\cdot Y,Z)$ for any three holomorphic vector fields $X, Y $ and $Z\,$, then $c$ and $\nabla c$ are totally symmetric;
\item [(FM5)] The Euler vector field satisfies:
\[
\begin{aligned}
    \Lie_E\cdot&=\cdot\,,&&& \Lie_E\eta&=(2-d)\eta\,,
\end{aligned}
\]
where $\Lie $ denotes the Lie derivative.
\end{enumerate}
\label{def:frob}
\end{defn}

\smallskip

A Frobenius manifold is {\it semisimple} if the set 
\[
\mathrm{Discr}(M) \coloneqq \{p\in M~|~\exists v \neq 0\in T_p M~{\rm with}~ v\cdot v=0\,\}
\] 
has positive complex codimension. Whenever $E$ is in the group of units of $(T_p M, \cdot)$,  one may define a second flat metric, $g \in \Gamma(M, \mathrm{Sym}^2 T^*M)$, by \cite{Dub04}
\begin{equation}
    g(E \cdot X, Y) = \eta(X, Y). 
    \label{eq:gdef}
\end{equation}
This second flat metric is usually called the \textit{intersection form}.

A key consequence of Definition \ref{def:frob} is the existence of a one-parameter affine family of flat metrics on $T^*M$ given by the pencil
\begin{equation}
    g^* + \lambda \eta^*,
\end{equation}
where $g^*, \eta^*$ denote the dual cotangent metrics and $\lambda \in \mathbb{C}$ \cite{Dub98b}.

The property (FM4) can be rephrased by saying that the tensor field $c$ is a Codazzi tensor of rank three \cite{Fer81}. Since the metric is flat, it follows that there locally exists a holomorphic function $F$ in the neighbourhood $U$ of each point of $M$ such that:
\begin{equation}\label{eq:ctensorF}
    c\,\lvert_U=\nabla ^3F\,.
    \end{equation}
The function $F$ is called \textit{prepotential} (or free energy) of the Frobenius manifold \cite{Dub96}, and it is determined up to a local holomorphic function in the kernel of $\nabla ^3\,$. Notice that, for any function $F\,$, defining $c$ as in \cref{eq:ctensorF} automatically ensures that it is a Codazzi tensor, and that the multiplication is commutative. Associativity, on the other hand, is a non-trivial constraint on $F\,$, as we shall discuss in more detail below.

Notice that the metric itself is a Codazzi tensor of rank two. Therefore, by the same argument, it can be locally written as the second covariant derivative of a holomorphic function. Owing to the property (FM3), it is easy to see that:
\begin{equation}
    \eta\lvert_U=\nabla^2\Lie_eF\,,
\end{equation}
in the same neighbourhood $U\,$. Therefore, the prepotential simultaneously determines the multiplication and the metric on the Frobenius manifold.

This description becomes more evident after observing that, since $\eta$ is a flat metric, there exists a distinguished local frame around each point. Namely, these are the coordinate vector fields $\bigl\{\pdv{t_i}\bigr\}_{i=1}^{\dim M}$ of a system of \textit{flat coordinates} $\{t^i\}_{i=1, \cdots, \text{dim}(M)}$ for $\eta\,$. In other words, $\nabla\pdv{t_i}=0$ for any $i\,$.
Throughout this paper, we let $\partial_{x^i}$ be short-hand for $\frac{\partial}{\partial x^i}$, and $\partial_i$ be short-hand for $\frac{\partial}{\partial t^i}$. Clearly, affine functions of the coordinates of any given flat system are again flat, so the coordinate frame $\{\partial_i\}$ will be defined up to such transformations. Since $e$ is itself flat it can, without loss of generality, be chosen to be a coordinate vector field in our system. A standard choice is $e=\partial_1\,$. However, for different applications, it might be more natural to choose $e$ to be a different flat coordinate vector field. Of course, this is just a notational issue.

In $t$-coordinates, the components of the $c$-tensor are simply given by third derivatives of the prepotential $F\,$, as the Christoffel symbols vanish by construction. The description of the structure, therefore, drastically simplifies to: 
\ben[(i)]
    \item $e=\partial_{1}$;
    \item $\eta_{ij} \equiv \eta(\partial_i, \partial_j) = \partial_{1ij}^3 F\,$, which are the entries of a constant, non-degenerate matrix;
    \item $c_{ijk} \equiv c(\partial_i, \partial_j, \partial_k) = \eta(\partial_i \cdot \partial_j, \partial_k) = \partial_{ijk}^3 F$;
    \item $E = \sum_{i} d_i t^i \partial_i + \sum_{i} r_i \partial_i $;
    \item $g_{ij} = \sum_k E^k c_{kij} $;
    \item $\partial_i \cdot \partial_j =  c^k_{ij} \partial_k$, where $c^k_{ij} \coloneqq  \eta^{k m}c_{mij}$, and $\eta^{ij} \coloneqq (\eta)^{-1}_{ij}$.
    \item As a consequence of the multiplication being associative, $F$ satisfies the Witten--Dijkgraaf--Verlinde--Verlinde (WDVV) equations,
\begin{equation}
\partial_{ijk}^3F  \,  \eta^{kl}\,\partial_{lmn}^3 F= j \longleftrightarrow m\,.
\end{equation}
\een 

The form of the Euler vector field follows from the following
\begin{prop}[\cite{Dub99}]
    The vector field $\nabla E$ is covariantly constant with respect to $\nabla$, i.e. $\nabla^2 E=0\,$.
\end{prop}

\subsection{Hurwitz Frobenius manifolds}
\label{sec:subHur}
A Hurwitz space is a moduli space parametrising ramified covers of the Riemann sphere. A point in a Hurwitz space is an equivalence class $[\lambda: C_g \to \mathbb{P}^1]$, where $C_g$ is a genus-$g$ smooth  complex projective curve (or compact Riemann surface) and $\lambda$ is a morphism to the complex projective line realising $C_g$ as a branched cover of $\mathbb{P}^1$. The equivalence relation is given by automorphisms of the cover. \\
We consider Hurwitz spaces with fixed ramification over infinity. Let the preimage of $\infty$ consist of $q+1$ distinct points, denoted by $\infty_i \in C_g$ for $i=0, \cdots, q$, with $\lambda$ having degree $n_i + 1$ near $\infty_i$. The corresponding Hurwitz space will be denoted $\Hw_{g;\underline{n}}$, where $\underline{n} \coloneqq (n_0, \cdots, n_q)$ encodes the ramification over $\infty\,$. 
This is a connected complex manifold (or indeed an irreducible quasi-projective complex algebraic variety \cite{Ful69}) of dimension
 \[
 d_{g,\underline{n}}:=\dim(\Hw_{g;\underline{n}}) = 2g + 2q + \sum_{i=0}^q n_i\,.\] 
 We shall build Frobenius manifold structures on (submanifolds of) $\mathcal{H}_{g;\underline{n}}$. The construction is laid out in detail in \cite{BvG22}, here we give a brief account.  By Riemann's existence theorem, the branch points $u_i$ of $\lambda$ may serve as local coordinates away from the closed subsets on which they pairwise coincide, i.e. $u_i = u_j$ for $i \neq j\,$. On the complement of the union of such subsets, we define a family of semisimple, commutative, associative, and unital $\mathbb{C}$-algebras on the tangent bundle defined fibre-wise by
 \begin{equation}
    \partial_{u_i} \cdot \partial_{u_j} = \delta_{ij}\,\partial_{u_i}.
\label{eq:prodss}
\end{equation}
This defines a semisimple multiplication on the tangent bundle of $\Hw_{g;\underline{n}}\,$, for which the coordinate vector fields in the directions of the branch points are idempotents. For this reason, the critical values $u_i$ will be called \textit{canonical coordinates} on the Frobenius manifolds. The identity and Euler vector field are given by:
\begin{equation}
 e =    \sum_{i=1}^{ d_{g,\underline{n}}}\partial_{u_i},
\quad     E = \sum_{i=1}^{ d_{g,\underline{n}}}u_i \partial_{u_i}.
\label{eq:eEhur}
\end{equation}
From the Hurwitz perspective, these are the generators of the affine action induced by the stabiliser of the marked point $\infty$ on the target $\RS\,$, i.e.: 
\begin{equation}
    (C_g,\lambda) \mapsto (C_g, a\lambda+b), \qquad u_i \mapsto au_i + b,
    \label{Eq:Hurwitzaffine}
\end{equation}
for $a\in \mathbb{C}^\ast$ and $b\in\CC\,$.

 The remaining ingredients of the Frobenius structure are constructed from a choice of an \textit{admissible primary differential} $\phi$ on the curve $C_g\,$. Such Abelian differentials have been classified into five families in \cite{Dub96}, and they generalise the construction of Saito primitive forms on the parameter space for miniversal unfoldings of surface singularities \cites{Sai81,Sai83}. For our purposes, it will be sufficient to restrict ourselves to an exact differential of the third kind of the form $\phi=\dd\log\mu$. This differential will be lambda-admissible by the analogous argument as that of \cite{BvG22}. Let us denote by $\{q_1,\dots, q_{d_{g;\underline{n}}}\}\equiv \Cr_\lambda$ the set of critical points of $\lambda\,$, i.e. $\lambda(q_i)=u_i\,$. The metric $\eta\,$ is then defined as follows in the frame induced by the canonical coordinates:
 \begin{equation}    \label{eq:etares}
     \eta_{ij}\equiv\eta\bigl(\partial_{u_i}\,,\,\partial_{u_j}\bigr):=-\delta_{ij}\Res_{q_i}\tfrac{\phi^2}{\dd\lambda}\,.
 \end{equation}
 As a consequence, the $c$-tensor and the intersection form $g$ are similarly given by:
 \begin{subequations}
     \begin{equation}    \label{eq:cres}
         c_{ijk}=-\delta_{ij}\delta_{ik}\Res_{q_i}\tfrac{\phi^2}{\dd\lambda}\,,
     \end{equation}
     \begin{equation}    \label{eq:gres}
         g_{ij}=-\delta_{ij}\,\Res_{q_i}\tfrac{\phi^2}{\lambda\dd{\lambda}}\,.
     \end{equation}
 \end{subequations}

\subsection{Frobenius manifolds from extended affine Weyl groups}
Dubrovin--Zhang (DZ) manifolds are semisimple Frobenius manifolds defined on the space of regular orbits of extended affine Weyl groups \cite{DZ98}. In the following, we recall their construction  following \cite{BvG22} closely. \\

Let $\mathfrak{g}_\Roots $ be a rank-$\ell_\Roots $ complex simple Lie algebra associated to a root system $\Roots $, $\mathfrak{h}_\Roots $ the associated Cartan subalgebra, $\dim \mathfrak{h}_\Roots =\ell_{\Roots }$, and  $\mathcal{W}_\Roots $ the Weyl group. The construction of DZ-manifolds relies on a canonical choice of Dynkin node,  which we denote by $\bar k \in \{1, \dots, \ell_\Roots \}$, and we let $\alpha_{\bar k}$ and $\omega_{\bar k}$ be the corresponding simple root and fundamental weight, respectively.  The canonical node $\bar{k}$ is an ``attaching" vertex in the diagram, that is, the one which if removed splits the Dynkin diagram into disconnected $A$-type pieces.  The canonical Dynkin node is marked by a $\cross$ on the associated affine Dynkin diagram  in  \cref{fig:Dynkin}. For $\Roots  \neq A_\ell$,  the representation $\rho_{\bar k}$ is the fundamental representation of $\mathfrak{g}_\Roots $ of highest dimension, whereas for $\Roots  = A_\ell$, any (non-affine) node may be marked, and we obtain a unique Frobenius manifold structure for any marking up to the symmetry $\ell+1-k \leftrightarrow k$ by reflection around the centre of the diagram. \\

The action of $\mathcal{W}_{\Roots }$ on $\mathfrak{h}_\Roots $ may be lifted to an action of the affine Weyl group $\widehat{\mathcal{W}}_{\Roots } \cong \mathcal{W}_{\Roots } \ltimes  \Lambda^\vee_r(\Roots )$, with  $\Lambda_r^\vee(\Roots )$ being the lattice of co-roots:
 \begin{align}
    \widehat{\mathcal{W}}_{\Roots } \times \mathfrak{h}_\Roots & \mapsto \mathfrak{h}_\Roots,\\
    ((w, \alpha^\vee), h) & \mapsto w(h) + \alpha^\vee.
\end{align}
Then the \textit{extended} affine Weyl group $\widetilde{\mathcal{W}}_{\Roots }$ is defined as the semi-direct product $\widetilde{\mathcal{W}}_{\Roots } \coloneqq \widehat{\mathcal{W}}_{\Roots } \ltimes \mathbb{Z}$ acting on $\mathfrak{h}_\Roots  \oplus \mathbb{C}$ by 
\begin{align}
    \widetilde{\mathcal{W}}_{\Roots } \times {\mathfrak{h}_{\Roots }} \oplus \mathbb{C} & \rightarrow \mathfrak{h}_\Roots  \oplus \mathbb{C},\\
    ((w, \alpha^\vee, {n}), (h, v)) & \mapsto (w(h) + \alpha^\vee + n\omega_{\bar k}, {v} - {n}){.} 
    \label{eq:extweyl}
\end{align}

Let $\Sigma_\Roots $ denote the hyperplane arrangement associated to the root system $\Roots $, and $\mathfrak{h}^{\text{reg}}_\Roots  \coloneqq \mathfrak{h}_\Roots  \backslash  \Sigma_\Roots  $ be the set of regular elements in $\mathfrak{h}_\Roots $. The restriction of \eqref{eq:extweyl} to $\mathfrak{h}^{\rm reg}_\Roots  \oplus \mathbb{C}$ is then a free affine action, whose quotient defines the regular orbit space of the extended affine Weyl group of $\Roots $ with marked node $\bar k$ as
\begin{equation}
M^{\rm DZ}_{\Roots } \coloneqq (\mathfrak{h}_\Roots ^{\text{reg}} \times \mathbb{C})/ \widetilde{\mathcal{W}}_{\Roots }  \cong \mathcal{T}_\Roots ^{\text{reg}}/\mathcal{W}_{\Roots } \times \mathbb{C}^*, 
   \label{Def:DZmaniLG}
\end{equation}
where $\mathcal{T}^{\text{reg}}_\Roots  = \text{exp}( \mathfrak{h}_\Roots ^{\text{reg}})$ is the image of the set of regular elements of $\mathfrak{h}_\Roots ^{\rm reg}$ under the exponential map to the maximal torus $\mathcal{T}_\Roots $.

$M^{\rm DZ}_{\Roots }$ is a smooth complex manifold homeomorphic to a Zariski open subset of the affine GIT quotient $\mathrm{Spec}~ \mathcal{I}_{\Roots }^{\widetilde{\mathcal{W}}_{\Roots }}$. The linear coordinates $(x_1, \dots, x_{\ell_\Roots }; x_{\ell_\Roots +1})$ on $\mathfrak{h}_\Roots ^{\text{reg}} \oplus \mathbb{C}$, defined w.r.t. the co-root basis with $x_{\ell_{\Roots }+1}$ parametrising linearly the right summand, may serve as local coordinates. In fact, they provide a flat frame for the intersection form \eqref{eq:gdef}. By orthogonal extension of minus the Cartan--Killing form on $\mathfrak{h}_\Roots $, we define a non-degenerate pairing $\xi$ on  $\mathfrak{h}_\Roots  \times \mathbb{C}$ by
\begin{equation}
      \xi(\partial_{x_i}, \partial_{x_j}) \coloneqq
  \begin{cases}
                                   -(\mathcal{K}_\Roots )_{ij} & \text{if } i, j< \ell_{\Roots }+1, \\
                                   d_{\bar k} & \text{if } i=j=\ell_{\Roots }+1, \\
                                   0 & \text{otherwise.}
  \end{cases}
\end{equation}
 Here, $d_i \coloneqq \left\langle \omega_i, \omega_{\bar k}\right\rangle$, with $\left\langle\alpha, \beta\right\rangle$ being the pairing on $\mathfrak{h}_\Roots ^*$ induced by the restriction of the Killing form on the Cartan subalgebra.  The quotient map $\aleph: \mathcal{T}_\Roots ^{\text{reg}} \times \mathbb{C}^* \rightarrow M_\Roots ^{\rm DZ}$ from \eqref{Def:DZmaniLG} defines a principal {$\mathcal{W}_{\Roots }$}-bundle on $M_\Roots ^{\rm DZ}$: a section $\tilde{\sigma}_i$ lifts a (sufficiently small) open $U \subset M_\Roots ^{\rm DZ}$ to the $i^\text{th}$ sheet of the cover $V_i \in \widetilde{\sigma}_i^{-1}(U) \equiv V_1 \sqcup \cdots \sqcup V_{|\mathcal{W}_{\Roots }|} $.  The following reconstruction theorem holds.

\begin{thm}[Theorem~2.1 in \cite{DZ98}]\label{thm:DZ:reconstruction}
There exists a unique (up to isomorphism) semisimple Frobenius structure $\mathcal{M}_\Roots ^{\rm DZ}=(M_\Roots ^{\rm DZ}, e, E, \eta, \cdot)$  satisfying the following properties: 
\begin{enumerate}
    \item $e = \partial_{{\bar k}}$;
    \item $E = \tfrac{1}{d_{\bar k}} \partial_{x_{\ell_\Roots +1}} =\tfrac{1}{d_{\bar k}}\bigl( \sum_{j=1}^{\ell_{\Roots }} d_j\,t^j \partial_{j} +  \partial_{\ell_{\Roots }+1}\bigr)\,$;
    \item $g = \tilde{\sigma}_i^* \xi$;
    \item $F \in \mathbb{C}[t^1, \cdots, t^{\ell_{\Roots }+1}][e^{{t^{\ell_{\Roots }+1}}}]=:\mathcal{A}_{\Roots}^{(\bar k)}\,$.
\end{enumerate}
\label{Thm:DZ}
\end{thm}

Such Frobenius manifolds will always have charge one, or equivalently, the prepotential will be a degree-2 quasi-homogeneous function of its arguments.

\subsection{Canonical spectral curves from the relativistic Toda chain}
\label{sec:Todabackground}
 
In \cite{BvG22} LG-models for DZ-manifolds are derived from the characteristic equation of Lax operators associated to the relativistic Toda chain. These are fully integrable systems classified by (affine) Dynkin diagrams $\Roots   \in \{A_\ell, B_\ell, C_\ell, D_\ell, E_6, E_7, E_8, F_4, G_2\}$, as shown in Table \ref{fig:Dynkin}.
\begin{figure}[h!]
    \centering
    \includegraphics[width=0.7\linewidth]{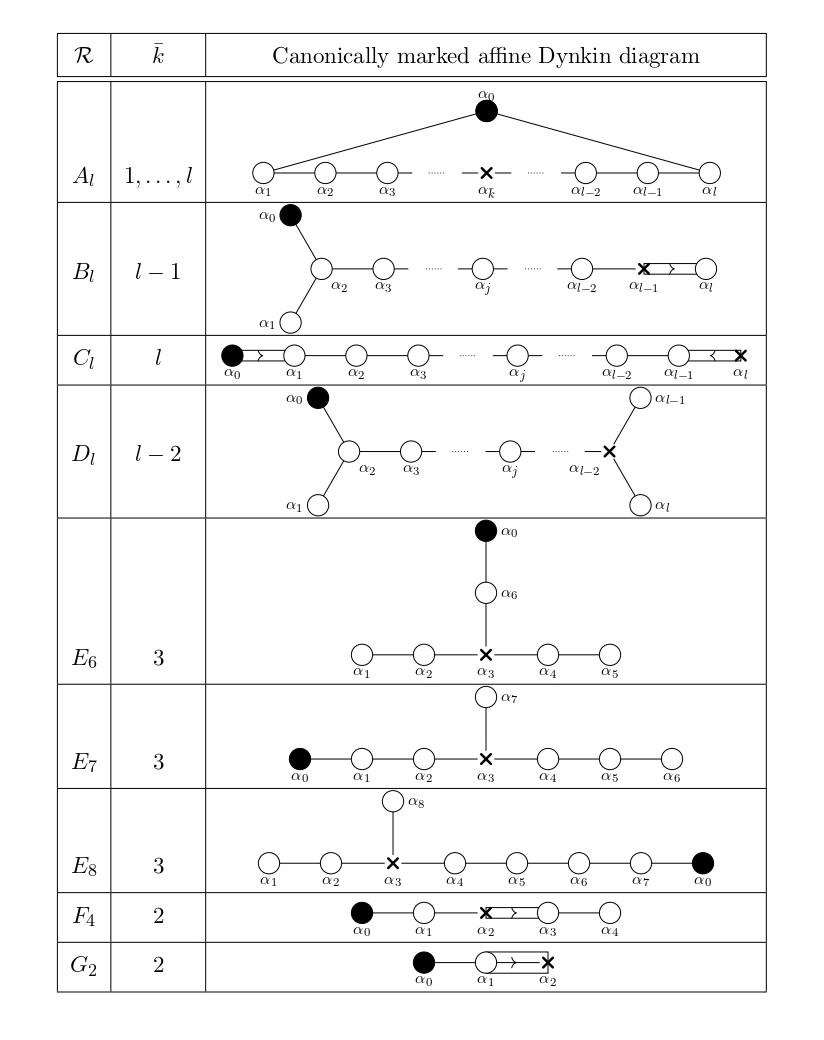}
    \caption{Affine Dynkin diagrams with canonical markings, as in \cite{DZ98}, Table~1. The node corresponding to the affine root is marked in black, and the canonical marked node is indicated with a $\times$. }
    \label{fig:Dynkin}
\end{figure}

In the following, we briefly recall the method employed in \cite{BvG22}. Let $\rho_{\omega}$ be a non-trivial irreducible representation with highest weight $\omega$. Then for $g \in \mathcal{G}$, where $\mathcal{G}$ is simply the connected complex simple Lie group associated to the complex simple Lie algebra with root system of type $\Roots $.  In the representation $\rho_{\omega}$ we have
\begin{equation}
\label{eq:Qnonred}
    \mathcal{Q}_{\omega} = \text{det}(g-\mathfrak{1} \cdot\mu) = \sum_{k=0}^{\text{dim}(\rho_{\omega})} (-\mu)^{\text{dim}(\rho_{\omega})-k}\chi_{\wedge^k\rho_{\omega}}(g),
\end{equation}
where $\chi_{\wedge^k\rho_{\omega}}$ denotes the exterior characters of the representation $\rho_{\omega}$ and the second equality is obtained from the cofactor expansion of the determinant via correspondence between elementary symmetric polynomials, Young diagrams and wedge products \cite{BvG22}. 

The representation ring of a simple Lie group is an integral polynomial ring generated by the fundamental characters, which implies that
\begin{equation*}
   \chi_{\wedge^k \rho_\omega}(g) \in \mathbb{Z}[\chi_1, \cdots, \chi_{\ell_{\Roots }}], 
\end{equation*}
where the $i^{\text{th}}$ fundamental character is given by the trace of $g$ in the $i^{\text{th}}$ fundamental representation: $\chi_i(g) = \text{Tr}{\rho_i}(g)$ and $\ell_{\Roots }$ denotes the rank of $\Roots $.  For each of the cases $\Roots  \in \{A_\ell, B_\ell, C_\ell, D_\ell, E_6, E_7, E_8, F_4, G_2\}$ $\rho_{\omega}$ is to be chosen as a non-trivial irreducible representation of smallest dimension. Via isomorphisms of flows on Prym--Tuyrin varieties, the choice of representation strongly affects the resulting superpotential but not the underlying Frobenius manifold \cite{Bri20}. Since $\rho_{\omega}$ is quasi-minuscule, \eqref{eq:Qnonred} can be written as
\begin{equation*}
    Q_{\omega} = (1-\mu)^{z_0}\prod_{0 \neq \tilde{\omega} \in \Gamma(\rho_{\omega})}^{}(e^{\tilde{\omega} \cdot h}-\mu),
\end{equation*}
where $z_0$ is the dimension of the zero weight space of $\rho_{\omega}$, $\Gamma(\rho)$ denotes the weight space of $\rho$ and $h$ is a choice of Cartan torus element conjugate to $g$; $[e^h] = [g]$.  Consider the reduced part of $Q_{\omega}$:
\begin{equation*}
    Q_{\omega}^{\text{red}} \coloneqq \prod_{0 \neq \tilde{\omega} \in \Gamma(\rho_{\omega})}^{}(e^{\tilde{\omega} \cdot h}-\mu).
\end{equation*}
Note that for  $\Roots  \neq B_\ell, E_8, F_4, G_2$ the representation is minuscule which means that $z_0 = 0$ and $Q_{\omega} = Q_{\omega}^{\text{red}}$. Define
\begin{equation}
\label{eq:shiftcanpoly}
    \mathcal{P}_{\omega}(w_{0}, \cdots, w_{\ell_{\Roots }};\lambda;\mu) \coloneqq Q_{\omega}^{\text{red}}\left(\chi_i = w_i - \delta_{i\bar{k}} \frac{\lambda}{w_0}\right),
\end{equation}
where $\bar{k}$ is the \textit{attaching} node of the Dynkin diagram (i.e. the one which if removed splits the Dynkin diagram into disconnected type $A$ pieces) as indicated in Figure \ref{fig:Dynkin}. As $\underline{w} \equiv (w_0, w_1, \cdots, w_{\ell_{\Roots }}) \in \mathbb{C}^* \times \mathbb{C}$  varies we obtain an $\ell_{\Roots }+1$-parameter family of algebraic planar curves in Spec$\mathbb{C}[\lambda, \mu]$ with fibre at $\underline{w}$ given by the vanishing locus of $\mathcal{P}_{\omega}$. After taking the normalisation of the closure of the fibres over $\underline{w}$ in $\mathbb{P}^2$, the $\lambda$-projection $(\lambda, \mu) \mapsto \lambda \in \mathbb{P}^1$, as $\underline{w}$ varies, defines a subvariety of the Hurwitz space $\mathcal{H}_{g;\underline{n}}$ where $g = h^{1,0}( \overline{\mathbb{V}(\mathcal{P}_{\omega})})$ and the ramification at $\infty$ of the $\lambda$-projection is recorded by $\underline{n}$. Note that we obtain, in general, a subvariety of the Hurwitz space unless $\text{dim}(\mathcal{H}_{g,\underline{n}}) = \ell_{\Roots }+1$, which is only the case for $\mathcal{R} = A_\ell$. 

\begin{theorem}[Theorem~3.6 in \cite{BvG22}]
For any simple Dynkin type $\Roots $ there exists a highest weight $\omega$ for the corresponding simple Lie algebra $\mathfrak{g}$, pairs of integers $(g_{\omega}, \mathsf{n}_{\omega})$, and an explicit embedding $\iota_{\omega} : \mathcal{M}^{\text{DZ}}_{\Roots } \hookrightarrow \mathcal{H}^{[\mu]}_{g_{\omega},\mathsf{n}_{\omega}}$, such that $\iota_\omega$ is a Frobenius manifold isomorphism onto its image $\mathcal{M}_\omega^{\text{LG}} \coloneqq \iota_\omega(\mathcal{M}^{\text{DZ}}_\Roots )$.
\label{thm:canonicalmirror}
\end{theorem}

In the following section, we generalise this method by modifying the shift in \eqref{eq:shiftcanpoly}, and specialise to $\Roots = A_{\ell}$.

\section{Spectral curves in type A}
In this work, we shall focus on the case $\Roots=A_\ell\,$, postponing the analysis of the other Dynkin types to a separate publication. Let $Q_{\ell} \coloneqq Q_{\omega}^{\text{red}}$, for $\omega$ being the highest weight of the defining $\ell+1$-dimensional representation of $A_\ell$ (the minimal representation). In this case we have,
    \begin{equation}
       Q_{\ell} = 1 + (-1)^{\ell+1}\mu^{\ell+1} + \sum_{i=1}^{\ell}(-1)^i w_i \mu^i.
        \label{eq:detA}
    \end{equation} 
For any Dynkin type, the parameters $w_i$ are fundamental characters, and they relate to the linear coordinates on the root space by
 \begin{equation}
    w_k =  \sum_{0 \neq \widehat{\omega} \in \Gamma(\rho_{\omega})}^{} \mathfrak{m}_{\widehat{\omega}} e^{(\widehat{\omega}, x)}
          \label{eq:wkinx}
\end{equation}
where $\widehat{\omega}$ appears in the weight system, $\Gamma(\rho_{\omega})$, of $\rho_{\omega}$ with multiplicity $\mathfrak{m}_{\widehat{\omega}}$, with $\rho_{\omega}$ being the minimal representation.

We analyse spectral curves arising from a class more general than that of \eqref{eq:shiftcanpoly}.  Specifically, we let
\begin{equation}
    \chi_i \equiv w_i \mapsto w_i + f_i \dfrac{\lambda}{w_0}, \quad \text{for } i \in I  \subset \{1, \cdots, \ell\},
    \label{eq:shiftgen}
\end{equation}
where $f_i$ is independent of $\mu, \lambda, w_0, \cdots, w_{\ell}$. In other words, we keep the same $Q_{\ell}$ as in \cite{BvG22}, but allow for a more general introduction of $\lambda$. 

One of the main goals of this paper is to characterise the cases in which applying the shift \eqref{eq:shiftgen} to the spectral curve defined by \eqref{eq:detA}  induces a Landau--Ginzburg model $(\lambda, \dd\log\mu)$ for some Frobenius manifold and to study the resulting geometric structure. 

Let $S_{\Roots}^{(I)}$ denote $Q_{\omega_{\text{min}}}^{\text{red}}$ after performing the shift \eqref{eq:shiftgen} for specific $I$. This means that $S_{\Roots}^{(\{\bar{k}\})} \equiv \mathcal{P}_{\omega}$ for the minimal choice of $\omega$ as considered in \cite{BvG22}.  

Note that if $S$ is linear in $\lambda$, we have that the spectral curve given by $S = 0$ has genus 0, and $\lambda$ is a rational function given by
\begin{equation}
    \lambda = \dfrac{w_0 \, Q_{\omega_{\text{min}}}^{\text{red}}}{\sum_{i \in I}f_i[w_i]Q_{\omega_{\text{min}}}^{\text{red}}}.
    \label{eq:lambdagen}
\end{equation}
Due to the form of $Q_{\ell}$,   $S$ will always be linear in $\lambda$ for $\Roots = A_{\ell}$. 

 The following Lemma characterises the identity and Euler vector field for an arbitrary shift of the form in \cref{eq:shiftgen} for any Dynkin type, in terms of the corresponding action on the superpotential.
\begin{lemma}
    Let 
    \begin{equation}
         e = -\sum_{i \in I}\frac{f_i}{w_0}\partial_{w_i}, \quad E = w_0 \partial_{w_0}.
         \label{eq:eandE}
    \end{equation}
Then, 
  \begin{equation*}
   \mathcal{L}_e \lambda = 1 \quad \text{and} \quad  \mathcal{L}_E \lambda = \lambda,
 \end{equation*}
 where $\lambda$ is the spectral parameter introduced by a shift of the type \eqref{eq:shiftgen} to $Q^{\text{red}}_{\omega}$, with $\omega$ being the highest weight of the lowest dimensional non-trivial irreducible representation of Dynkin type $\Roots$, as in \cite{BvG22}.
 \label{lemma:eEgen}
\end{lemma}
\begin{proof}
 Let $S \coloneqq S_{\Roots}^{(I)}$ for a fixed choice of  $\Roots \in \{A_{\ell}, B_{\ell}, C_{\ell}, D_{\ell}, E_6, E_7, E_8, F_4, G_2 \}$,  $I \subset \{1, \cdots, \ell\}$. Then, 
 \begin{equation*}
     0 = \mathcal{L}_e S(\lambda, \mu, \underline{w}) = \partial_{\lambda}S \, e(\lambda) + e(S).
 \end{equation*}
 Moreover, 
\begin{equation*}
     \partial_{\lambda}S = \sum_{i \in I}\frac{f_i}{w_0}\partial_{w_i}S =  -e(S).
\end{equation*}
Since the shift defining \(S\) is invertible over \(\mathbb C(w_0)\), the irreducibility of
\(Q_{\Roots}\) \cite{BvG22} implies that \(S\) is irreducible. Moreover, for generic
\(f_i\), it is non-constant in \(\lambda\). Hence \(\partial_\lambda S\) is non-zero on the
spectral curve, and
\begin{equation*}
     \mathcal{L}_e\lambda = 1. 
\end{equation*}
 Similarly, one shows that
$    E(S)  = -\lambda \partial_{\lambda} S \implies   \mathcal{L}_E \lambda = \lambda.$
\end{proof}
This means that $e, E$ generate the affine action on the associated  Hurwitz space of $\lambda$. As the unit and Euler vector fields of a Hurwitz Frobenius manifold  are characterised by this affine action, this identifies $e$ and $E$ with the unit and Euler vector field, respectively, under the assumption that $\lambda$ is indeed a Landau--Ginzburg superpotential for a Frobenius manifold. Note that Lemma \ref{lemma:eEgen} holds for arbitrary Dynkin type. 

\begin{remark}
\label{rmk:charge}
    Notice that for $e, E$ as in \eqref{eq:eandE}, the charge of any associated Frobenius manifold is one. This follows from the fact that $\mathcal{L}_E \lambda = \lambda$ when applying the Lie derivative in the direction of $E$ to the three-tensor \eqref{eq:cres}. 
\end{remark}

In the following, we specialise to $\Roots = A_{\ell}$ and analyse the geometric structure arising from shifts of the form \eqref{eq:shiftgen}. As mentioned, in this case $\lambda$ is given by \eqref{eq:lambdagen} after a shift of the type \eqref{eq:shiftgen}. 
Letting 
\begin{equation*}
    w_i \mapsto w_i + f_i \dfrac{\lambda}{w_0}, \quad \text{for } i \in I = \{\bar{k}, \bar{k}+1,  \cdots, \bar{k}+p\} \subset \{1, 2, \cdots, \ell\},
\end{equation*}
for some $\bar{k} \in \{1, \cdots, \ell-p\}$, gives 
\begin{equation*}
    \lambda = (-1)^{\bar{k}+p}\dfrac{w_0 Q_{\ell}}{\mu^{\bar{k}} Q_{\ell}^{(I)}},
\end{equation*}
where $Q_{\ell}^{(I)}$ is an order $p$ monic polynomial in $\mu$.  Without loss of generality, we have let $f_{\bar{k}+p} = 1$ \footnote{This normalisation is equivalent to letting $f_i \mapsto \frac{f_i}{f_{\bar{k}+p}}$, and $w_0 \mapsto \frac{w_0}{f_{\bar{k}+p}}$.}.  Since  $|\{f_i\}| = p$  (after normalisation) we can define the map
\begin{equation*}
    (f_{\bar{k}}, \cdots, f_{\bar{k}+p-1}) \mapsto ((-1)^{p}f_{\bar{k}}, \cdots, (-1)^{1}f_{\bar{k}+p-2}, f_{\bar{k}+p-1}) \overset{\text{Vieta}}{\mapsto} (w_{-1}, \cdots, w_{-p}),
\end{equation*}
where the latter arrow represents the Vieta mapping, which sends the coefficient vector of a monic polynomial to its roots. Note that for generic $\{f_i\}$, all roots $\{w_{-j}\}$ will be simple and the composition map $(f_i) \mapsto (w_{-j})$ is a global algebraic isomorphism between the parameter space $(\mathbb{C}^*)^p$ and  $Sym^p(\mathbb{C}^*)$. As such,  for generic $\{f_i\}$ and by absorbing the sign $(-1)^{\bar{k}+p}$ into $w_0$, we obtain
\begin{equation}
  \lambda = \dfrac{w_0 \, Q_{\ell}}{\mu^{\bar{k}} \prod_{i=1}^{p}(\mu-w_{-i})},
  \label{eq:lambdaAsimple}
\end{equation} 
and we may take $\{w_{-p}, \cdots, w_{-1}, w_0, w_{1}, \cdots, w_{\ell}\}$ as coordinates on the parameter space of $\lambda$. 

Hence, the Vieta mapping provides an explicit change of parameters from the shift coefficients $f_i$ to the pole locations $w_{-j}$, thereby identifying any $\lambda$ arising from shifts of the form \eqref{eq:shiftgen} with a natural generalised Dubrovin--Zhang form.

\begin{theorem}
\label{thm:FrobAsimple}
  $(\lambda_{A_{\ell}}, \dd\log\mu)$, for $\lambda_{A_{\ell}}$ as in \eqref{eq:lambdaAsimple}, is a Landau--Ginzburg model for a $\ell+p+1$-dimensional Frobenius manifold, $M$, of charge $d=1$ with unit $e$ and Euler vector field $E$ given by
  \begin{equation}
      e = -\sum_{i \in I}\dfrac{f_i}{w_0}\partial_{w_i}, \quad E = w_0 \partial_{w_0}\,.
      \label{eq:eEAsimple}
  \end{equation}
  The differential $\phi=\dd\log\mu$ is, in particular, $\lambda$-admissible.
\end{theorem}
\begin{proof}
The function \eqref{eq:lambdaAsimple} has a pole of order $\ell+1-\bar{k}-p$ at $\mu=\infty$, a pole of order $\bar{k}$ at zero and $p$ simple poles at $w_{-1},\dots, w_{-p}$. Hence, it represents a point in the Hurwitz space $\Hw_{0\,;\,\underline{n}}$, with $\underline{n} = (\ell-\bar{k}-p,\bar{k}-1,\underset{p}{\underbrace{0,\dots, 0}})$. In fact, by the Riemann-Hurwitz formula 
\begin{equation*}
    \text{dim}(\mathcal{H}_{0;\underline{n}}) = 2(p+1) +\ell-\bar{k}-p+\bar{k}-1 = \ell+1+p = \text{dim}(M). 
\end{equation*}
Thus, since the source curve has genus zero, we use the \(\operatorname{PGL}_2\)-freedom
on the source to choose a representative in which the two distinguished poles
are located at \(0\) and \(\infty\). Then, on the corresponding open locus, a
generic rational function with pole orders \(\ell+1-\bar{k}-p\) at \(\infty\), \(k\)
at \(0\), and simple poles at \(w_{-1},\ldots,w_{-p}\), is represented 
by a function of the form \eqref{eq:lambdaAsimple}, with 
the parameters \(w_{-p},\ldots,w_{-1},w_0,w_1,\ldots,w_\ell\)  providing
local coordinates. 
  Hence,  $(\lambda_{A_{\ell}}, \dd\log\mu, e, E)$ gives a Frobenius manifold of dimension $\ell+1+p$ and charge one, by virtue of the $\lambda$-admissibility of the primary differential $\phi = \dd\log\mu$ for Dubrovin--Zhang manifolds \cite{BvG22} and Lecture 5 of \cite{Dub96},  together with  Lemma  \ref{lemma:eEgen}, and  Remark \ref{rmk:charge},  $d=1$. 
\end{proof}

\begin{remark}
      The Frobenius manifold structure on Hurwitz spaces of the type  $\Hw_{0\,;\,(\ell-\bar{k}-p,\bar{k}-1,0,\dots, 0)}$ with primary differential $\phi$ was studied in \cite{PS26}, where a set of flat coordinates and the corresponding prepotential were explicitly given in terms of the structures on the lower-dimensional Hurwitz space $\Hw_{0\,;\,(\ell-\bar{k}-p-1,\bar{k}-1)}$.
\end{remark}

\begin{remark}
\label{rmk:Afepsilon}
   Note that \eqref{eq:lambdaAsimple} is obtained from $Q_{\ell}$ via the shift \begin{gather}
     w_{\bar{k}+i} \rightarrow w_{\bar{k}+i} + \frac{\lambda}{w_0}\epsilon_{p-i}(w_{-1}, \cdots, w_{-p}), \qquad \text{ for } i = 0, \cdots, p,
 \end{gather}
where $\epsilon_i$ is the $i^{\text{th}}$ elementary symmetric polynomial in the coordinates $w_{-1}, \cdots, w_{-p}$. This can be seen as follows.
    Performing the proposed shift on \eqref{eq:detA} gives
    \begin{equation}
        \frac{\lambda}{w_0}\sum_{i=0}^p\left((-1)^{\bar{k}+i}\mu^{k+i}\epsilon_{p-i}(w_{-1}, \cdots, w_{-p})\right) +  Q_{\mathcal{A_{\ell}}}.
        \label{eq:Ashifttemp1}
    \end{equation}
    Now, we know that elementary symmetric polynomials are precisely related to factorisations of  polynomials as
    \begin{equation}
        P(x) = \prod_{i=1}^{n}(x-x_i) = \sum_{i=0}^{n}\epsilon_i(x_1, \cdots, x_n)(-1)^{i}x^{n-i},
    \end{equation}
     for a monic polynomial $P$ in $x$ of degree $n$.     Hence, \eqref{eq:Ashifttemp1} becomes
    \begin{equation}
        \frac{\lambda}{w_0}(-1)^{p+\bar{k}} \mu^{\bar{k}} \sum_{i=0}^{p}\underset{\prod_{j=1}^{p}(\mu-w_{-j})}{\underbrace{\left(-1)^{p-i}\mu^i \epsilon_{p-i}\right)}} +  Q_{\ell}.
    \end{equation}
    Considering the zero locus and solving for $\lambda$ gives the expression. 
\end{remark}

\begin{remark}
\label{rmk:maZuoisoA1}
    Letting $|I| = 2$ gives
    \begin{equation}
        \lambda = \dfrac{w_0 \, Q_{\ell}}{\mu^{\bar{k}} (\mu - w_{-1})} \underset{\eqref{eq:wkinx}}{=} \dfrac{w_0 \prod_{\tilde{\omega} \in \Gamma(\rho_{\omega})}^{}(\mu-e^{\tilde{\omega}(x)})}{\mu^{\bar{k}} (\mu-w_{-1})} = \dfrac{w_0 \prod_{i=1}^{\ell+1}(\mu-a_i)}{\mu^k(\mu-w_{-1})},
    \end{equation}
    where  $\Gamma(\rho_{\omega})$ is the weight system of $\rho_\omega$, the $\ell+1$-dimensional defining representation of $SL_{\mathbb{C}}(\ell+1)$, and  $w_0 = e^{c_{\omega}x_{0}}$. By letting $\bar{k} \mapsto m$, $w_{-1} \mapsto e^{i \phi_{\ell+2}}$, $\hat{a}_i \coloneqq a_i e^{\frac{c_{\omega}x_0}{\ell+1-\bar{k}}}$, and $\hat{\mu} = \mu e^{-\frac{c_{\omega}x_0}{\ell+1-\bar{k}}}$, we obtain precisely the form of \cite{Zuo20}, with $i\phi_j = \log\tilde{a}_j$, and $i \phi = \log\hat{\mu}$.
\end{remark}

On the locus in which movable poles coalesce \eqref{eq:lambdaAsimple} becomes
\begin{equation}
\lambda(\mu) =\frac{w_0Q_{\hat{\ell}+\bar{k}+p}}{\mu^{\bar{k}}\prod_{i=1}^{\ell(\pi)}(\mu-w_{-i})^{m_i}}\, ,  
     \label{eq:lambdaAmult}
\end{equation}
 for any partition $\pi=(m_1,\dots, m_{\ell(\pi)})\vdash p\,$, where $\ell(\pi)$ is the length of the partition $\pi$,  and $\bar{k} = \text{min}(I)$. By Remark \ref{rmk:Afepsilon}, this is obtained from $Q_{\ell\equiv \hat{\ell}+\bar{k}+p}$ using the shift functions  $f_{\bar{k}+p-i} \equiv \epsilon_i(w_{-p}, \cdots, w_{-1})|_{h(w_{-p, \cdots, w_{-1}})}$, where the map $h$ describes how the movable poles coalesce; $\eqref{eq:lambdaAsimple}|_{h(w_{-1}, \cdots,  w_{-p}) } = \eqref{eq:lambdaAmult}$. Writing $\lambda$ in this form will prove convenient when interpreting coalescence of movable poles in terms of a stratification of the Hurwitz boundary in Section \ref{sec:polecollision}.

\begin{theorem}
\label{thm:FrobAmult}
$(\lambda, \dd\log\mu)$, with $\lambda$ as in \eqref{eq:lambdaAmult} gives a Frobenius manifold structure of charge $d=1$ on the Hurwitz space $\mathcal{H}_{0; \underline{n}}$, with $\underline{n} = (\ell-\bar{k}-p, \bar{k}-1, m_1-1, \cdots m_{\ell(\pi)}-1)$, where   
\begin{equation*}
    e = e^{(I)}|_{h(w_{-p}, \cdots, w_{-1})}, \quad \text{and} \quad E = w_0 \partial_{w_0},
\end{equation*}
for $e^{(I)}$ as in \eqref{eq:eEAsimple}, where $h$ denotes the map of coalescence defined by $\eqref{eq:lambdaAsimple}|_{h(w_{-1}, \cdots,  w_{-p}) } = \eqref{eq:lambdaAmult}$. The differential $\phi=\dd\log\mu$ is, in particular, still $\lambda$-admissible.
\end{theorem}

\begin{proof}
    Again, the dimension of the Frobenius manifold equals the dimension of the associated Hurwitz space, which means that the theorem follows by the same logic as in the proof of Theorem \ref{thm:FrobAsimple}, using Lemma \ref{lemma:eEgen} restricted to $h(w_{-1}, \cdots, w_{-p})$.
\end{proof}

\begin{remark}
   By letting $|I|=r+1$ and $w_{-r}= \cdots w_{-1} = e^{i \phi_{\ell+2}}$, $\bar{k} \mapsto m$ and defining $\hat{a}_i$, $\hat{\mu}$ analogously to Remark \ref{rmk:maZuoisoA1}, we obtain precisely the form of Proposition 2.4 in \cite{MZ24}, with $i\phi_j = \log\tilde{a}_j$ and $\mu = z = e^{i\phi}$. 
\end{remark}

While we can still send the shifting parameters $\{f_i\}$ to the roots of $Q_{\ell}^{(I)}$, the map is no longer an isomorphism. As a result, it is no longer possible to know precisely which nodes in the associated  Dynkin diagram are marked. To illustrate this consider $\ell=7$ with $I= \{1, 2, 3, 4\}$. For instance, letting  $w_{-3}\mapsto w_{-2}$ or $w_{-2}\mapsto w_{-1}$, should naively be associated with extensions of Dynkin nodes $\{1,2,4\}$ and $\{1,3,4\}$, respectively. However,  they give the same superpotential. This raises the question of whether the corresponding generalised DZ-manifolds, in an orbit space description, would be isomorphic as well, or whether this reflects a true discrepancy for higher order poles in more than two extended directions. 

Now, suppose we shift non-consecutively. That is, let $I \subset \{1, \cdots, \ell\}$, such that, in the ordering $i_{j_1}<i_{j_2}$ for $j_1<j_2$,  $I = \{\bar{k},  i_1, \cdots, i_p\} \neq  \{\bar{k}, \bar{k}+1, \cdots, \bar{k}+p\}$ for any $\bar{k} \in \{1, \cdots, \ell\}$. The resulting $Q_{\ell}^{(I)}$ is a monic polynomial of degree $\text{max}(I) - \bar{k}$ with some coefficients vanishing (specifically, coefficients of $\mu^s$ for $s$ being the skipped labels in $I$), after taking out the common factor $\mu^{\bar{k}}$.  In this case, the map sending the shift parameters to the roots of $Q_{\ell}^{(I)}$ fails to be surjective. However, it is still injective; in particular, we can still send the parameters $\{f_i\}$ to a size $p$ subset of the set of roots. Let us denote the image by $\{w_{-j}\}_{j=1, \cdots, p}$.  The remaining $\text{max}(I) - \bar{k}-p$ roots become smooth algebraic functions of $\underline{f}(w_{-1}, \cdots, w_{-p})$ given by the Vieta relations. Note that for generic $f_i$, as in the consecutive case, all poles away from $0, \infty$ are simple.

\begin{proposition}
\label{prop:Anonconsec}
  The Saito metric is degenerate for  $I$ not of the form  $ \{\bar{k}, \bar{k}+1, \cdots, \bar{k}+p\}$  for $\bar{k} \in \{1, \cdots, \ell\}$. 
\end{proposition}
\begin{proof}
 Let  $I = \{\bar{k},  i_1, \cdots, i_p\} $, where $i_{j_1}<i_{j_2}$ for $j_1<j_2$, and $\bar{k} = \text{min}(I)$. The proposition will be proved by showing that there exist at least $p+2$ columns and rows of $\eta$, each having at most $p+1$ non-zero elements located in the first $p+1$ positions. It follows that  $\text{det}(\eta)=0$. 

   Near a pole at $\mu=r$ of order $s$ we have that $\lambda' \sim (\mu-r)^{-s-1}$ and for $i>0$
   \begin{equation*}
       \partial_{w_i}\lambda \sim \begin{dcases}
           \mu^{i-\bar{k}}, & \text{ if } r = 0,\\
           \mu^{i-\text{max}(I)}, & \text{ if } r = \infty,\\
           (\mu-r)^{-s}, & \text{ otherwise.}
       \end{dcases} 
   \end{equation*}
   Thus, for $\eta$ as defined in \eqref{eq:etares}:
   \begin{itemize}
       \item $\mu = 0$: there is no non-zero contribution to $\eta_{ij}$ as $i+j-\bar{k}-1 \geq 0$,  for $j>0, \bar{k} \leq i$;
       \item $\mu = \infty$: there is no non-zero contribution to $\eta_{ij}$ as $i+j-\ell-2-\text{max}(I) \leq -2$,  for $j\leq \ell, i \leq \text{max}(I)$;
       \item $\mu = r$: to have a non-zero contribution to $\eta_{ij}$ we must have $-s+1 \leq -1$. 
   \end{itemize}
   Since all poles away from $0, \infty$ are simple, we have that $\eta_{ij} = 0$ for $\bar{k} \leq i \leq \text{max}(I)$, $j >0$ (or equivalently $i \leftrightarrow j$),    which gives the result.
\end{proof}

As a consequence of Theorems \ref{thm:FrobAsimple}, \ref{thm:FrobAmult} and Proposition \ref{prop:Anonconsec}, we have the following classification result. 
\begin{theorem}
\label{thm:AMain}
    All Frobenius manifolds obtained for $\Roots = A_{\ell}$ by a shift of type \eqref{eq:shiftgen} are generalised type $A_{\ell}$ DZ-manifolds in the sense that they correspond to  Frobenius structures obtained from Landau--Ginzburg models $(\lambda, \dd\log\mu)$ with $\lambda$  of the form \eqref{eq:lambdaAmult}.
\end{theorem}

We conclude this section by explicitly giving the flat coordinates for the intersection form.
 \begin{proposition}
     Let
     \begin{equation*}
         w_i = \begin{dcases}
             \eqref{eq:wkinx}, & \text{ for } i > 0,\\
             e^{x_{i}}, & \text{ for } i \leq 0.
         \end{dcases}
     \end{equation*}
        The coordinates $\{x_i\}_{i=-\ell(\pi), \cdots, \ell}$ constitute a flat frame for the intersection form $g$ of the Frobenius manifold of Theorem \ref{thm:AMain} associated with the partition $\pi=(m_1,\dots, m_{\ell(\pi)})\vdash p\,$. In particular, its Gram matrix is block diagonal: 
     \[
    g=
    \begin{pmatrix}
        G_1 & 0\\
        0 & G_2
    \end{pmatrix},
\]
with blocks:
\[
    G_1=
    \dfrac{1}{\bar{k}}\, \begin{pmatrix}
        -1-\frac{\bar{k}}{\ell+1-\bar{k}-p}
        &
       m_1
        &
        \cdots
        &
       m_{\ell(\pi)}
        \\[1em]
       m_1
        &
        -\bar{k} m_1-m_1^2
        &
        \cdots
        &
        -m_1 m_{\ell(\pi)}
        \\
        \vdots
        &
        \vdots
        &
        \ddots
        &
        \vdots
        \\
        m_r
        &
        -m_{\ell(\pi)} m_1
        &
        \cdots
        &
        -\bar{k} m_{\ell(\pi)}-m_{\ell(\pi)}^2
    \end{pmatrix},
\]
and
\[
    (G_2)_{ij}
    =
    \sum_{\widehat{\omega}\in\Gamma(\rho_{\omega})}
    \widehat{\omega}_i\widehat{\omega}_j,
    \qquad 1\leq i,j\leq \ell.
\]
 \end{proposition}
 \begin{proof}
 Written in $x$-coordinates, $\lambda$ takes the form
 \begin{equation*}
     \lambda(\mu;x) = \dfrac{e^{x_0} \prod_{\widehat{\omega} \in \Gamma(\rho_{\omega})}(\mu-e^{\widehat{\omega},x})}{\mu^{\bar{k}}\prod_{\alpha=1}^{\ell(\pi)}(\mu-e^{x_{-\alpha}})^{m_{\alpha}}}.
 \end{equation*}
   Let
     \begin{equation}
         \bar{g}_{ij} \coloneqq \dfrac{\partial_{x_i}\log\lambda \,\partial_{x_j}\log\lambda}{\mu^2 \, \partial_\mu \log\lambda}\,,
         \label{eq:xflattemp}
     \end{equation}
so that: $    g_{ij} = -\sum_{q_a}\underset{q_a}{\Res} \,\bar{g}_{ij} \text{d}\mu
$ for $q_a \in \text{Cr}_{\lambda}\,$. 

Turning around the contour allows us to consider the residues of \eqref{eq:xflattemp} at the poles $(\mu=0, \infty, e^{x_{-\alpha}})$ and zeros $(\mu=e^{(\widehat{\omega}, x)})$ of $\lambda$ rather than its critical points. The leading order behaviours of the derivatives of log $\lambda$ near the relevant points are shown in Table \ref{tab:gasym}

 \begin{table}[h!]
\begin{center}
{\renewcommand{\arraystretch}{1.5}
\begin{tabular}{ |c|c|c| }
\hline
Point
& 
$\partial_\mu\log\lambda$
& 
$L_a=\partial_{x_a}\log\lambda$
\\
\hline
\hline
$\mu=0$
&
$-\frac{\bar{k}}{\mu}+\order{1}$
&
$L_0=1,\quad
L_{-\alpha}=-m_\alpha+\order{\mu},\quad
L_i=\order{\mu}$
\\
\hline
$\mu=e^{x_{-\alpha}}$
&
$-\frac{m_\alpha}{\mu-e^{x_{-\alpha}}}+\order{1}$
&
$L_{-\alpha}
=
\frac{m_\alpha e^{x_{-\alpha}}}
{\mu-e^{x_{-\alpha}}}
+\order{1},
\quad
L_a=\order{1}\ $ for $a\neq-\alpha$
\\
\hline
$\mu=e^{(\widehat{\omega},x)}$
&
$\frac{1}{\mu-e^{(\widehat{\omega},x)}}+\order{1}$
&
$L_i
=
-\frac{\widehat{\omega}_i e^{(\widehat{\omega},x)}}
{\mu-e^{(\omega,x)}}
+\order{1},
\quad
L_0,L_{-\alpha}=\order{1}$
\\
\hline
$\mu=\infty$
&
$\frac{N}{\mu}+\order{\mu^{-2}}$
&
$L_0=1,\quad
L_{-\alpha}=\order{\mu^{-1}},\quad
L_i=\order{\mu^{-1}}$\\
\hline
\end{tabular}
}
\end{center}
\caption{Asymptotics of derivatives of log $\lambda$ near $g$-relevant points.}
\label{tab:gasym}
\end{table}
 Hence, 
     \begin{itemize}
         \item at $\mu=0$:
         \vspace{-0.5em}
         \begin{equation*}
             \bar{g}_{ij} = -\dfrac{L_iL_j}{\bar{k}\mu} + \order{\mu}\, \,   \implies \underset{0}{\text{Res}}\, \bar{g}_{ij}\, \text{d}\mu = -\frac{1}{\bar{k}}\begin{dcases}
                 1 & i=j=0,\\
                 -m_a & i=0, j=-a \text{ or } i=-a, j=0,\\
                 m_{-i} m_{-j} & i,j<0,\\
                 0 & \text{ otherwise;}
             \end{dcases}
         \end{equation*}
         \item at $\mu = \infty$:
         \begin{equation*}
             \bar{g}_{ij} = \frac{L_i L_j}{(\ell+1-\bar{k}-p)\mu} + \order{\mu^{-2}} \, \, \implies \underset{\infty}{\text{Res}}\, \bar{g}_{ij}\, \text{d}\mu = -\frac{1}{\ell+1-\bar{k}-p},
         \end{equation*}
         for $i=j=0$, and zero otherwise. 
         \item at $\mu = e^{x_{-\alpha}}$: 
         \begin{equation*}
             \bar{g}_{ij} = - \frac{(\mu-e^{x_{-\alpha}})}{m_{\alpha}\mu^2}L_i L_j + \text{less singular terms} \, \, \implies \underset{e^{x_{-\alpha}}}{\text{Res}}\bar{g}_{ij}\text{d}\mu = -m_{\alpha}, 
         \end{equation*}
         for $i=j=-\alpha$, and zero otherwise;
         \item at $\mu = e^{(\omega,x)}$:
         \begin{equation*}
             \bar{g}_{ij} = \frac{\mu-e^{(\widehat{\omega}, x)}}{\mu^2}L_i L_j + \text{ less singular terms} \, \, \implies \underset{e^{(\omega,x)}}{\text{Res}}\bar{g}_{ij}\text{d}\mu = \omega_i \omega_j.
         \end{equation*}
     \end{itemize}
 \end{proof}
\label{sec:Ashifts}

\section{Examples}\label{sec:Ex}
 
We provide a set of examples where the $\eta$-flat coordinates and prepotentials are computed explicitly. In particular, we observe that when the prepotential is associated with a longest partition, we recover the form of \cite[Section 4]{PS26} (\cref{ex:A5123}), and, when it is associated to a shortest partition, we recover the conjectured form of \cites{Zuo20,MZ24} (\cref{ex:A513,ex:A71234deep}). In the intermediate cases,  new types of terms appear that were unobserved in those previous works (as seen in \cref{ex:A71234}). The examples appear in pairs, where the latter example of each pair is related to the former by colliding a pair of movable poles of the associated superpotential. The flat coordinates not obtained by taking residues of functions involving poles that have collided, are related simply by taking the corresponding limit. For the remaining coordinates some regularisation and renormalisation must be applied. In the next section, we will introduce a framework able to predict these phenomena, for an arbitrary partition, via subsequent coarsenings of the longest partitions, producing higher-order poles in the superpotential.

In the following examples, up to taking linear combinations respecting the grading, when indicated,  $\eta$-flat coordinates have been computed using the formulae from \cite{Dub96} defined on the Hurwitz space $\mathcal{H}_{g; \underline{n}}\,$, for $\underline{n}\in\ZZ_{\geq 0}^{q+1}\,$:
\begin{subequations}
\begin{equation}	
    \tau_{i;\alpha} \coloneqq \,  \underset{\infty_i}{\Res} \,  \kappa_i^{-\alpha} \log\mu \, \dd\lambda\,, \quad \alpha = 1, \dots, n_i,
    \label{Eq:5.1a}
\end{equation} 
\begin{equation}
    \tau_j^{\text{ext}} \coloneqq \mathrm{p.v. } \,  \int_{\infty_0}^{\infty_j} \dd \log\mu, \quad j=1, \dots, q+1,
    \label{Eq:5.1b}
\end{equation}
\begin{equation}
    \tau_i^{\text{res}} \coloneqq\,  \underset{\infty_i}{\Res}\lambda \, \dd\log\mu\,, \quad i = 0, \dots, q+1,
    \label{Eq:5.1c}
\end{equation}
\end{subequations}
where $\kappa_i$ are local coordinates defined by $\lambda(\kappa)=\kappa_i^{n_i+1} +\order{1}$ near $\infty_i$, in the notation of Section \ref{sec:subHur}, and the principal value, \text{p.v.}, indicates the subtraction of the divergent part in $\kappa_i$.

Common to all the examples are: 
\begin{equation}
    t_{0} = \text{log}(w_0), \quad t_i = \text{log}(w_0 w_i), \text{ for } i<0,
    \label{eq:exnegflats}
\end{equation}
\begin{equation*}
    e =  -\sum_{i \in I}^{} \dfrac{f_i}{w_0}\partial_{w_i} = \partial_\ell, \qquad E = w_0 \partial_{w_0} = \sum_{i=1}^{6}  d_i t^i \partial_i   + \sum_{j = -p}^{0}\partial_{j},
\end{equation*}
for $\underline{d} = (d_i)_{i=1, \cdots, \ell}$ to be specified. Here we have chosen $t_i = \text{log}(w_i) + \text{log}(w_0)$ for $i<0$ as opposed to simply $\text{log}(w_i)$ as obtained by the residue formulae in \cite{Dub96}. This is to simplify comparison with orbit space constructions and for the Euler vector field to take standard form. This is allowed due to the grading attached to all non-positively indexed coordinates being 0. The prepotentials are obtained by computing the c-tensor components in $\eta$-flat coordinates   using \eqref{eq:cres} (with $\phi = \dd\log\mu$) with Wolfram Mathematica and integrating three times.

\begin{example}[$\Roots = A_5$, $I = \{1,2,3\}$, $\pi=(1,1)$] \label{ex:A5123}
Consider $Q_5$, as in \eqref{eq:detA} and let 
\begin{equation*}
    w_i \mapsto w_i + f_i \frac{\lambda}{w_0}, \text{ for } i \in I = \{1,2,3\}, \quad \text{with }  f_i = \epsilon_{3-i}(w_{-1}, w_{-2}). 
\end{equation*}

This gives
\begin{align*}
    S \, & \, =  1+ \mu^6+\sum_{i=1}^{3}(-1)^i(w_i+\epsilon_{3-i}(w_{-1}, w_{-2}))\mu^i+\sum_{i=4}^{5}(-1)^iw_i\mu^i \\ \, & \,  =  Q_{5} + \frac{\lambda}{w_0}\left(-w_{-1}w_{-2}\mu+(w_{-1}+w_{-2})\mu^2 -\mu^3\right), 
\end{align*}
with $\epsilon_i$ denoting the $i^{\text{th}}$ elementary symmetric polynomial in its arguments, and $Q_5$ as in \eqref{eq:detA}. By solving $S=0$ for $\lambda$ we get
\begin{equation*}
   \lambda = \dfrac{w_0 \, Q_{5}}{\mu^{}(\mu-w_{-1})(\mu-w_{-2})}\,,
\end{equation*}
after letting $w_0 \mapsto -w_0$. A flat frame is given by $t_i$ as in \eqref{eq:exnegflats} for $i\leq 0$, together with
\begin{gather*}
 t_1 = w_0^{\frac{1}{3}}\left(w_{-2}+w_{-1}-w_5\right), \qquad (\text{\ref{Eq:5.1a}}, \infty_i = \infty, \alpha = 1);\\
  t_2 = w_0^{\frac{2}{3}}\left(5(w_{-2}^2+w_{-1}^2 + 4w_{-2}w_{-1}-4w_5(w_{-2}+w_{-1})-w_5^2+6w_4)\right), \qquad (\text{\ref{Eq:5.1a}}, \infty_i = \infty, \alpha = 2);\\
  t_3 = \dfrac{w_0}{w_{-1}^2(w_{-1}-w_{-2})}\left(w_{-1}^6+1+\sum_{j=1}^{5}(-1)^j w_j w_{-1}^j\right),\qquad (\text{\ref{Eq:5.1c}}, \infty_i = w_{-1});\\
    t_4 = w_0\left((w_{-2}+w_{-1})(w_{-2}^2+w_{-1}^2+w_4)-w_5(w_{-2}^2+w_{-2}w_{-1}+w_{-1}^2)-w_3\right)-t_5, \\ (\text{\ref{Eq:5.1c}}, \infty_i = \infty)- (\text{\ref{Eq:5.1c}}, \infty_i = 0),\\
    t_5 =-\dfrac{w_0}{w_{-2}^2w_{-1}^2}\left(w_{-2}w_{-1}w_1-w_{-2}-w_{-1}\right), \qquad  (\text{\ref{Eq:5.1c}}, \infty_i = 0),
\end{gather*}
with $\underline{d}= \left(\frac{1}{3}, \frac{2}{3}, 1,1,1\right)$, and the label next to each coordinate denotes how the coordinate was obtained (up to a constant factor). Note that for $t_4$ we have taken a linear combination of flat coordinates such that $\partial_i \lambda = 1$ for some $i$ (here $i=5$). This means that the unit vector field is  $\partial_{i}$  and $\eta_{jk} = c_{ijk} =  \partial^3_{ijk} F$.  The flat coordinates have been ordered in ascending degree. 

The associated prepotential is given by, 
\begin{align} 
    F = \, & \, 
\frac{1}{3}e^{3t_{-2}-2t_0}
\left(t_3-t_4\right)
-\frac{1}{3}e^{3t_{-1}-2t_0}t_3
+\frac{1}{2}e^{2t_{-2}-\frac{4}{3}t_0}
t_1\left(t_3-t_4\right)
-\frac{1}{2}e^{2t_{-1}-\frac{4}{3}t_0}t_1t_3
\\ \notag
\, & \, 
+\frac{1}{6}e^{t_{-2}-\frac{2}{3}t_0}
\left(t_1^2t_3+t_2t_3-t_1^2t_4-t_2t_4\right)
-\frac{1}{6}e^{t_{-1}-\frac{2}{3}t_0}
\left(t_1^2+t_2\right)t_3
-\frac{1}{6}e^{-t_{-2}-t_{-1}+\frac{10}{3}t_0}
\left(t_1^2+t_2\right)
\\ \notag
\, & \, 
-e^{-2t_{-2}-t_{-1}+4t_0}(t_3-t_4)
+e^{-t_{-2}-2t_{-1}+4t_0}t_3
+t_3(t_3-t_4)\log(e^{t_{-2}}-e^{t_{-1}})
-\frac{1}{2}t_3^2\log t_3
\\ \notag 
\, & \, 
-\frac{1}{2}(t_3-t_4)^2\log(t_3-t_4)
+\frac{t_1^6}{19440}
-\frac{1}{3888}(t_1^4-3t_1^2t_2+t_2^2)t_2
-\frac{1}{18}t_1t_2t_4
-\frac{1}{2}(t_{-2}+t_{-1})t_3^2
-\frac{1}{2}t_{-2}t_4^2
\\
\notag \, & \,
+\frac{1}{3}t_0t_4^2
+t_{-2}t_3t_4 + t_5 \left(\frac{t_{-2}}{2}(2t_3-2t_4 + t_5 ) - \frac{t_{-1}}{2}(2t_3-t_5)+\frac{t_0}{3}(2t_4-5t_5)-\frac{t_1t_2}{18} \right),
\end{align}
    \label{eq:F1}
    which, as expected, has the form of \cite{PS26}.
\end{example}

\begin{example}[$\Roots = A_5$, $I = \{1,2,3\}$, $\pi=(2)$]\label{ex:A513}
Here we consider the same $Q_5$ and $I$ as in the previous example, but now the shift functions $f_i$ for $i \in I$ are
\begin{equation*}
  f_i = \epsilon_{3-i}(w_{-1}, w_{-2})|_{w_{-2} \mapsto w_{-1}} \implies    \lambda = \dfrac{w_0 \, Q_{5}}{\mu^{}(\mu-w_{-1})(\mu-w_{-2})}\Big|_{w_{-2}\mapsto w_{-1}}= \dfrac{w_0 \, Q_{5}}{\mu^{}(\mu-w_{-1})^2}.
\end{equation*}

A flat frame is given by $t_i$ as in \eqref{eq:exnegflats} for $i\leq 0$, together with
\begin{gather*}
    t_i = \tilde{t}_i, \quad \text{ for } i=1,2, 4,5, \\
    t_3 = -2 \left(\frac{w_0}{w_{-1}^3}(w_{-1}^6+1+\sum_{i=1}^{5}(-1)^iw_iw_{-1}^i)\right)^{\frac{1}{2}},  \qquad (\text{\ref{Eq:5.1a}}, \infty_i = w_{-1}, \alpha = 1),    
\end{gather*}
where $\underline{d}= \left(\frac{1}{3}, \frac{2}{3}, \frac{1}{2}, 1,1\right)$, and $\tilde{t}_i$ denotes $t_i$ as in  Example \ref{ex:A5123} after imposing $w_{-2}\mapsto w_{-1}$. Notice that the flat coordinates obtained from taking residues away from the  movable pole $w_{-1}$ correspond to limits of flat coordinates of the previous example. The fact that coordinates obtained from poles not involved in a collision in general have well-defined limits and give flat coordinates for the target Frobenius manifold is not difficult to prove analytically. The behaviour of the coordinates obtained from colliding poles is more involved and will be studied in the subsequent section.  

The associated prepotential is given by
\begin{align} 
F = \, & \,
-\frac{1}{4}e^{3t_{-1}-2t_0}
\left(t_3^2+\frac{4}{3}t_4\right)
-\frac{1}{4}e^{2t_{-1}-\frac{4}{3}t_0}
\left(t_3^2+2t_4\right)t_1
-\frac{1}{24}e^{t_{-1}-\frac{2}{3}t_0}
\left(t_3^2+4t_4\right)\left(t_1^2+t_2\right)
\\ \notag
\, & \,
-\frac{1}{6}e^{\frac{10}{3}t_0-2t_{-1}}
\left(t_1^2+t_2\right)
-\frac{1}{4}e^{4t_0-3t_{-1}}
\left(t_3^2-4t_4\right)
-\frac{1}{2}t_4^2\log t_3
+\frac{t_1^6}{19440}
-\frac{1}{3888}(t_1^4-3t_1^2t_2+t_2^2)t_2
\\ \notag
\, & \,
-\frac{1}{18}t_1t_2t_4
+\frac{t_3^4}{384}
-\frac{1}{8}t_3^2t_4
-\left(\frac{1}{2}t_{-1}-\frac{1}{3}t_0\right)t_4^2 -t_5 \left(t_{-1}(t_4-t_5)- \frac{t_0}{3}(2t_4-5t_5)+\frac{1}{18}t_1t_2+\frac{1}{4}t_3^2 \right),
    \label{eq:F2}
\end{align}
which has the form conjectured in \cite{MZ24}, for a superpotential of type $A$ with an order two movable pole. 
\end{example}

\begin{example}[$\Roots = A_7$, $I = \{1,2,3,4\}$, $\pi=(1,2)$]\label{ex:A71234}
In this case, we have:
\begin{equation*}
  f_i = \epsilon_{4-i}(w_{-1}, w_{-2}, w_{-3})|_{w_{-3} \mapsto w_{-1}} \text{ for } i \in I \implies    \lambda = \dfrac{w_0 \, Q_{7}}{\mu^{}(\mu-w_{-1})^2(\mu-w_{-2})}.
\end{equation*}
A flat frame is given by $t_i$ as in \eqref{eq:exnegflats} for $i\leq 0$, together with
\begin{gather*}
    t_1 = w_0^{\frac{1}{4}}(2w_{-2}+w_{-1}-w_7), \qquad (\ref{Eq:5.1a}, \infty_i = \infty, \alpha = 1),\\[1mm] 
    t_2 = w_0^{\frac{1}{2}}\left(
    8w_{-2}^2+3w_{-1}^2+4w_{-1}w_{-2}
    -2w_7(2w_{-2}+w_{-1})-w_7^2+4w_6
    \right), \qquad (\ref{Eq:5.1a}, \infty_i = \infty, \alpha = 2),\\[1mm]
    t_3 =\left(\frac{w_0 (1+w_{-1}^8+\sum_{i=1}^{7}(-1)^iw_{-1}^iw_i)}{w_{-1}^3(w_{-1}-w_{-2})}\right)^{\frac{1}{2}} \qquad (\ref{Eq:5.1a}, \infty_i = \infty, \alpha = 3),\\[1mm]
    t_4 = w_0^{\frac{3}{4}}\Big(
    77w_{-2}^3+126w_{-2}^2w_{-1}
    +180w_{-2}w_{-1}^2+280w_{-1}^3 
    -w_7(63w_{-2}^2+108w_{-2}w_{-1}+180w_{-1}^2) \\
    +(72w_6-9w_7^2)(w_{-2}+2w_{-1}) 
    -5w_7^3-96w_5+24w_6w_7
    \Big) \qquad (\ref{Eq:5.1a}, \infty_i = w_{-1}, \alpha = 1),\\[1mm]
    t_5=
w_0\Bigg(
\sum_{k=0}^{4}(k+1)w_{-2}^{4-k}w_{-1}^{k}
+
\sum_{i=4}^{7}(-1)^i w_i
\sum_{k=0}^{i-4}
(k+1)w_{-2}^{i-4-k}w_{-1}^{k}
\Bigg)
 - t_7, \qquad  (\ref{Eq:5.1c}, \infty_i  = \infty) - (\ref{Eq:5.1c}, \infty_i = 0),\\[1mm]
    t_6 =
    \frac{w_0}{w_{-1}^3(w_{-2}-w_{-1})^2}
\Bigg(
w_{-1}^8(5w_{-1}-6w_{-2})
+\sum_{i=1}^{7}(-1)^i w_i w_{-1}^i
\big((i-3)w_{-1}-(i-2)w_{-2}\big)
+2w_{-2}-3w_{-1}
\Bigg), \\ (\ref{Eq:5.1c}, \infty_i = w_{-1}),\\[1mm]
    t_7 = \dfrac{w_0(w_{-2}(w_{-1}w_1-2)-w_{-1})}{w_{-2}^2w_{-1}^3}, \qquad (\ref{Eq:5.1c}, \infty_i = 0),
\end{gather*}
with $\underline{d} = \left(\frac{1}{4}, \frac{1}{2}, \frac{1}{2}, \frac{3}{4}, 1, 1, 1\right)\,$. Again, we have taken a linear combination of flat coordinates to make the unit vector field $\partial_7$. Moreover, 
\begin{align}
F ={}& 
-\frac{1}{4}e^{4t_{-2}-3t_0}(t_5-t_6)
-\frac{1}{4}e^{4t_{-1}-3t_0}(4t_3^2+t_6)
-\frac{1}{3}e^{3t_{-2}-\frac{9}{4}t_0}t_1(t_5-t_6)
-\frac{1}{3}e^{3t_{-1}-\frac{9}{4}t_0}t_1(3t_3^2+t_6)
\\\notag
&\,
-\frac{1}{8}e^{2t_{-2}-\frac{3}{2}t_0}
(t_1^2+t_2)(t_5-t_6)
-\frac{1}{8}e^{2t_{-1}-\frac{3}{2}t_0}
(t_1^2+t_2)(2t_3^2+t_6)
-\frac{1}{96}e^{t_{-2}-\frac{3}{4}t_0}
(t_1^3+6t_1t_2+t_4)(t_5-t_6)\\ \notag & \, 
-\frac{1}{96}e^{t_{-1}-\frac{3}{4}t_0}
(t_1^3+6t_1t_2+t_4)(t_3^2+t_6)
+\frac{1}{96}e^{-t_{-2}-2t_{-1}+\frac{17}{4}t_0}
(t_1^3+6t_1t_2+t_4)
-e^{-2t_{-2}-2t_{-1}+5t_0}(t_5-t_6)
\\ \notag
&\,
+e^{-t_{-2}-3t_{-1}+5t_0}(t_3^2-t_6)
+\frac{t_3^2}{e^{t_{-2}}-e^{t_{-1}}}
\left(
e^{t_{-1}}t_5
-\frac{e^{t_{-2}}+e^{t_{-1}}}{2}t_6
\right)
+\log(e^{t_{-2}}-e^{t_{-1}})\,t_6(t_6-t_5)
\\\notag 
&\,-\frac{1}{2}t_6^2\log t_3
-\left(\frac{1}{2}t_{-2}-\frac{3}{8}t_0\right)t_5^2
+t_{-2}t_5t_6
-\frac{1}{2}(t_{-2}+t_{-1})t_6^2
-\frac{1}{2}(t_5-t_6)^2\log(t_5-t_6)
-\frac{1}{128}t_2^2t_5
+\frac{t_3^4}{24}\\ \notag & \, 
+\frac{1}{4128768}
\left(
t_1^8-14t_1^6t_2+84t_1^4t_2^2
-168t_1^2t_2^3+84t_2^4
\right)
+\frac{1}{2949120}
\left(
t_1^5-20t_1^3t_2+60t_1t_2^2
+5(t_1^2-2t_2)t_4
\right)t_4 
\\ \notag
&\,
+t_7\left(
-\frac{1}{128}t_2^2
-t_3^2
-\frac{1}{384}t_1t_4
+\left(-t_{-2}+\frac{3}{4}t_0\right)t_5
+(t_{-2}-t_{-1})t_6
+\left(\frac{1}{2}t_{-2}+t_{-1}-\frac{17}{8}t_0\right)t_7
\right).
\end{align}
Notice in particular that this prepotential includes rational terms in the difference, $e^{t_{-2}}-e^{t_{-1}}$, between the two movable poles. Such terms are not present in the predicted forms of \cites{PS26, Zuo20, MZ24}. 
\end{example}

\begin{example}[$\Roots = A_7$, $I = \{1,2,3,4\}$, $\pi=(3)$]\label{ex:A71234deep}
\begin{equation*}
    f_i = \epsilon_{4-i}(w_{-1}, w_{-2}, w_{-3})|_{w_{-3}, w_{-2} \mapsto w_{-1}} \implies \lambda = \dfrac{w_0 \, Q_{7}}{\mu(\mu-w_{-1})^3}\,.
\end{equation*}
A flat frame is given by $t_i$ as in \eqref{eq:exnegflats} for $i \neq 0$, together with
\begin{gather*}
t_1=-\widetilde t_1,\quad
t_j=\widetilde t_{j-1}, \, \text{ for } j \in \{3,5, 6\}\quad
t_7 = \widetilde t_7, \\
t_2
=
\left(
\frac{
w_0\left(
w_{-1}^8+1+\sum_{i=1}^{7}(-1)^i w_i w_{-1}^i
\right)
}{
w_{-1}^4
}
\right)^{\frac13}, \qquad (\ref{Eq:5.1a}, \infty_i = w_{-1}, \alpha = 1),
 \\
t_4
=
\frac{
w_0^{\frac{2}{3}}
\left(
11w_{-1}^8-5+
\sum_{i=1}^{7}(2i-5)(-1)^i w_iw_{-1}^i
\right)
}{
6w_{-1}^{\frac{8}{3}}
\left(
w_{-1}^8+1+\sum_{i=1}^{7}(-1)^iw_iw_{-1}^i
\right)^{\frac{1}{3}}
}, \qquad (\ref{Eq:5.1a}, \infty_i = w_{-1}, \alpha = 2),
\end{gather*}
where $\tilde{t}_i$ denotes $t_i$ as in  Example 4.3 after imposing $w_{-2}\mapsto w_{-1}$,  $\underline{d} = \left(\frac{1}{4}, \frac{1}{3}, \frac{1}{2}, \frac{2}{3}, \frac{3}{4}, 1,1 \right)$, and the first line of coordinates are obtained from $(\ref{Eq:5.1a}, \infty_i=\infty, \alpha=1), \,  (\ref{Eq:5.1a}, \infty_i=\infty, \alpha=2),  \, (\ref{Eq:5.1a}, \infty_i=\infty, \alpha=2), \, (\ref{Eq:5.1a}, \infty_i=\infty, \alpha=3), \, (\ref{Eq:5.1c}, \infty_i=\infty)-(\ref{Eq:5.1c}, \, \infty_i=0), (\ref{Eq:5.1c}, \infty_i=0)$, respectively. Moreover, 
\begin{align}
F ={}&
-\frac{1}{4}e^{4t_{-1}-3t_0}
\left(8t_2^3+12t_2t_4+t_6\right)
+\frac{1}{6}e^{3t_{-1}-\frac{9}{4}t_0}
t_1\left(9t_2^3+18t_2t_4+2t_6\right)
\\ \notag
&\,
-\frac{1}{8}e^{2t_{-1}-\frac{3}{2}t_0}
(t_1^2+t_3)\left(2t_2^3+6t_2t_4+t_6\right)
+\frac{1}{192}e^{t_{-1}-\frac{3}{4}t_0}
(t_1^3+6t_1t_3-t_5)
\left(t_2^3+6t_2t_4+2t_6\right)
\\ \notag
&\,
-\frac{1}{96}e^{-3t_{-1}+\frac{17}{4}t_0}
(t_1^3+6t_1t_3-t_5)
-\frac{1}{2}e^{-4t_{-1}+5t_0}
(t_2^3-6t_2t_4+2t_6)
-\frac{1}{2}t_6^2\log t_2
-\frac{3t_4^2t_6}{2t_2}
+\frac{3t_4^4}{4t_2^2}
\\ \notag
&\,
+\left(-\frac{1}{2}t_{-1}+\frac{3}{8}t_0\right)t_6^2
+\frac{1}{4128768}
\left(
t_1^8-14t_1^6t_3+84t_1^4t_3^2-168t_1^2t_3^3
\right)
-\frac{1}{2949120}
\left(
t_1^5-20t_1^3t_3+60t_1t_3^2
\right)t_5\\ \notag
&\,
+\frac{1}{589824}t_1^2t_5^2
+\frac{1}{384}t_1t_5t_6
+\frac{t_2^6}{960}
-\frac{1}{24}t_2^3t_6
+\frac{3}{8}t_2^2t_4^2
-\frac{3}{2}t_2t_4t_6
+\frac{t_3^4}{49152}
-\frac{t_3^2t_6}{128}
-\frac{t_3t_5^2}{294912}
\\ \notag
&\,
+t_7\left(
\frac{1}{384}t_1t_5
-3t_2t_4
-\frac{1}{128}t_3^2
+\left(-t_{-1}+\frac{3}{4}t_0\right)t_6
+\left(\frac{3}{2}t_{-1}-\frac{17}{8}t_0\right)t_7
\right).
\end{align}
Notice the presence of rational terms in $t_2$ as conjectured by \cite{MZ24}, which were not present in Example \ref{ex:A513}. This is due to the latter being associated to a movable pole of order two only. In the next section we show that such rational dependence will arise whenever we have a movable pole in the superpotential of order strictly larger than two, or in other words, when we consider the collision of more than two movable poles in the superpotential associated with the longest partition.
\end{example}

\section{Description of the Frobenius structures on the Hurwitz boundary}\label{sec:polecollision}
We have seen that all Landau--Ginzburg superpotentials we obtain from shifting the fundamental characters in the $A_\ell$ spectral curve are of the form in \cref{eq:lambdaAmult}. In particular, for a fixed initial node $\bar{k}:=\min I\,$, different superpotentials correspond to different partitions of $p\,$, which is the number of nodes after $\bar{k}$ we are shifting. In particular, these are all rational maps of degree $\ell+1\,$. For notational convenience, in this section we fix $\hat{\ell},\bar{k},p\in\ZZ_{\geq 0}$ so that $\ell:=\hat{\ell}+\bar{k}+p$ and we denote by:
\[
\Rat_{\ell+1}:=\mathbb{P}\bigl(H^0(\RS,\Hol(\ell+1))^{\oplus 2}\bigr)\smallsetminus\{\mathrm{RES}=0\}\,\subseteq \mathbb{P}^{2\ell+3}\,,
\]
the space of rational maps of degree $\ell+1\,$, where $\mathrm{RES}$ is the resultant. This is a quasi-projective variety of dimension $2\ell+3\,$. For any fixed $\hat{\ell}, \bar{k},p$, the superpotentials in  \cref{eq:lambdaAmult}  have poles at $\infty$ and zero of order $\hat{\ell}+1$, $\bar{k}$ respectively, in addition to other poles in the complement $\RS\smallsetminus\{0,\infty\}\cong\CC^\ast$ so that their orders sum up to $p\,$. Notice that, unless $p=0\,$ (i.e. in the  Dubrovin--Zhang case), the reflection symmetry about the centre of the type $A$ Dynkin diagram does not allow one to swap the poles at $\infty$ and $0\,$. Hence, before taking the quotient by the action of the group $\PSL_2\equiv \PSL_2(\CC)$ of automorphisms of the source $\RS\,$ to get a Hurwitz space, we are looking at rational functions whose pole configuration is given by:
\[
\mathrm{PConf}_p:=\bigl\{\bigl((X,Y),D\bigr)\in\bigl((\RS)^2\smallsetminus\mathrm{diag}\bigr)\times \Div^p(\RS)\,:\, X,Y\not\in\mathrm{supp}(D)\,\bigr\}\,,
\]
where $\Div^n(\RS)$ denotes the space of degree-$n$ effective divisors on $\RS\,$. Our type $A$ generalised Dubrovin--Zhang superpotentials correspond to the degree-$(\ell+1)$ rational function whose pole divisor is $(\hat{\ell}+1)\,\infty+\bar{k}\,0+D\,$. In other words, we need to rigidify our space so that it keeps track of the positions of the first two poles. From the Hurwitz Frobenius manifold point of view, they are marked by the fact that the primary differential only has poles at them, while being holomorphic at the movable poles of the superpotential. The additional information is introduced by defining the fibre product:
\begin{equation*}
\Rat_{\hat{\ell},\bar{k}\,;\,p}:=\Rat_{\ell+1}\,\times_{\Div^{\ell+1}(\RS)}\, \mathrm{PConf}_p\,,
\end{equation*}
with respect to the maps:
\[
\begin{aligned}
\begin{aligned}
        \rm{pole}:\Rat_{\ell+1}&\to\Div^{\ell+1}(\RS)\\
    (A,B)&\mapsto (B)
\end{aligned}\,,\qquad\begin{aligned}
    \mathrm{ad}_{\hat{\ell},\bar{k}}:\mathrm{PConf}_p&\to \Div^{\ell+1}(\RS)\\
    \bigl((X,Y),\,D\bigr)&\mapsto D+(\hat{\ell}+1)\,X+\bar{k}\,Y
\end{aligned}\,.
\end{aligned}
\]
The fibre product is also a quasi-projective variety of dimension:
\[
\dim  \Rat_{\hat{\ell},\bar{k}\,;\,p}=\dim \Rat_{\ell+1}+\dim \mathrm{PConf}_p-\dim \Div^{\ell+1}(\RS)=\ell+1+p+3\,.
\]

Now, the space $\Div^p(\RS)$ can be stratified by partitions $\pi=(m_1,\dots, m_{\ell(\pi)})\vdash p\,$. In particular, for any such partition we consider the subset $\Div^p_\pi(\RS)$ of degree-$p$ effective divisors on $\RS$ of the form $D=\sum_{i=1}^{\ell(\pi)}m_i\,x_i\,$, for pair-wise distinct $x_1,\dots, x_{\ell(\pi)}\in \RS\,$. Clearly, $\Div^p(\RS)=\bigsqcup_{\pi\vdash p}\Div^p_\pi(\RS)\,$. This induces a similar partition of the pole configuration space into disjoint subsets $\mathrm{PConf}_p(\pi)\,$, and therefore of the fibre product:
\begin{equation*}
    \begin{aligned}
\Rat_{\hat{\ell},\bar{k}\,;\,p}&=\bigsqcup_{\pi\vdash p}\Rat_{\hat{\ell},\bar{k}}(\pi)\,,\\
        \Rat_{\hat{\ell},\bar{k}}(\pi)&:=\Rat_{\ell+1}\,\times_{\Div^{\ell+1}(\RS)}\, \mathrm{PConf}_p(\pi)\,.
    \end{aligned}
\end{equation*}
Each stratum is itself a quasi-projective variety with $\dim \Rat_{\hat{\ell},\bar{k}}(\pi)=\ell+1+\ell(\pi)+3\,$. As expected, the top-dimensional stratum comes from choosing the longest partition $\pi=(1^p)\,$, and it is not difficult to see that:
\begin{equation*}
    \overline{\Rat_{\hat{\ell},\bar{k}}(\pi)}=\bigsqcup_{\pi' \text{coarser than } \pi}\Rat_{\hat{\ell},\bar{k}}(\pi')\,,
\end{equation*}
where the bar denotes the Zariski closure.

Finally, the group of M\"{o}bius maps $\PSL_2$ acts by pre-composition on $\Rat_{\ell+1}\,$, and clearly on the pole configuration space. Therefore, we consider, for any $\pi\vdash p\,$, the stack quotients:
\begin{equation*}
\Hw_{\hat{\ell},\bar{k}\,;\,p}:=\bigl[\Rat_{\hat{\ell},\bar{k}\,;\,p}\,/\,\PSL_2\bigr]\,,\qquad  \Hw_{\hat{\ell},\bar{k}}(\pi):=\bigl[\Rat_{\hat{\ell},\bar{k}}(\pi)\,/\,\PSL_2\bigr]\,.
\end{equation*}
In particular, the former is the disjoint union of the latter as $\pi$ ranges over all partitions of $p\,$. Their coarse moduli spaces are also quasi-projective varieties, of dimensions reduced by $3$ with respect to the previous expressions, i.e.:
\[
\begin{aligned}
    \dim  \Hw_{\hat{\ell},\bar{k}\,;\,p}&=\ell+1+p\,,&&& \dim  \Hw_{\hat{\ell},\bar{k}}(\pi)&=\ell+1+\ell(\pi)\,.
\end{aligned}
\]
In this sense, for any fixed partition, $\Hw_{\hat{\ell},\bar{k}}(\pi)$ is a Hurwitz space -- with additional information coming from remembering the positions of the movable poles -- and  it admits Frobenius manifold structures given after choosing a suitable primary differential.  Note that, while the number of admissible differentials varies from stratum to stratum, the third kind primary differential $\phi \equiv\dd\log\mu$ is well-defined and admissible on all strata, therefore it defines a Frobenius manifold on each of them, thus confirming \cref{thm:FrobAsimple,thm:FrobAmult}.
\begin{remark}
    In the standard Hurwitz theory convention, the integers in the partition $\underline{n}$ that determine the ramification profile over $\infty$ represent the ramification indices, rather than the order of the points as poles of the meromorphic function\footnote{The ramification index is one less than the corresponding order.}. In this convention, $\pi$ is a non-negative integer partition representing the ramification indices over $\infty\,$, and its entries sum up to $p-\ell(\pi)\,$. Since the length of the partition does not change, however, neither do the dimension formulae.
\end{remark}

The top-dimensional stratum $\Hw_{\hat{\ell},\bar{k}}^{\rm{(top)}}:=\Hw_{\hat{\ell},\bar{k}}(1^p)\,$ corresponds to the case where all the movable poles are simple. The associated Frobenius structure admits an explicit description in terms of its flat coordinates \cite{PS26}. The aim of this section is to study the behaviour of the structure restricted to boundary strata, i.e under  subsequent coarsenings of the longest partition $(1^p)\,$, or in other words, under coalescence of the movable poles of \eqref{eq:lambdaAsimple} to obtain \eqref{eq:lambdaAmult}. To do so, however, requires some results from  complex and functional analysis.

%%%%%%%%%%%%%%%%%%%%%%%%%%%%%%%%%%%%%%%%%%%%%%%%%%%%%%%%%%%%%%%

\subsection{Pole collision framework}
For $m\leq p\, $, we consider a sequence of functions $\{\lambda_n\}_{n\geq 0}\subseteq \Hw_{\hat{\ell},\bar{k}}^{\rm{(top)}}$ in the top-dimensional stratum, which converges -- at this stage pointwise away from its singularities -- to a rational function $\lambda\in \Hw_{\hat{\ell},\bar{k}}(m,1^{p-m})$ in a deeper stratum in the boundary of $\Hw_{\hat{\ell},\bar{k}}^{\rm{(top)}}\,$. In particular, the limit function has a pole of order $m$ in $\CC^\ast\,$, and $p-m$ remaining simple poles. By partial fraction decomposition, in an affine chart we can uniquely write a function $\lambda_n$ in the sequence as:
\begin{equation}\label{eq:lambdanresiduepole}
    \lambda_n(z)=h(z)+\frac{r_1(n)}{z-p_1(n)}+\dots+\frac{r_m(n)}{z-p_m(n)}\,,
\end{equation}
for some mutually disjoint sequences $(p_1(n),\dots, p_m(n))\in(\CC^\ast)^m\smallsetminus\diag$ representing the colliding poles of the rational function $\lambda_n$ and $\{r_1(n),\dots, r_m(n)\}_{n\geq 0}$ being the sequences of residues. Here, $h$ is a meromorphic function on $\RS$ encoding the poles at $0\,$, $\infty$ and the remaining $p-m$ simple poles. In particular, $h$ is holomorphic at the other poles of $\lambda_n$ and is unaffected by the collision. 

Since we want the limit function to have an order-$m$ pole coming from collision of the $m$ simple poles we have pinpointed, this forces the sequences $\{p_1(n),\dots p_m(n)\}_{n\geq 0}$ to converge to the same point $p^\ast\in\CC^\ast$ as $n\to\infty\,$. The natural parametrisation of the moving simple poles of $\lambda_n$ by their locations and residues as in \cref{eq:lambdanresiduepole} is not well-defined on the boundary. Therefore, it will be convenient to introduce a better-suited coordinate system. To this end, we define:
\begin{equation}\label{eq:collisionquantities}
    \begin{aligned}
        p^\ast&:=\tfrac1m\sum_{j=1}^m p_j\,;&&&
        \Delta_i&:=p_i-p^\ast\,,&&&  
        c_i&:=\sum_{j=1}^mr_j \Delta_j^{i-1}\,,&&& i&=1,\dots, m\,.
    \end{aligned}
\end{equation}
Note that the first quantity, $p^\ast$, represents the location of the higher-order pole in the lower-dimensional stratum. We call this the  \textit{collision point}. The differences $\Delta_i\,$ measure the distance between a simple pole and the collision point for finite $n\,$, and therefore the limit stratum corresponds to their vanishing set. As the displacements $\Delta_i$ satisfy $\Delta_1+\dots+\Delta_m=0\,$,   when defining a suitable coordinate system on the top-dimensional stratum, we replace these with the elementary symmetric polynomials $e_2,\dots, e_m$ generating the coordinate ring $\CC[\Delta]^{S_m}/<\Delta_1+\dots+\Delta_m>\,$. Clearly, the lower-dimensional stratum will still be the vanishing set of these polynomials. Notice that the remaining quantities are linear combinations of the residues with coefficients given by the entries of the Vandermonde matrix $V_\Delta $ in $\Delta_1,\dots, \Delta_m\,$. 
\begin{defn}
The \textit{collision coordinates} $(p^\ast,c_1,\dots, c_m,e_2,\dots, e_m)$ in the collision cluster of size $m$ on the top dimensional stratum $\Hw_{\hat{\ell},\bar{k}}^{\rm{(top)}}$ are defined, in terms of the residue-pole affine chart, as follows: $p^\ast,c_1,\dots, c_m$ are as in \cref{eq:collisionquantities}, while $e_2,\dots, e_m$ are the elementary symmetric polynomials in $\Delta_1,\dots, \Delta_m\,$, ordered so that $\deg e_i=i\,$.
\end{defn}
\begin{prop}
    The transition function between the residue-pole coordinates and the collision coordinates is singular on the limit stratum $\Hw_{\hat{\ell},\bar{k}}(m,1^{p-m})$. In particular, if $J$ denotes the Jacobian matrix:
    \[
    \det J\propto\biggl[\prod_{1\leq i <j\leq m}\bigl(p_i-p_j\bigr)\biggr]^2
    \]
\end{prop}
\begin{proof}
      It is easy to see that:
    \[
    J=\pdv{(p^\ast,e,c)}{(p_1,\dots,p_m,r_1,\dots,r_m)}=\mqty[ \tfrac{1}{m}\mathbf{1}&0\\ E & 0\\
    V_\Delta  & \ast]\,,
    \]
    where $\mathbf{1}$ is the row with $m$ entries equal to one.     It follows that:
    \[
    \det J=\det V_\Delta \,\det \mqty[\tfrac{1}{m}\mathbf{1}&0\\ E & 0].
    \]
     The second determinant is a polynomial in $\Delta$ of degree at most $\tfrac12m(m-1)\,$, which vanishes whenever two variables coincide. It follows that it is a multiple of the Vandermonde polynomial, which proves the statement.
\end{proof}
\begin{remark}
This setup can be extended to any given collision. In particular if, from the top-dimensional stratum, we want to land in the stratum $\Hw_{\hat{\ell},\bar{k}}(\pi)\,$, for some partition $\pi=(m_1,\dots, m_{\ell(\pi)})\,\vdash p\,$, then we divide the residue-pole pairs $(p_i,r_i)$ for the simple poles into $\ell(\pi)$ subsets of size $m_1,\dots, m_{\ell(\pi)}\,$, and change coordinates to the collision chart relative to each of these subsets. The target stratum $\Hw_{\hat{\ell},\bar{k}}(\pi)$ coincides with the simultaneous vanishing set of all the elementary-symmetric-polynomial coordinates from all such subsets. The coordinate transformation splits up in blocks of size $m_1,\dots, m_{\ell(\pi)}\,$, and so does the Jacobian matrix. Hence, it is sufficient to study explicitly the case of a single collision cluster.
\end{remark}

The role of the coordinates $c_1,\dots, c_m$ is highlighted by the following key result, which also makes our intuitive discussion above precise.
\begin{thm}\label{prop:polecoalescence}
    Let $\{p_1(n),\dots, p_m(n)\}_{n\geq 0}$ be sequences of complex numbers converging to the same limit $p^\ast\in\CC\,$, and let $h$ be a meromorphic function holomorphic at $p^\ast\,$. Consider the sequence of meromorphic functions $\{\lambda_n\}_{n\geq 0}$ defined by:
\begin{equation}\label{eq:sequenceoffunctions}
    \lambda_n(z)=h(z)+\frac{r_1(n)}{z-p_1(n)}+\dots+\frac{r_m(n)}{z-p_m(n)}\,,
\end{equation}
where $r_i(n)=\Res_{p_i(n)}\{\lambda_n\dd{z}\}\,$.

Then:
\begin{enumerate}
    \item Assume that $\{\lambda_n\}_{n\geq 0}$ converges uniformly to a function $\lambda$ on compact sets not containing $p^\ast\,$ or poles of $h\,$.
Then $\lambda$ is meromorphic and of the form
    \begin{equation}\label{eq:limitfunction}
            \lambda(z)=h(z)+\sum_{k=1}^m\frac{c_k}{(z-p^\ast)^k}\,,
    \end{equation}
    where:
    \[
    \begin{aligned}
c_k=\lim_{n\to\infty}c_k(n)\equiv\lim_{n\to\infty}\sum_{i=1}^mr_i(n)\,\Delta^{k-1}_i(n)\,.
    \end{aligned}
    \]
    \item Conversely, assuming that the sequences $c_1(n),\dots, c_m(n)$ have finite limits $c_1,\dots, c_m\in\CC$ as $n\to\infty\,$ -- with $c_m\neq 0$ --
then $\{\lambda_n\}_{n\geq 0}$ converges to \cref{eq:limitfunction} uniformly on compact subsets not containing $p^\ast\,$ or poles of $h\,$.
\end{enumerate}
\end{thm}
\begin{proof}
\begin{enumerate}
    \item  Let $\gamma_r\,$ be a circle centred at $p^\ast$ with radius $r>0\,$. Since $\lambda_n$ converges uniformly to $\lambda$ on $\gamma_r\,$, we have, for any $k\in\ZZ_{\geq 1}\,$, that the principal Laurent coefficient $c_k$ of $\lambda-h$ at $p^\ast$ is given by
        \[
        \begin{aligned}
            c_k&=\tfrac{1}{i2\pi}\lim_{n\to\infty}\int_{\gamma_r}(z-p^\ast)^{k-1}\bigl(\lambda_n(z)-h(z)\bigr)\dd{z}\\
            &=\lim_{n\to\infty}\sum_{i=1}^m\Res_{p_i(n)}\bigl\{(z-p^\ast)^{k-1}\bigl(\lambda_n(z)-h(z)\bigr)\bigr\}=\lim_{n\to\infty}c_k(n)\,.
        \end{aligned}
        \]
        Now, let $D_n:=\prod_{i}(z-p_i(n))\,$. Clearly, $D_n$ is a sequence of entire functions having entire uniform limit $(z-p^\ast)^m$ on any compact set. As a consequence, each function in the sequence $\{F_n:=D_n(\lambda_n-h)\}_{n\geq 0}$ is holomorphic on any open neighbourhood of $p^\ast\,$ not containing poles of $h\,$, and converges uniformly to $(z-p^\ast)^m(\lambda-h)$ on any compact subset not containing $p^\ast$ or poles of $h\,$. Thus, by the maximum principle $\{F_n\}_{n\geq 0}$ is uniformly bounded on any open disk centred at $p^\ast\,$, and therefore, by Montel's Theorem, the limit function extends holomorphically to $p^\ast\,$. This proves that $\lambda$ has at most a pole of order $m$ at $p^\ast\,$. 
        \item Conversely, let $F_n:=D_n(\lambda_n-h)$ as in (1). Notice that:
        \[
        \begin{aligned}
            F_n(z)&=\sum_{i=1}^mr_i(n)\prod_{j\neq i}\bigl(z-p^\ast-\Delta_j(n)\bigr)\\
            &=\sum_{k=0}^{m-1}(z-p^\ast)^{m-1-k}\sum_{i=1}^m(-1)^{k} r_i e_{k}^{(i)}(\Delta) \,,
        \end{aligned}
        \]
        where $e_k^{(i)}(\Delta)$ denotes the elementary symmetric polynomial of degree $k$ in the $m-1$ variables $\Delta_1,\dots, \widehat{\Delta_i},\dots, \Delta_m\,$. It follows from Newton identities that $e_s^{(i)}=\sum_{r=0}^{s}(-1)^{r} \Delta_i^{r}e_{s-r}\,$ for $s=0,\dots,m-1\,$. Hence, $F_n$ further simplifies to:
         \[
        \begin{aligned}
            F_n(z)&=\sum_{k=0}^{m-1}(z-p^\ast)^{m-1-k}\sum_{r=0}^{k}(-1)^{k+r} e_{k-r}(\Delta)\, c_{r+1}(n)\,,
        \end{aligned}
        \]
        and only depends on the first $m$ sequences $c_1(n),\dots, c_m(n)\,$, which we have assumed to have finite limits. As a consequence, since for any $k>0\,$, $e_k(\Delta)\to 0$  as $n\to\infty\,$, $F_n$ converges uniformly on any compact set not containing poles of $h$ to the function:
        \[
        F(z):=\sum_{k=0}^{m-1}(z-p^\ast)^{m-1-k}c_{k+1}\,.
        \]
        Finally, since $D_n\to (z-p^\ast)^m$  uniformly on any compact set, it follows that $\lambda_n-h$ converges to $(z-p^\ast)^{-m}\,F(z)$ on any compact set that does not contain $p^\ast\,$ or poles of $h\,$.
\end{enumerate}
\end{proof}
 \begin{cor}\label{cor:existencesequenceuniformlimit}
        Let $p^\ast\in\CC$ and $\lambda$ be a meromorphic function of the form:
        \[
        \lambda(z)=\sum_{k=0}^m\frac{c_k}{(z-p^\ast)^k}+h(z)\,,
        \]
       with $h$ being a meromorphic function with no pole at $p^\ast\,$. There exist $m$ mutually disjoint sequences $\{p_1(n),\dots, p_m(n)\}_{n\geq 0}$ and a sequence of functions $\{\lambda_n\}_{n\geq 0}$ such that:
        \begin{itemize}
            \item $p_1(n),\dots, p_m(n)$ all converge to $p^\ast$ as $n\to\infty$,
            \item  $\lambda_n$ has simple poles at $p_1(n),\dots, p_m(n)\,$.
            \item $\lambda_n-h$ converges uniformly to $\lambda-h$ on any compact subset not containing $p^\ast\,$.
        \end{itemize}
    \end{cor}
    \begin{proof}
         Let $(\kappa_1,\dots, \kappa_m)\in\CC^m\smallsetminus\diag$ and consider the sequences $p_i(n):=p^\ast+\tfrac1n\kappa_i\,$. By \cref{prop:polecoalescence}, if 
        \[
        \lambda_n(z):=h(z)+\sum_{i=1}^m\frac{r_i(n)}{z-p_i(n)}\,,
        \]
        then, to prove the statement, it is sufficient to show that we can prescribe the residues $r_1(n),\dots, r_m(n)$ at the simple poles of $\lambda_n$ in such a way that the corresponding sequences $c_1(n),\dots, c_m(n)$ converge to the principal Laurent coefficients $c_1,\dots, c_m$ of $\lambda$. 

         Now, if $\Delta_i:=p_i-p^\ast=\tfrac1n \kappa_i\,$. Then,
        \[
        \mqty(c_1(n)\\\vdots\\c_m(n))=V_\Delta \,\mqty(r_1(n)\\\vdots\\r_m(n))=\mqty[1 & 0&\dots  &0 \\
       0 &\tfrac1n &\dots &0\\
        \vdots&\vdots&\ddots &\vdots\\
       0 &0&\dots&\tfrac{1}{n^{m-1}}]\,V_\kappa\,\mqty(r_1(n)\\\vdots\\r_m(n))\,.
        \]
       Since we chose the constants $\kappa_1,\dots, \kappa_m$ to be pairwise distinct, the Vandermonde matrix $V_\kappa$ is invertible and it suffices to set:
        \[
        r_i(n):=\sum_{j=1}^m n^{j-1}\,\bigl(V_\kappa^{-1}\bigr)_{i j}c_j\,,
        \]
        to obtain the statement.
    \end{proof}
     \begin{cor}\label{cor:convergencederivatives}
             Let $p^\ast\in\CC$ and $\lambda$ be as in \cref{cor:existencesequenceuniformlimit}, and let $\{\lambda_n\}_{n\geq 0}$ be a sequence of functions with $m$ simple poles converging to $\lambda$ uniformly on compacts not containing $p^\ast\,$. 

For any $k\in\ZZ_{\geq 0}\,$,  $\bigl\{\lambda_n^{(k)}\,,\,\partial_p\lambda_n\,,\,\partial_{c_1}\lambda_n,\dots, \partial_{c_m}\lambda_n\bigr\}_{n\geq 0}$  converge to $\bigl\{\lambda^{(k)}\,,\,\partial_p\lambda\,,\,\partial_{c_1}\lambda,\dots, \partial_{c_m}\lambda\bigr\}$ uniformly on compact subsets in $\CC\smallsetminus\{p^\ast\}\,$.
\end{cor} 
\begin{proof}
     The uniform convergence of the $z$-derivatives is a standard result that follows from Cauchy's integral formula for the derivatives.

     For derivatives with respect to $c_1,\dots, c_m\,$, we can write, in the notation of the proof of \cref{prop:polecoalescence}:
      \[
    \begin{aligned}
    \pdv{\lambda_n}{c_i}=\frac{1}{D_n(z)}\sum_{k=0}^{m-i}(-1)^{k}(z-p^\ast)^{m-i-k}\,e_{k}(\Delta)\,.
\end{aligned}
    \]
    This clearly converges pointwise to $(z-p^\ast)^{-i}=\pdv{\lambda}{c_i}$ away from the collision point. On the other hand:
    \[
    \begin{aligned}
        \pdv{\lambda_n}{c_i}-\pdv{\lambda}{c_i}&=\frac{(-1)^{m-i-1}}{D_n(z) }\sum_{k=1}^i\frac{(-1)^{k}\,e_{m-i+k}(\Delta)}{(z-p^\ast)^{k}}\,.
    \end{aligned}
    \]
    Therefore, let $r>0$ and consider the circle $\gamma_r$ of radius $r$ centred at $p^\ast\,$. For any $c\in(0,1)\,$, there exists $N\in\ZZ_{\geq 0}$ such that $\max_j\abs{\Delta_j(n)}\leq c r$ for $n\geq N\,$. It follows that, for $n\geq N\,$,
    \[
    \begin{aligned}
      \abs{D_n(z)}&\geq r^m-\sum_{k=0}^{m-2}\abs{e_{m-k}(\Delta)}r^k\\&\geq r^m\biggl(1-\sum_{k=0}^{m-2}\binom{m}{k}c^k\biggr)=(2+mc-(c+1)^m)\,r^m\,.
    \end{aligned}
    \]
Now, the function $c\mapsto 2+mc-(c+1)^m$ is differentiable, strictly decreasing for $c\geq 0\,$, and  positive at $c=0\,$. It follows that it has a unique zero $\bar{c}>0\,$. If, then, we choose any $c\in(0,\bar{c})\,$, we have the following bound on $\gamma_r\,$:
\[
\begin{aligned}
    \abs{\pdv{\lambda_n}{c_i}-\pdv{\lambda}{c_i}}\leq \frac{1}{2+cm-(c+1)^m}\sum_{k=1}^i\binom{m}{i-k}\frac{1}{r^{k+m}}\max_j\abs{\Delta_j}^{m+k-i}\,.
\end{aligned}
\]
Since, for $k=1,\dots,i\,$, $m+k-i\geq m-i+1>0\,$, then the difference between the two derivatives vanishes uniformly on the circle as  $n\to\infty\,$. Similarly,  $\bigl\{\partial_p\lambda_n\bigr\}_{n\geq 0}$ converges uniformly to $\partial_p\lambda$ away from the collision point.
\end{proof}
\begin{remark}
    \cref{cor:convergencederivatives} highlights the importance of the collision chart: it admits a subset which still constitutes a coordinate system on the target stratum, and therefore the comparison of derivatives in the limit in the tangent directions is well-defined. This is in contrast with the residue-pole frame, where the residue derivatives do not even converge in the limit.
\end{remark}

\subsection{Critical points} Since the quantities of the Frobenius structure ($\eta$, $c$) are defined in terms of residues at the critical points of the superpotential, to study the Frobenius structure on the Hurwitz boundary, it is necessary to analyse the limit behaviour of the critical points of the functions in the sequences $\{\lambda_n\}_{n\geq 0}\subseteq \Hw_{\hat{\ell},\bar{k}}^{\rm{(top)}}$. The following theorem will be useful in this endeavour.
\begin{thm}\label{thm:criticalpts}
       Let $\lambda$ be a meromorphic function on $\RS$  of the form in \cref{cor:existencesequenceuniformlimit}, and let $\{\lambda_n\}_{n\geq 0}$ be a sequence of rational functions having $m$ distinct simple poles for any $n$ which converges to $\lambda$ as $n\to\infty$ uniformly on compact sets in $\CC\smallsetminus\{p^\ast\}\,$.
       
       Then,
\begin{itemize}
    \item There are $m-1$ distinct sequences $\{s_1(n),\dots, s_{m-1}(n)\}_{n\geq 0}$ converging to $p^\ast$ such that, for any fixed $n\in\ZZ_{\geq 0}\,$, $s_1(n),\dots, s_{m-1}(n)\in\Cr_n\equiv \Cr_{\lambda_n}$ are critical points of $\lambda_n\,$.
    \item For any $q\in \mathrm{Cr}_\lambda\,$, there is a unique sequence $\{q_n\}_{n\geq 0}$ such that $q_n\in\mathrm{Cr}_n$ for every $n\geq 0$ and $q_n\to q$ as $n\to\infty\,$.
\end{itemize}
\end{thm}
\begin{proof}
Using \cref{cor:existencesequenceuniformlimit},  we let $D_n(z):=\prod_{i=1}^m(z-p_i(n))$ and $D_n^{(i)}(z):=D_n(z)\,/\,(z-p_i(n))\,$. The critical points of $\lambda_n$, $\lambda$ are respectively the zeros of:
\[
    \begin{aligned}
        P_n(z)&=\sum_{i=1}^mr_i \bigl(D_n^{(i)}(z)\bigr)^2-h'(z)\,D_n(z)^2\,,\\ Q(z)&=\sum_{k=0}^{m-1}(k+1)c_{k+1}(z-p^\ast)^{m-1-k}-h'(z)\,(z-p^\ast)^{m+1}\,.
    \end{aligned}
    \] In particular, $P_n$ and $Q$ are both holomorphic at $p^\ast\,$. For the sake of brevity, from now on we denote $w:=z-p^\ast\,$. We want to express the coefficients of $P_n$ in terms of the collision coordinates, and relate it to $Q\,$. We have
    \[
    \begin{aligned}
           \bigl(D_n^{(i)}(w)\bigr)^2 &=\sum_{k=0}^{2(m-1)}(k+1)\,\Delta_i^k\,w^{2(m-1)-k}+R^{(i)}_n(w)\,,
    \end{aligned}
    \]
    where $R^{(i)}_n$ denotes a polynomial whose coefficients are polynomials of positive degree of $e_2,\dots, e_m\,$. In particular, it converges to zero as $n\to\infty\,$. This leads to the expression:
    \[
    \begin{aligned}
        P_n(w)&=\sum_{k=0}^{m-1}(k+1)\,c_{k+1}\,w^{2(m-1)-k}-h'(p^\ast+w)\,w^{2m}+R_n(w)\\
        &=w^{m-1}Q(w)+R_n(w)\,,
    \end{aligned}
    \]
    where, again, $R_n\to 0$ as $n\to\infty\,$. 
Let $\{r_k\}_{k\geq 0}$ be a sequence of positive real numbers converging to zero as $k\to\infty\,$ and consider the corresponding sequence of circles $\{\gamma_k\}_{k\geq 0}$ of radius $r_k$ centred at $p^\ast\,$. Since $Q$ does not depend on $\Delta_1,\dots, \Delta_m\,$, there exists,  for any $k\geq 0$, $N(k)\in\ZZ_{\geq 0}$ such that:
    \[ 
    \sup_{z\in \gamma_k}\abs{R_n(z)}\leq r_k^{m-1}\sup_{z\in\gamma_k}\abs{Q(z)}\,,\qquad \forall n\geq N(k)\,.
    \]
Thus, it follows from Rouch\'{e}'s Theorem that $P_n(z)$ and $(z-p^\ast)^{m-1}Q(z)$ have the same number of zeros inside $\gamma_k\,$, counted with  multiplicity. Since $p^\ast$ cannot be a root of $Q\,$, it follows that $\gamma_k$ will contain $m-1$ critical points of $\lambda_n$ for $n\geq N(k)\,$. Similarly, one can show that any circle centred at a critical point $q$ of $\lambda$ eventually contains precisely one zero of $P_n\,$.
\end{proof}
\begin{remark}
    The critical points of the functions in the top-dimensional stratum, therefore, split into those that converge to the critical points of the limit function and those that converge to the collision point. The latter set of points consists precisely of those which control the drop in dimension between the two strata. See \cref{fig:collisionscheme} for a schematic representation.
\end{remark}

\begin{figure}[h]
\centering
\resizebox{\textwidth}{!}{%

\begin{tikzpicture}[
  scale=1.05,
  line cap=round,
  line join=round,
  % ---------- colours (edit if you like)
  contourLambda/.style={draw=cyan!80!black, line width=1.15pt},
  contourColl/.style={draw=green!70!black, line width=1.15pt},
  tend/.style={draw=black!65, line width=0.75pt, ->},
  % all points are dots; colour encodes type
  ptLambda/.style={circle, fill=cyan!80!black, inner sep=1.55pt},
  ptLn/.style={circle, fill=red!80!black, inner sep=1.55pt},
  ptPole/.style={circle, fill=black, inner sep=1.55pt},
  ptColl/.style={circle, fill=green!70!black, inner sep=1.75pt},
  lab/.style={font=\scriptsize}
]

\definecolor{colLambdaContour}{HTML}{56B4E9} % medium/light blue  (medium gray in B/W)
\definecolor{colCollContour}{HTML}{A1D99B}   % very light green   (light gray in B/W)

\definecolor{colLambdaPts}{HTML}{081D58}     % very dark navy     (near-black in B/W)
\definecolor{colLnPts}{HTML}{D55E00}         % vermillion/orange  (mid-gray in B/W)

\definecolor{colPole}{HTML}{000000}          % black
\definecolor{colCollision}{HTML}{AE017E}     % dark magenta       (dark gray in B/W)

\pgfdeclarelayer{bg}
\pgfdeclarelayer{fg}
\pgfsetlayers{bg,main,fg}

\begin{pgfonlayer}{fg}
% ============================================================
% COORDINATES
% ============================================================

% --- Region A (lambda-contour): 4 crit points of lambda
\coordinate (L1) at (-3.25,  1.10);
\coordinate (L2) at (-3.45,  0.45);
\coordinate (L3) at (-3.20, -0.10);
\coordinate (L4) at (-2.95, -0.70);

% --- Region A: 4 crit points of lambda_n that tend to those 4
\coordinate (N1) at (1,  2.5);
\coordinate (N2) at (-1.5,  0.35);
\coordinate (N3) at (0.2, -1.3);
\coordinate (N4) at (-4, 2);

% --- Region B (collision contour): collision point
\coordinate (C)  at ( 4.85,  0.05);
%\node[ptColl] at (C) {};
%\node[lab] at ($(C)+(0.25,0.10)$) {$p^\ast$};

% --- Region B: remaining 3 crit points of lambda_n, all tend to collision
\coordinate (N5) at ( 2.40,  1.2);
\coordinate (N6) at ( 3,  0.05);
\coordinate (N7) at ( 2.45, -1.5);

% --- Region B: 4 simple poles of lambda_n, all tend to collision
\coordinate (P1) at ( 5,  1.5);
\coordinate (P2) at ( 4.2,  1.7);
\coordinate (P3) at ( 5.4, 0.7);
\coordinate (P4) at (3.7, -0.50);

\node[ptLambda, label=above:{$q_1$}] at (L1) {};
\node[ptLambda, label=above:{$q_2$}] at (L2) {};
\node[ptLambda, label=right:{$q_3$}] at (L3) {};
\node[ptLambda, label=below:{$q_4$}] at (L4) {};
\node[ptLn, label=below:{$q_1^{(n)}$}] at (N1) {};
\node[ptLn, label=above:{$q_2^{(n)}$}] at (N2) {};
\node[ptLn, label=above:{$q_3^{(n)}$}] at (N3) {};
\node[ptLn, label=above:{$q_4^{(n)}$}] at (N4) {};

\node[ptPole,   label=right:{$p_1$}]    at (P1) {};
\node[ptPole,   label=left:{$p_2$}]    at (P2) {};
\node[ptPole,   label=right:{$p_3$}]    at (P3) {};
\node[ptPole,   label=left:{$p_4$}]    at (P4) {};

\node[ptColl,   label=right:{$p^\ast$}]    at (C) {};

\node[ptLn,   label=below:{$s_1^{(n)}$}]    at (N5) {};
\node[ptLn,   label=left:{$s_2^{(n)}$}]    at (N6) {};
\node[ptLn,   label=above:{$s_3^{(n)}$}]    at (N7) {};

% ============================================================
% CONTOURS
% colours: lambda contour = lambda crit colour; collision contour = collision colour
% ============================================================

% Lambda contour (blue) enclosing L1..L4 and N1..N4
\draw[contourLambda, contour antipodal reversed]
  plot[smooth cycle, tension=0.9] coordinates {
    (-4.80, 0.55)
    (-5.75, 1.85)
    (-5.10, 3.55)
    (-3.05, 3.35)
    (-1.20, 3.15)
    ( 0.20, 3.55)
    ( 1.25, 2.85)  % keep rightmost <= 1.25 so it won't hit green contour
    ( 1.05, 0.55)
    (-0.10,-2.15)
    (-1.80,-2.35)
    (-3.85,-1.85)
    (-4.70,-0.80)
  };

% Collision contour (green) 
\draw[contourColl, contour antipodal reversed] 
  plot[smooth cycle, tension=0.95] coordinates {
    (1.55, 0.20)   % left-mid
    (1.35,-0.95)   % left-lower (taken from your control)
    (2.05,-1.70)   % lower-left (taken from your control)
    (3.20,-1.85)   % bottom
    (4.55,-2.05)   % bottom-right-ish (taken from your control)
    (5.60,-1.25)   % right-lower (taken from your control)
    (5.90,-0.25)   % right-mid
    (6.20, 0.85)   % right-upper
    (5.65, 1.85)   % upper-right
    (4.25, 2.15)   % top
    (2.55, 2.55)   % upper-left-ish (taken from your control)
    (1.55, 1.55)   % left-upper
  };

% ============================================================
% POINTS (all dots)
% ============================================================

% lambda crit points (blue)
\node[ptLambda] at (L1) {};
\node[ptLambda] at (L2) {};
\node[ptLambda] at (L3) {};
\node[ptLambda] at (L4) {};

% lambda_n crit points (red)
\node[ptLn] at (N1) {};
\node[ptLn] at (N2) {};
\node[ptLn] at (N3) {};
\node[ptLn] at (N4) {};
\node[ptLn] at (N5) {};
\node[ptLn] at (N6) {};
\node[ptLn] at (N7) {};

% simple poles (orange)
\node[ptPole] at (P1) {};
\node[ptPole] at (P2) {};
\node[ptPole] at (P3) {};
\node[ptPole] at (P4) {};

% collision point (purple)
\node[ptColl] at (C) {};\end{pgfonlayer}

\begin{pgfonlayer}{bg}
% ============================================================
% TENDING-TO ARROWS (only these arrows; no extra streamlines)
% ============================================================

% 4 crit points of lambda_n -> 4 lambda crit points (in lambda contour)
\draw[tend,-,midarrow=0.55] (N1) .. controls (-2.55,0.7) .. (L1);
\draw[tend,-,midarrow=0.55] (N2) .. controls (-2.65,0.6) .. (L2);
\draw[tend,-,midarrow=0.7] (N3) .. controls (-2.65,-0.5) .. (L3);
\draw[tend,-,midarrow=0.55] (N4) .. controls (-3.5,-1) .. (L4);

% 3 crit points of lambda_n -> collision point
\draw[tend,-,midarrow=0.55] (N5) .. controls (4.05,1.05) .. (C);
\draw[tend,-,midarrow=0.55] (N6) .. controls (4.00,0.10) .. (C);
\draw[tend,-,midarrow=0.55] (N7) .. controls (4.05,-0.65) .. (C);

% 4 simple poles -> collision point
\draw[tend,-,midarrow=0.55] (P1) .. controls (4.55,0.55) .. (C);
\draw[tend,-,midarrow=0.55] (P2) .. controls (4.60,0.15) .. (C);
\draw[tend,-,midarrow=0.7] (P3) .. controls (5.55,-0.40) .. (C);
\draw[tend,-,midarrow=0.55] (P4) .. controls (4.55,-0.55) .. (C);
\end{pgfonlayer}

% ============================================================
% LEGEND
% ============================================================
\begin{scope}[shift={(7,2)}] % move as needed (x,y)
  % background box (optional) — adjust height for number of lines
  \draw[black!20, fill=white, rounded corners=2pt]
    (-0.25,-1.1) rectangle (3,0.3);

  % line spacing
  \def\dy{0.42}

  % Row 1
  \node[ptLambda] at (0,0) {};
  \node[lab, anchor=west] at (0.30,0) {Critical points of $\lambda$};

  % Row 2
  \node[ptLn] at (0,-1*\dy) {};
  \node[lab, anchor=west] at (0.30,-1*\dy) {Critical points of $\lambda_n$};

  % Row 3
  \node[ptPole] at (0,-2*\dy) {};
  \node[lab, anchor=west] at (0.30,-2*\dy) {Simple poles of $\lambda_n$};

\end{scope}
\end{tikzpicture}}
\caption{Schematic representation of the behaviour of the simple poles and critical points of the sequence of functions $\{\lambda_n\}_{n\geq 0}$ in the collision limit of $m=4$ simple poles. The contours enclose all the points that converge to the critical points of the limit function or to the collision point respectively. Integrating over the blue contour gives the three-point function components in the limit, as we shall discuss briefly.}\label{fig:collisionscheme}
\end{figure}
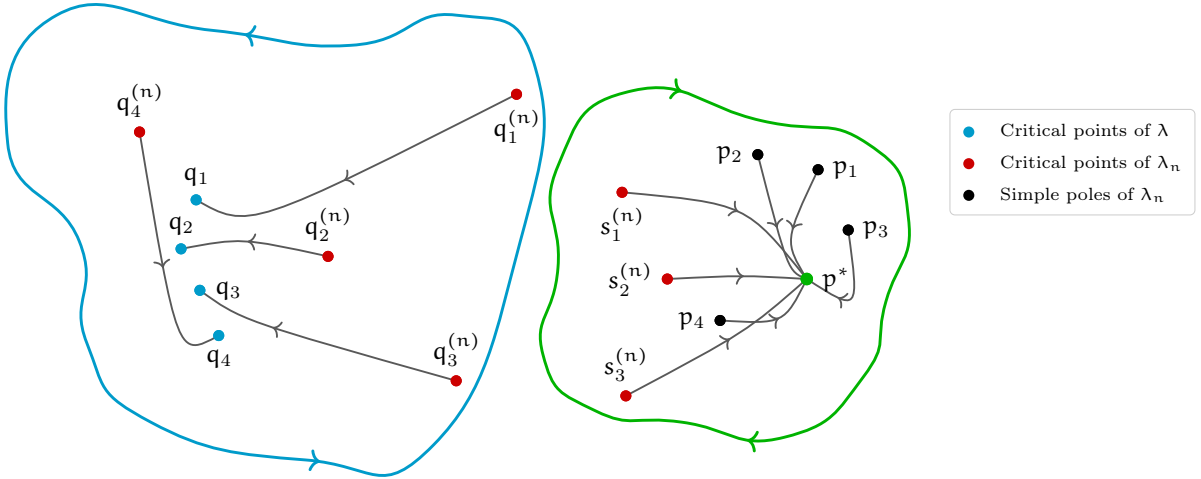

\begin{cor}\label{cor:asymptoticinnerCP}
       Let $\lambda$ and $\{\lambda_n\}_{n\geq 0}$ be as in \cref{thm:criticalpts}.
    Denote by $\delta_n:=\max_i\abs{\Delta_i(n)}\,$. The sequences $\{\kappa_1(n),\dots, \kappa_m(n)\}_{n\geq 0}$ defined by $\Delta_i(n)=\delta_n \,\kappa_i(n)\,$ all admit a converging subsequence to pairwise distinct complex numbers $(\kappa_1,\dots, \kappa_m)\in(\CC^\ast)^m\smallsetminus\diag$.

    Furthermore, if $\{s_1(n),\dots, s_{m-1}(n)\}_{n\geq 0}$ denote the sequences of critical points of $\lambda_n$ which converge to $p^\ast\,$, then:

    \begin{itemize}
        \item If $m>2\,$, then, generically,
        \[
        s_i(n)=p^\ast+\xi_i\,\delta_n+\order{\delta_n^2}\,,\qquad i=1,\dots, m-1\,,
        \]
        where $\xi_1,\dots, \xi_{m-1}$ are the critical points of $\widetilde{D}(w):=\prod_{i=1}^m(w-\kappa_i)\,$;
        \item If $m=2\,$ then,
        \[
        s_n=p^\ast-\tfrac{c_1}{2c_2}\delta_n^2+\order{\delta_n^4}\,.
        \]
    \end{itemize}
\end{cor}
\begin{proof}
The fact that the sequences $\{\kappa_1(n),\dots,\kappa_m(n)\}_{n\geq 0}$ admit subsequences that converge follows from the obvious bound $\abs{\kappa_i(n)}\leq 1\,$. In particular, the sequences are contained in a compact set.

    For the asymptotic expansion of the critical points which converge to $p^\ast\,$, suppose for the moment that $h=0\,$. In the notation of the proof of \cref{thm:criticalpts}, we look for the Puiseux expansion of the roots of the equation:
    \[
   P_n(w)= \sum_{i=1}^mr_i(n)\,\bigl(D_n^{(i)}(w)\bigr)^2=0\,,
    \]
    near $(w,\delta_n)=(0,0)\,$. 
    
    To compute the leading order term in the expansion, we draw the Newton diagram $\Lambda\subseteq \ZZ_{\geq 0}^2$ of $P_n\,$.     Each polynomial $(D_n^{(i)})^2$ is homogeneous of degree $2(m-1)\,$. On the other hand,  $r_i(n)=\sum_{j}(V_\Delta^{-1})_{ij}c_j(n)\,$. The sequences $\{c_1(n),\dots, c_m(n)\}_{n\geq 0}$ all have finite limits $c_1,\dots, c_m\,$, according to \cref{prop:polecoalescence}, with $c_m\neq 0\,$, and the entry $\bigl(V_\Delta^{-1}\bigr)_{ij}$ is a homogeneous rational function in $\Delta_1,\dots, \Delta_m$ of degree $1-j\,$. As a consequence, the coefficient of $c_j(n)$ in $P_n$ is a homogeneous polynomial in $w,\Delta_1,\dots, \Delta_m$ of degree $2m-1-j\,$. Thus, $\Lambda$ is made up of integral points on parallel lines $i+j=2m-1-j$ for $j=1,\dots,m\,$.
   
   Now, since $e_1(\Delta)=0\,$, for $m>2$    the Newton polygon will only have one edge consisting of the segment    joining $(m-1,0)$ and $(0,m-1)\,$. On the other hand, when $m=2\,$, $(m-1,0)=(1,0)\notin\Lambda\,$.   As a consequence, the unique edge of the Newton polygon has slope $-2\,$.     For $m>2\,$, therefore, the leading term in the Puiseux expansion of a root of $P_n(w)$ around $(0,0)$ is of order $\delta_n\,$, while for $m=2$ it is of order $\delta_n^2\,$. \cref{fig:Newtondiagramcollision} shows the Newton diagram for $m=2,4\,$. 

    The leading coefficient of the expansion is given by a root of the edge polynomial. Using the final expression for $F_n$ in the proof of \cref{prop:polecoalescence}, it is easy to see that $\lambda_n=S_n+c_m\,/\,D_n\,$, where the numerator of $S_n(w)$ only depends on $c_1(n),\dots, c_{m-1}(n)\,$.    As a consequence:
    \[
    P_n(w)=-D_n(w)^2\,\lambda'_n(w)=c_m(n)\,D_n'(w)-D_n^2(w)\,S_n'(w)\,.
    \]
   It follows that    the edge polynomial is the $w$-derivative of $\widetilde{D}(w)=\prod_i(w-\kappa_i)\,$. Generically, therefore, we have $m-1$ distinct branches converging at the singular point $(w,\delta_n)=(0,0)\,$ with leading-order coefficient given by a simple critical point of $\widetilde{D}\,$. This proves the first part of the statement for $h=0\,$.

    For $m=2\,$, on the other hand,   the edge polynomial of the only edge of the Newton polygon is $2c_2\xi+c_1\,$, which has unique root $\xi=-\tfrac{c_1}{2c_2}\,$.  
    
    If $h\neq 0\,$, then $P_n(w)=\sum_ir_i(n)\bigl(D_n^{(i)}(w)\bigr)^2-h'(p^\ast+w)\,D_n(w)^2\,$. Since $h'$ does not have a pole at $w=0\,$, it follows that the contribution of the second term consists of monomials of the form $w^j\delta_n^i$ with $i+j\geq 2m\,$. In particular, this will change the shape of $\Lambda$ in the integer lattice, but will not affect the lower boundary of the convex hull or the edge polynomial.
\end{proof}
\begin{figure}[h]
\centering
\resizebox{0.7\textwidth}{!}{%
\begin{tikzpicture}[scale=0.95]

% ==========================================================
% Panel (a): m=2  (x=q, y=p)
% Support in (p,q): (2,0),(1,0),(0,2)
% Support in (q,p): (0,2),(0,1),(2,0)
% Lower hull in (q,p): (0,1)->(2,0)   (no vertical edge)
% ==========================================================
\begin{scope}[xshift=0.2cm]
  % axes
  \draw[->] (-0.4,0) -- (3.0,0) node[below] {$i$};
  \draw[->] (0,-0.4) -- (0,3.0) node[left]  {$j$};

  % all lattice points in 0<=p,q<=2 (plot at (q,p))
  \foreach \p in {0,1,2}{
    \foreach \q in {0,1,2}{
      \fill[gray!18] (\q,\p) circle (1.1pt);
    }
  }
  \draw[very thin,gray!12] (0,0) grid (2,2);

  % support points (q,p): (0,2),(0,1),(2,0)
  \foreach \Q/\P in {0/2,0/1,2/0}{
    \fill[black] (\Q,\P) circle (2.0pt);
  }

  % Newton polygon (lower convex hull in (q,p)): (0,1)->(2,0)
  \draw[thick,blue] (0,1) -- (2,0);

  % labels for hull lattice points (now as (q,p))
  \node[left]  at (0,1) {\small $(0,1)$};
  \node[below] at (2,0) {\small $(2,0)$};

  % panel title
  \node[font=\small] at (1.2,2.75) {(a)\ $m=2$};
\end{scope}

% ==========================================================
% Panel (b): m=4  (x=q, y=p)
% Support S_4 given previously in (p,q); swap to (q,p) for plotting.
% Lower hull in (q,p): (0,3)->(3,0)->(6,0)
% Hull lattice points include (1,2),(2,1) and (4,0),(5,0)
% ==========================================================
\begin{scope}[xshift=5.6cm]
  % axes
  \draw[->] (-0.6,0) -- (7.4,0) node[below] {$i$};
  \draw[->] (0,-0.6) -- (0,7.4) node[left]  {$j$};

  % all lattice points in 0<=p,q<=6 (plot at (q,p))
  \foreach \p in {0,...,6}{
    \foreach \q in {0,...,6}{
      \fill[gray!18] (\q,\p) circle (0.8pt);
    }
  }
  \draw[very thin,gray!12] (0,0) grid (6,6);

  % support points S_4 (listed here already in (q,p) form)
  \foreach \Q/\P in {
    0/6,0/5,0/4,0/3,        % from (p,q) = (6,0),(5,0),(4,0),(3,0)
    2/4,2/3,2/2,            % from (p,q) = (4,2),(3,2),(2,2)
    3/3,3/2,                % from (p,q) = (3,3),(2,3)
    4/2,4/1,                % from (p,q) = (2,4),(1,4)
    5/1,                    % from (p,q) = (1,5)
    6/0,5/0,4/0,3/0,        % from (p,q) = (0,6),(0,5),(0,4),(0,3)
    2/1                     % from (p,q) = (1,2)
  }{
    \fill[black] (\Q,\P) circle (1.8pt);
  }

  % Newton polygon (lower convex hull in (q,p)): (0,3)->(3,0)->(6,0)
  \draw[thick,blue] (0,3) -- (3,0) -- (6,0);

  % label ALL lattice points on the hull (as (q,p))
  % slanted edge: (0,3)->(3,0)
  \node[left]  at (0,3) {\small $(0,3)$};
  \node[left] at (2,1) {\small $(2,1)$};
  \node[below] at (3,0) {\small $(3,0)$};

  % horizontal edge: (3,0)->(6,0)
  \node[below] at (4,0) {\small $(4,0)$};
  \node[below] at (5,0) {\small $(5,0)$};
  \node[below] at (6,0) {\small $(6,0)$};

  % panel title
  \node[font=\small] at (3.2,7.05) {(b)\ $m=4$};
\end{scope}

\end{tikzpicture}
}

\caption{\textbf{Newton diagrams for $m=2,4$ for $h=0\,$.} 
Lattice points $(i,j)$ correspond to monomials $w^j\delta_n^i\,$. The blue lines represent the Newton polygon boundary. }\label{fig:Newtondiagramcollision}
\end{figure}
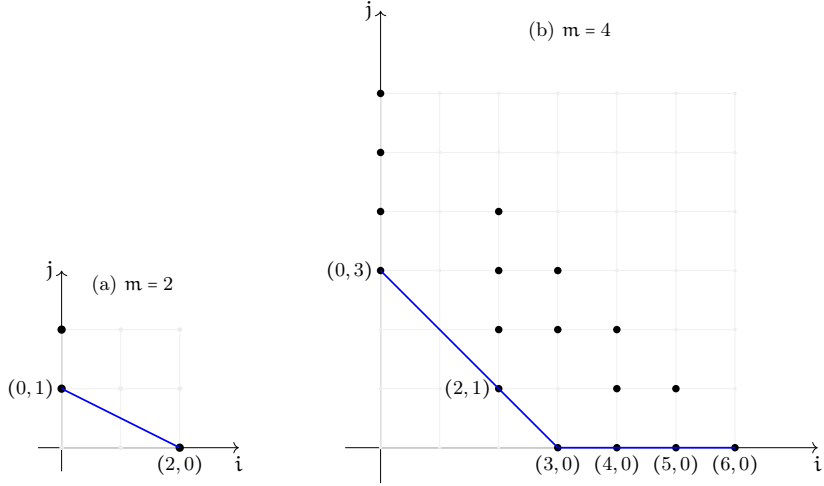

\begin{cor}
    Let $\lambda$ and $\{\lambda_n\}_{n\geq 0}$ be as in \cref{thm:criticalpts}. Let us denote by $\mathrm{CV}_n$ the set of critical values of $\lambda_n\,$, and by $\mathrm{CV}$ the critical values of $\lambda\,$.
    \begin{itemize}
        \item For any $u\in \mathrm{CV}\,$, there is a sequence $\{u_n\}_{n\geq 0}$ convergent to $u$ such that $u_n\in\mathrm{CV}_n$ for any $n\in\ZZ_{\geq0}\,$.
        \item There are $m-1$ disjoint sequences $\{v_1(n),\dots, v_{m-1}(n)\}_{n\geq 0}$ with $v_i(n)\in\mathrm{CV}_n$ for any $n\in\ZZ_{\geq0}$ and they diverge to the order $\delta_n^{-m}$ as $n\to\infty\,$.
    \end{itemize}
\end{cor}
\begin{proof}
    The only non-trivial consequence of \cref{thm:criticalpts} is that the divergent sequences are of order $\delta_n^{-m}$ as $n\to\infty\,$. This follows from evaluating $\lambda_n$ onto the critical points that converge to the collision point and using \cref{cor:asymptoticinnerCP}.
\end{proof}
\begin{remark}
    Notice the difference between our limit strata and caustics or discriminants of semi-simple Frobenius manifolds. The latter are natural submanifolds, as described in \cite{Str04}. On the other hand, we have a boundary stratification of a compactification of the Hurwitz space by partitions. As such, our point of view is closer in spirit to e.g. that of \cite{Lan19}, albeit the compactification in question is not the stable-map compactification of  Kontsevich. 
\end{remark}
\subsection{Metric and three-point function}\label{subs:metric3ptfunction}
In view of the previous subsection, the non-trivial part of comparing Frobenius manifold structures at the Hurwitz boundary arises from the fact that some critical points of the superpotential in the higher-dimensional stratum collapse into the collision point, rather than converging to the critical points of the lower-dimensional stratum. In particular, the components of the metric and of the three-point function may diverge.
Moreover, if they converge in the limit, they will in general not converge to the
corresponding components in the target stratum, even in directions tangent to the boundary. However, our discussion gives a concrete way to renormalise such components so that they do indeed converge to the desired result.
\begin{cor}\label{lem:expansioncritical}
       Let $\lambda$ and $\{\lambda_n\}_{n\geq 0}$ be as in \cref{thm:criticalpts}.

Denote by $\mathrm{Cr}^{\rm{(in)}}_n$ the critical points of $\lambda_n$ that converge to $ p^\ast$ as $n\to\infty\,$, and $\mathrm{Cr}^{\rm{(out)}}_n$ the critical points of $\lambda_n$ that converge to critical points of $\lambda$ in the limit. Let $\phi$ be a meromorphic differential on $\mathbb{P}^1$ such that $(\phi)+(\dd h)_\infty\geq 0\,$.
Define:
\[
\begin{aligned}
\phi_{ijk}&:=\lambda_i\lambda_j\lambda_k\tfrac{\phi^2}{\dd\lambda}\,,\\
    \phi_{ijk}^{(n)}&:=(\lambda_n)_i(\lambda_n)_ j(\lambda_n)_k\tfrac{\phi^2}{\dd\lambda_n}\,,
\end{aligned}
\]
where $\lambda_\alpha\,$, $(\lambda_n)_\alpha$ are the derivatives of $\lambda$ and $\lambda_n$, respectively, with respect to the $\alpha^{\rm{th}}$ collision coordinate in the limit stratum. Then,
\[
\begin{aligned}
    \sum_{x\in\mathrm{Cr}_\lambda}\Res_x\phi_{ijk}&=\lim_{n\to\infty}\biggl[\sum_{x\in\mathrm{Cr}_n}\Res_x\phi_{ijk}^{(n)}-\sum_{x\in\mathrm{Cr}_n^{\rm{(in)}}}\Res_x\phi_{ijk}^{(n)}\biggr]\,.
\end{aligned}
\]
\end{cor} 
\begin{proof}
Since $\phi_{ijk}=\lim_n\phi_{ijk}^{(n)}$ pointwise and, in view of \cref{cor:convergencederivatives}, one can exchange the limit with the integral sign, it is immediate to see that:
    \[
    \begin{aligned}
\sum_{x\in\mathrm{Cr}_\lambda}\Res_x\phi_{ijk}&=\lim_{n\to\infty}\sum_{x\in\mathrm{Cr}_n^{\rm{(out)}}}\Res_x\phi_{ijk}^{(n)}\\
&=\lim_{n\to\infty}\biggl[\sum_{x\in\mathrm{Cr}_n}\Res_x\phi_{ijk}^{(n)}-\sum_{x\in\mathrm{Cr}_n^{\rm{(in)}}}\Res_x\phi_{ijk}^{(n)}\biggr]\,.
    \end{aligned}
    \]
\end{proof}
Therefore, geometrically, we have the following picture. For any fixed $\lambda$ in the stratum $\Hw_{\hat{\ell},\bar{k}}(m,1^{p-m})\,$, we construct a sequence of functions $\{\lambda_n\}_{n\geq 0}\subseteq \Hw_{\hat{\ell},\bar{k}}^{\mathrm{(top)}}$ that converges to $\lambda$ uniformly on compact sets that do not include its poles. Concretely, this sequence of functions depends on pairwise distinct non-zero complex numbers $\kappa_1,\dots, \kappa_m$ that determine the locations of the simple poles of $\lambda_n\,$. Different choices of such parameters will determine different sequences converging to the same point in $\Hw_{\hat{\ell},\bar{k}}(m,1^{p-m})\,$. For this reason, we shall call the parameters $\underline{\kappa}\in(\CC^\ast)^m\smallsetminus\diag$ the \textit{collision shape}.

For a fixed collision shape $\underline{\kappa}\,$, we have the maps $\bigl\{\iota_n:\Hw_{\hat{\ell},\bar{k}}(m,1^{p-m})\hookrightarrow\Hw_{\hat{\ell},\bar{k}}^{\mathrm{(top)}}\bigr\}_{n\geq 0} $ sending $\lambda$ to $\lambda_n\,$, and we can pull-back the Frobenius metric and three-point function on the top-dimensional stratum along these maps, thus giving families of tensor fields $\{\iota_n^\ast\eta^{\mathrm{(top)}}\,,\,\iota_n^\ast c^{\mathrm{(top)}}\}_{n\geq 0}$ on the limit stratum. Using that $\Cr_n=\Cr_n^\out\sqcup \Cr_n^{\inn}\,$, we can write:
\[
\begin{aligned}
    \iota_n^\ast \eta^\ttop&=\eta^\out_n+\eta^\inn_n\,,&&& \iota_n^\ast c^\ttop&=c^\out_n+c^\inn_n\,.
\end{aligned}
\]
The superscript $\out$, $\inn$ denotes the sum over the critical points of $\lambda_n$ in $\Cr^\out_n$, $\Cr^\inn_n$, respectively, and the equality is to be understood at the level of symmetric tensor fields. In particular, $\eta^\out_n$ and $\eta^\inn_n$ can be degenerate. From \cref{lem:expansioncritical},  the components of the $\out$ tensors in the collision chart admit a finite limit as $n\to\infty\,$, which coincides with the Frobenius metric and three-point function on the lower-dimensional stratum. Schematically, we write:
\[
\lim_{n\to\infty}\eta^\out_n=\eta^\low\,,\qquad \lim_{n\to\infty}c^\out_n=c^\low\,.
\]
We shall say that a collision-chart component of $\eta^\low$, $c^\low$ needs to be \textit{renormalised} if it differs from the limit of the corresponding component of $\iota_n^\ast\eta^\ttop$, $\iota_n^\ast c^\ttop$, respectively. Equivalently, a component needs renormalisation if and only if the corresponding inner-tensor component has a non-vanishing limit as $n\to\infty\,$.
\begin{prop}\label{lem:divergencectensorcomponents}
    Let us denote components of a tensor field $T$ in the collision chart on the target stratum by $T_{\ast ij}:=T_{p^\ast c_ic_j}\,$. The only components of $\iota^\ast_n c^\ttop$ that may require renormalisation as $n\to\infty$ are the ones that contain at least two Latin indices. Equivalently:
    \[
    \begin{aligned}
        \lim_{n\to\infty}(c^\inn_n)_{\ast\ast\ast}=\lim_{n\to\infty}(c^\inn_n)_{\ast\ast i}=0\,,
    \end{aligned}
    \]
    for all $i=1,\dots, m\,$.
\end{prop}
\begin{proof}
Since an inner critical point $s_n\in\Cr^\inn_n$ is generically a simple zero of $\dd\lambda_n\,$,  around $s_n$ we have
\[
\phi_{ijk}^{(n)}=\biggl[\varphi^2(s_n)\,\eval{\frac{(\lambda_n)_i(\lambda_n)_j(\lambda_n)_k}{\lambda_n''}}_{s_n}\frac{1}{\zeta}+\order{1}\biggr]\dd{\zeta}\,,
\]
where $\zeta$ is a local coordinate centred at $s_n$ and $\phi=\varphi\,\dd{\zeta}\,$. The formula clearly holds even if $i$ $j$ or $k$ are equal to $\ast\,$. As a consequence, a rough estimate for each component of $c^\inn$ is given by estimating the derivatives of $\lambda_n$ at the inner critical points. This will give a lower bound for the leading-order term in the corresponding component as $n\to\infty\,$.

We start from the $c_i$-derivative, which can be written as:
\[
\pdv{\lambda_n}{c_i}=\sum_{j=1}^m\pdv{r_j}{c_i}\frac{1}{z-p_j}=\delta_n^{1-i}\sum_{j=1}^m(V_\kappa^{-1})_{ij}\frac{1}{z-p_j}\,.
\]
At $s_n\,$, using \cref{cor:asymptoticinnerCP}, we find that, for $m>2\,$, $\partial_{c_i}\lambda_n\lvert_{s_n}$ diverges at order $\delta_n^{-i}$ as $n\to\infty\,$.  

As for the $p^\ast$-derivative, since  $\lambda_n'=h'-\partial_{p^\ast}\lambda_n$ and $s_n$ is a zero of $\lambda_n'\,$, it follows that $\partial_{p^\ast}\lambda_n\lvert_{s_n}$ converges to a non-zero constant as $n\to\infty\,$.

As for the second $z$-derivative:
\[
\lambda_n''=h''(z)+2\sum_{i=1}^m\frac{r_i}{(z-p_i)^3}=h''(z)+2\sum_{i,j=1}^m\delta_n^{1-j}\bigl(V^{-1}_\kappa\bigr)_{ij}\frac{c_j}{(z-p_i)^3}\,.
\]
At $s_n\,$, this blows up like $\delta_n^{-(m+2)}$ as $n\to\infty\,$.

Hence, we have the estimates:
\[
\begin{aligned}
    \Res_{s_n}\phi_{ijk}^{(n)}&\sim \delta_n^{m+2-(i+j+k)}\,,&&& \Res_{s_n}\phi_{\ast jk}^{(n)}&\sim\delta_n^{m+2-(j+k)}\,,\\
    \Res_{s_n}\phi^{(n)}_{\ast \ast k}&\sim\delta_n^{m+2-k}\,, &&& \Res_{s_n}\phi_{\ast\ast\ast}^{(n)}&\sim \delta_n^{m+2}\,. 
\end{aligned}
\]
In particular, only the residues in the top line can diverge. 

When $m=2\,$, we can write
\[
\lambda_n(w)=h(p^\ast+w)+\frac{c_1 w+c_2}{w^2-\Delta^2\,},
\]
where $w$ is a coordinate centred at $p^\ast\,$. It follows that $\partial_{c_1}\lambda_n$ converges when computed at the inner critical point, $\partial_{c_2}\lambda_n$ diverges at order $\Delta^{-2}$ and $\lambda'_n$ at order $\Delta^{-4}\,$. Hence, the only components of $c^\inn_n$ that could have a non-finite limit as $n\to\infty$ are $(c^\inn_n)_{122}$ and $(c^\inn_n)_{222}\,$. The statement, therefore, holds even in this case.
\end{proof}

\subsection{Prepotential}
Let us now fix $h$ to be a type $A$ Dubrovin--Zhang superpotential -- and, if $m<p\,$, it contains the additional simple poles that we are not collapsing -- and $\phi=\dd\log\mu$ the corresponding primary differential. Once we have established which three-point function components may require renormalisation, we can compute the actual $c$-tensor in the top-dimensional stratum in the collision chart using the explicit formula for its prepotential in \cite{PS26}. In particular, we write the prepotential on $\Hw_{\hat{\ell},\bar{k}}^\ttop$ as follows:
\begin{equation}\label{eq:Ftop}
    F^\ttop=F_{\log}+F_{\mathrm{lin}}\,,
\end{equation}
where:
\begin{equation}\label{eq:Flog}
    \begin{aligned}
        F_{\log}(\alpha,\beta)&=\tfrac12 \sum_{i=1}^m\alpha_i^2\log(\alpha_i e^{\beta_i})+\sum_{1\leq i<j\leq m}\alpha_i\alpha_j\log(e^{\beta_i}-e^{\beta_j})\,,
    \end{aligned}
\end{equation}
while $F_{\mathrm{lin}}$ is the part that depends at most linearly on the $\alpha$-variables. The coordinate transformation between the flat $(\alpha,\beta)$ coordinates and the residue-pole coordinates we have been using so far is:
\begin{equation*}
    \begin{aligned}
        p_i&=:e^{\beta_i}\,,&&& r_i&=:\alpha_ie^{\beta_i}\,.
    \end{aligned}
\end{equation*}
Therefore, the contribution $F_{\mathrm{lin}}$ -- when written in residue-pole coordinates -- will still be linear in the residues. The importance of $F_{\mathrm{lin}}$ is highlighted by Proposition \ref{cor:generalpolecollapselinear}, which is preceded by a technical lemma. 
\begin{lem}[Polarised Newton's formulae]\label{lem:polarisednewton}
Let $n\in\ZZ_{\geq 1}$ and consider the polarised power sums $\Theta_r\in\CC[x,y]^{S_n}$ in the $2n$ variables $x=(x_1,\dots, x_n)\,$ and $y=(y_1,\dots, y_n)$ that are linear in the first set of variables: $\Theta_r:=\sum_{i=1}^nx_i y_i^r\,$. For any $r\in\ZZ_{\geq 1}$, we have
\[
    \Theta_{n+r}(x,y)+\sum_{i=0}^{n-1}(-1)^{n-i}e_{n-i}(y)\,\Theta_{r+i}(x,y)=0\,,\]
    where $e_1(y),\dots, e_n(y)$ denote the elementary symmetric polynomials in the second set of variables, ordered such that $\deg e_i=i\,$. 

    In particular, $\Theta_r$ is a linear combination of the preceding $n$ such polynomials, with coefficients in $\CC[y]^{S_n}\,$.
\end{lem}
\begin{proof}
One can write the degree-$n$ monic polynomial vanishing at $y_1,\dots, y_n$ as:
\[
P(t)=\prod_{i=1}^n(t-y_i)=t^n+\sum_{i=0}^{n-1}(-1)^{n-i} e_{n-i}t^i\,.
\]
Therefore, $y_j^n+\sum_{i=0}^{n-1}(-1)^{n-i} e_{n-i}y_j^i=0$ for any $j=1,\dots, n\,$. Multiplying by $x_j y_j^r$ and summing over $j$ gives the statement.
\end{proof}

\begin{prop}\label{cor:generalpolecollapselinear}
           Let $\lambda$ be a meromorphic function on $\RS$, as in  \cref{cor:existencesequenceuniformlimit}, and let $\{\lambda_n\}_{n\geq 0}$ be a sequence of rational functions having $m$ distinct simple poles for any $n$ which converges to $\lambda$ as $n\to\infty$ uniformly on compact sets in $\CC\smallsetminus\{p^\ast\}\,$.
 
    Let $\Phi$ be a holomorphic function having $\ell^1$ Taylor coefficients at $p^\ast$ -- i.e. $\Phi(z)=\sum_{k\geq 0}\varphi_k(z-p^\ast)^k$ in a disk centred at $p^\ast\,$, and assume that the series $\sum_{k\geq 0}\varphi_k$ is absolutely convergent. Then:
    \[
    \lim_{n\to\infty}\sum_{i=1}^m r_i \,\Phi(p_i)=\sum_{i\geq 0} \varphi_{i} c_{i+1}\,,
    \]
    where $\{c_i\}_{i\in\ZZ_{\geq 1}}$ are the principal Laurent coefficients of $\lambda$ around $p^\ast\,$.
\end{prop}
\begin{proof}
    For finite, sufficiently large $n\,$, with all the finite simple poles of $\lambda_n$ lying within the disk around $p^\ast$ where the Taylor expansion of $\Phi$ converges, we have:
    \[
    \sum_{i=1}^m r_i \,\Phi(p_i)=\sum_{j\geq 0}\varphi_j\sum_{i=1}^m r_i \Delta_i^j=\sum_{j\geq 0}\varphi_j\, c_{j+1}(n)\,.
    \]
    Thanks to \cref{prop:polecoalescence}, $\lim_{n\to\infty} c_i(n)=c_{i}\,$, implying that the statement follows if we can exchange the order of the limit and the series.
   Hence, by the dominant convergence Theorem, we look to bound $ \xi_i(n):=\abs{\varphi_i\,c_{i+1}(n)}$ by the general term $g_i$ -- crucially, independent of $n$ -- of an absolutely convergent series $\sum_{i\geq 0}g_i\,$. 

   According to \cref{prop:polecoalescence}, the first $m$ sequences $\{c_1(n),\dots, c_m(n)\}_{n\geq 0}$ converge to the Laurent coefficients of $\lambda\,$ -- in particular, they are bounded -- while the higher-degree ones converge to zero. Using the relation in \cref{lem:polarisednewton}, we have, for  positive $i\,$:
\[
\abs{c_{m+i}(n)}\leq\sum_{j=0}^{m-1}\abs{e_{m-j}}\,\abs{c_{i+j}(n)} \leq C \sum_{j=0}^{m-1}\delta_n^{m-j}\,\abs{c_{i+j}(n)}\,,
\]
for some $C>0\,$. Here, $\delta_n:=\max_i\abs{\Delta_i(n)}$ is a positive sequence that converges to zero. In particular, $\abs{c_i(n)}$ for $i>m$ is bounded by a linear combination of the previous $m$ such sequences, with coefficients that vanish in the limit. Therefore, it is easy to prove inductively that, for any $i> m\,$:
\[
\abs{c_i(n)}\leq C'\,\delta_n^{i+1-m}\,.
\]
Fix, then, $\varepsilon\in(0,1)$ and let $N\in\ZZ_{\geq 1}$ such that $\delta_n<\varepsilon$ for $n\geq N\,$. Then, for $n\geq \max\{N,m-1\}\,$, we have the uniform bound:
\[
\xi_i(n)=\abs{\varphi_i\,c_{i+1}(n)}<C'\,\abs{\varphi_i}\varepsilon^{i+1-m}\,.
\]
Under the assumption that $\sum_{i\geq 0}\varphi_i$ absolutely converges, the ratio test gives that the series with general term $g_i:= \abs{\varphi_i}\varepsilon^{i+1-m}>0$ converges, which proves the statement.
\end{proof}
\begin{remark}
    In particular, only $c_1,\dots, c_m$ are at most non-zero, therefore the series $\sum_{i\geq 0}\varphi_ic_{i+1}$ is actually a finite sum and the result only depends on the first $m$ Taylor coefficients of $\Phi$ at the collision point.
\end{remark}
Thus, according to \cref{cor:generalpolecollapselinear}, $F_{\mathrm{lin}}$ will have a finite, computable limit as $n\to\infty\,$, while any divergence of the prepotential is contained in $F_{\log}$.

Before moving on to computing the components of $\iota_n^\ast c^\ttop$ in the collision chart, we notice the following. Since movable poles lie in $\CC^\ast$ and the $\phi$-flat coordinates are  given by logarithms of the simple-pole locations in the top-dimensional stratum, it is more natural to deform the collision point multiplicatively rather than additively. Let us pick a branch of the logarithm and set the collision point at $e^{\beta}\,$. In other words, we redefine our collision quantities as:
\begin{equation}\label{eq:collisionchartexp}
\begin{aligned}
        \beta&:=\tfrac1m (\beta_1+\dots+\beta_m)\,,&&&
        \Delta_i&:=\beta_i-\beta\,,&&&
c_i&:=\sum_{j=1}^m \alpha_j\,\Delta_j^{i-1}\,.
\end{aligned}
\end{equation}
Two items need to be checked: first, that this still defines a sequence of functions converging uniformly on compact subsets that do not contain the collision point, and second, whether the new $c$-coordinates still have a meaningful interpretation as coefficients of the limit function. It is easy to check that \cref{prop:polecoalescence} still holds true for the function $\widetilde{\lambda}_n(z):=\lambda_n(e^z)\,$, i.e. that the sequence of functions $\{\widetilde{\lambda}_n\}_{n\geq 0}$ converges to $\widetilde{\lambda}(z):=\lambda(e^z)$ uniformly on compact sets away from $\beta\,$, and that $c_1,\dots, c_m$ are the principal Laurent coefficients of $\widetilde{\lambda}$ at $\beta\,$. Since, once a branch is chosen, the map $\beta\mapsto e^\beta$ is biholomorphic, it follows that the result regarding uniform convergence is preserved. The collision shape $(\kappa_1,\dots, \kappa_m)\in (\CC^\ast)^m\smallsetminus\diag$ is introduced likewise. With an abuse of notation, we shall still call $c_1,\dots, c_m$ \textit{principal Laurent coefficients} of the limit function.

\begin{lem}\label{lem:ctensorcomponebeta}
    In the exponential collision chart defined by \eqref{eq:collisionchartexp}, $F_{\log}$ contributes only to the components of $\iota_n^\ast c^\ttop$ in the directions with at most one $\beta$-entry, while, conversely, the third derivatives of $F_{\mathrm{lin}}$ only appear in the $\beta\beta$-block. 
    
    Explicitly, in the first block we have:
    \[
\begin{aligned}
 \bigl(\iota_n^\ast c^\ttop\bigr)_{ijk}&=\delta_n^{m+2-(i+j+k)}\sum_{a=1}^m\frac{(V^{-1}_\kappa)_{ai}(V^{-1}_\kappa)_{aj}(V^{-1}_\kappa)_{ak}}{\sum_{s=1}^m\delta_n^{m-s}(V^{-1}_\kappa)_{as}c_s}\,,\\
 \bigl(\iota_n^\ast c^\ttop\bigr)_{\beta jk}&=\delta_{j1}\delta_{k1}\,.
 \end{aligned}
    \]
\end{lem}
\begin{proof}
  Firstly, the fact that the logarithmic terms do not contribute to any component in more than one $\beta$ dimension is a consequence of the fact that $\pdv{\beta}\log(e^{\beta_i}-e^{\beta_j})=1$ for any $i<j\,$. Since $\beta$ is a linear combination of flat coordinates, this proves that the logarithmic term does not contribute to the $\beta\beta $-block. Similarly, $F_{\mathrm{lin}}$ does not contribute to any component in more than one $c$-direction as it is, by construction, at most linear in $\alpha\,$. Secondly, we have, as a consequence,
    \begin{align*}
        (\iota_n^\ast c^\ttop)_{ijk}&=\sum_{a,b,c=1}^m\pdv{\alpha_a}{c_i}\pdv{\alpha_b}{c_j}\pdv{\alpha_c}{c_k}\frac{\partial^3 F_{\log}}{\partial \alpha_a\partial\alpha_b\partial\alpha_c}\\&=\delta_n^{m+2-(i+j+k)}\sum_{a=1}^m\frac{(V^{-1}_\kappa)_{ai}(V^{-1}_\kappa)_{aj}(V^{-1}_\kappa)_{ak}}{\sum_{s=1}^m\delta_n^{m-s}(V^{-1}_\kappa)_{as}c_s}\,,\\
        (\iota_n^\ast c^\ttop)_{\beta jk}&=\sum_{b,c=1}^m\pdv{\alpha_b}{c_i}\pdv{\alpha_c}{c_j}\frac{\partial^3 F_{\log}}{\partial \beta\partial\alpha_b\partial\alpha_c}\\
        &=\sum_{b,c=1}^m (V^{-1}_\Delta)_{bj}(V^{-1}_\Delta)_{ck}=\delta_{j1}\delta_{k1}\,,
    \end{align*}
    where the last equality follows from $\sum_{i=1}^m (V^{-1})_{ij}=\delta_{j1}$ in any set of variables.
\end{proof}
According to \cref{lem:divergencectensorcomponents}, the components of $c^\low$ that do not require renormalisation are precisely the ones  arising from third derivatives of  $F_{\mathrm{lin}}\,$, which has a finite limit as $n\to\infty$ by  \cref{cor:generalpolecollapselinear}. In particular, by \cite[Thm. 4.3]{PS26}, we have:
\begin{equation}\label{eq:Flin}
    F_{\mathrm{lin}}(\underline{t},\underline{\alpha},\underline{\beta})=F^{\mathrm{DZ}(\hat{\ell},\bar{k})}(\underline{t})+\sum_{i=1}^m\alpha_i\,H(\beta_i,\underline{t})\,,\qquad H'(z,\underline{t})=h(e^z,\underline{t})\,,
\end{equation}
where $F^{\mathrm{DZ}(\hat{\ell},\bar{k})}(\underline{t})$ denotes the type $A$ Dubrovin--Zhang solution associated with the extended affine Weyl group $\widetilde{W}^{(\bar{k})}(A_{\hat{\ell}+\bar{k}})$, and the coordinates $\underline{t}$ are flat coordinates for the corresponding Dubrovin--Zhang manifold.
\begin{lem}\label{lem:ctensorcomptwobeta}
    In the exponential collision coordinates defined by \cref{eq:collisionchartexp}, the components of $\iota_n^\ast c^\ttop$ that depend on $F_{\mathrm{lin}}$ are:
    \[
    \begin{aligned}
\bigl(\iota_n^\ast c^\ttop\bigr)_{\beta\beta k}&=\tfrac{1}{(k-1)!}\,H^{(k+1)}(\beta)+\order{\delta_n^{m+1-k}}\,,\\
\bigl(\iota_n^\ast c^\ttop\bigr)_{\beta\beta\beta}&=\sum_{i=1}^m\tfrac{1}{(i-1)!} c_i\,H^{(i+2)}(\beta)+\order{\delta_n}\,.
\end{aligned}
    \]
    The derivatives of $H$ are to be taken with respect to $\beta\,$.
\end{lem}
\begin{proof}
    According to \cref{lem:ctensorcomponebeta}, $F_{\log}$ does not contribute to the components in more than one $\beta$ direction. In the computations, we omit the $t$-dependence for the sake of brevity. We have
    \[
    \begin{aligned}
        \bigl(\iota_n^\ast c^\ttop\bigr)_{\beta\beta k}&=\sum_{i=1}^m\pdv{\alpha_i}{c_k}\frac{\partial^3 F_{\mathrm{lin}}}{\partial \beta\partial\beta\partial\alpha_i}=\sum_{i=1}^m(V^{-1}_\Delta)_{ik}\,H''(\beta+\Delta_i)\\
        &=\sum_{j\geq 0}\tfrac{1}{j!}H^{(j+2)}(\beta)\sum_{i=1}^m(V^{-1}_\Delta)_{ik}\Delta_i^j\\&=\tfrac{1}{(k-1)!}\,H^{(k+1)}(\beta)+\delta_n^{m+1-k}\sum_{j\geq m}\delta_n^{j-m}\tfrac{1}{j!}H^{(j+2)}(\beta)\sum_{i=1}^m (V^{-1}_\kappa)_{ik}\kappa_i^j\,,\\
         \bigl(\iota_n^\ast c^\ttop\bigr)_{\beta\beta \beta}&=\sum_{i=1}^m H'''(\beta+\Delta_i)\sum_{j=1}^m (V_\Delta^{-1})_{ij}c_j\\
         &=\sum_{i=1}^m\tfrac{1}{(i-1)!} c_i\,H^{(i+2)}(\beta)+\delta_n\sum_{k\geq m}\sum_{i,j=1}^m\delta_n^{k-j} \tfrac{1}{k!}H^{(k+3)}(\beta)\,c_j\,(V^{-1}_\kappa)_{ij}\kappa_i^k\,. 
    \end{aligned}
    \]
\end{proof}
On the other hand, by  \cref{cor:generalpolecollapselinear},
\[
\lim_{n\to\infty}\sum_{i=1}^m\alpha_iH(\beta_i)=\sum_{i=1}^m\tfrac{1}{(i-1)!}\,c_i\,H^{(i-1)}(\beta)\,.
\]
Clearly, taking the appropriate third derivatives with respect to the collision coordinates reproduces the limits of the expressions for the components of $\iota_n^\ast c^\ttop$ in \cref{lem:ctensorcomptwobeta}.
\begin{prop}\label{prop:prepotentiallimit}
    Let $F^\ttop$ be the prepotential in the top-dimensional stratum $\Hw_{\hat{\ell},\bar{k}}^\ttop\,$ and let $F_{\log}$ and $F_{\mathrm{lin}}$ be as in \cref{eq:Ftop,eq:Flog,eq:Flin}. If $F^\low$ denotes the prepotential in the target stratum in the collision chart, then:
    \[
    F^\low(\underline{t},\underline{c},\beta)=F_{\mathrm{ren}}(\underline{c},\beta)+\lim_{n\to\infty} F_{\mathrm{lin}}(\underline{t},\underline{c},\beta)\,,
    \]
    where $F_{\mathrm{ren}}$ comes from suitably renormalising the divergences in $F_{\log}\,$. In particular, it is universal in the sense that it only depends on the number of poles that are being collapsed, and not on the terms in the superpotential that are kept fixed.
\end{prop}
\begin{proof}
We provide an algorithmic procedure to construct the renormalisation, while showing that the only terms that can be affected come from $F_{\log}\,$.

    Fix a collision shape $\underline{\kappa}\in (\CC^\ast)^m\smallsetminus\diag$ and let $\{\iota_n\}_{n\geq 0}$ be the associated family of immersions of the lower-dimensional stratum into the top-dimensional, as discussed in \cref{subs:metric3ptfunction}. Consider the family of functions $\{F^\ttop_n:=F^\ttop\circ \iota_n\}_{n\geq 0}$ on the target stratum, written in terms of the collision coordinates. Due to \cref{cor:generalpolecollapselinear}, the failure of this family to have a pointwise limit as $n\to\infty$ is completely determined by $F_{\log}\,$. 
    
    In particular, given the explicit expression for $F_{\log}$ in \cref{eq:Flog} and the change of variables in \eqref{eq:collisionchartexp}, one can explicitly identify the divergent terms and \textit{regularise} the family of functions by subtracting a $\underline{\kappa}$-dependent counterterm $CT_n\,$. This produces the sequence $\bigl\{F^{\mathrm{(reg)}}_n:=F^\ttop_n-CT_n\}_{n\geq 0}$ that has a finite pointwise limit $F^{\mathrm{(reg)}}$  on the target stratum as $n\to\infty\,$. Let ${}^\low\nabla$ be the Levi--Civita connection of the Saito metric $\eta^\low\,$, and compute the third covariant derivative ${}^\low\nabla^3 F^{\mathrm{(reg)}}\,$. If this coincides with $c^\low\,$, then regularisation coincides with renormalisation and no further step is needed.  If it does not, define the tensor field:
    \[
    D:={}^\low\nabla^3 F^{\mathrm{(reg)}}-c^\low\,.
    \]
    If there exists a function $G$ on the target stratum such that ${}^\low\nabla^3 G=D\,$, then the prepotential $F^\low$ is given by $F^\low=F^{\mathrm{(reg)}}-G\,$. Since $\eta^\low$ is flat, this follows from \cite{Fer81}. As we have noticed, the third derivatives of the limit of $F_{\mathrm{lin}}$ already give the three-point function components to which they contribute (i.e. the ones in at least two $\beta$ directions). As a consequence, the components of $D$ in those directions are automatically zero, up to terms depending on lower derivatives of $F^{\mathrm{(reg)}}$ due to the fact that the collision coordinates are, in general, not flat. Hence, the explicit form of $G$ will depend on how the flat coordinates of $\eta^\low$ are given as functions of the collision chart -- and, possibly, up to remaining discrepancies coming from how the naive counterterm $CT_n$ was chosen -- but will not modify the terms coming from the limit of $F_{\mathrm{lin}}\,$, as they are known to contribute to some three-point function components of $c^\low\,$. 
\end{proof}
\begin{example}\label{ex:prepotentialtwopoles}
      Let us compute the renormalisation explicitly for the collision of two poles. The logarithmic terms that need renormalisation are:
    \[
    F_{\log}(\underline{\alpha},\underline{\beta})=\tfrac12 \alpha_1^2\log(\alpha_1e^{\beta_1})+\tfrac12\alpha_2^2\log(\alpha_2e^{\beta_2})+\alpha_1\alpha_2\log(e^{\beta_1}-e^{\beta_2})\,.
    \]
    Consider the exponential collision chart in \cref{eq:collisionchartexp}. For $m=2\,$, we pick $\Delta:=\Delta_1=\tfrac12(\beta_2-\beta_1)$ as the shape collision parameter, which tends to zero as we approach the limit stratum. Using the inverse coordinate transformation, we get to:
    \[
    \begin{aligned}
        F_{\log}(c_1,c_2,\beta)&=\frac{c_2^2}{4\Delta^2}\log\frac{c_2}{4\Delta^2}+\tfrac12c_1^2\beta+\tfrac14 c_1^2\log c_2+\tfrac12 c_1c_2-\tfrac{1}{24}c_2^2+\order{\Delta^2}\,.
    \end{aligned}
    \]
    We regularise by a minimal regularisation scheme, i.e. we subtract:
    \[
    CT_\Delta(c_2):=\frac{c_2^2}{4\Delta^2}\log\frac{c_2}{4\Delta^2}\,.
    \]
    
   In order to match the corresponding finite limit with the prepotential on the target, we  need flat coordinates on the limit stratum. Using the expression for the metric in residue-pole coordinates in \cite{PS26} on the top-dimensional stratum, it is easy to see that, in collision coordinates, the pull-back to the lower-dimensional stratum is degenerate: $\dd{\beta}\dd{c_1}\,$. This is true for any number of colliding poles. In particular, this ensures that $\beta$ and $c_1$ are still flat coordinates on the target stratum -- as expected from the formulae in \cite[Theorem 5.1]{Dub96} -- and that the non-trivial flat coordinates come from the inner critical points. In particular, using the asymptotics in \cref{cor:asymptoticinnerCP}, one finds:
   \[
   \eta^\low=\dd{\beta}\dd{c_1}+\tfrac{1}{2c_2}\dd{c_2}^2\,.
   \]
   Therefore, we can take as flat coordinates on the lower-dimensional stratum:
   \begin{equation}\label{eq:flattwopole}
       \begin{aligned}
           t_{-1}&:=\beta\,,&&& t_{-2}&:=c_1\,,&&&t_{-3}&:=\sqrt{c_2}\,.
       \end{aligned}
   \end{equation}

  The pull-backs of the components of $c^\ttop$ to the lower-dimensional stratum that need renormalisation are computed using the criterion provided by \cref{lem:divergencectensorcomponents}, and  the asymptotic expansion of the corresponding components of the inner three-point function  is worked out using \cref{lem:expansioncritical}. Subtracting the two and taking the limit as $\Delta\to0$ gives the collision-chart components of $c^\low\,$.

  Finally, we express both $c^\low $ and the finite limit of $F_{\log}-CT_\Delta$ in the flat coordinates of $\eta^\low$. In this case, the expressions already match, so no further renormalisation other than subtraction of $CT_\Delta$ is required, and we get the renormalised logarithmic term:
  \begin{equation*}
       F^{(m=2)}_{\mathrm{ren}}(c_1,c_2,\beta)=\tfrac12c_1^2\beta+\tfrac14 c_1^2\log c_2+\tfrac12 c_1c_2-\tfrac{1}{24}c_2^2\,,
  \end{equation*}
 or the equivalent expression in the flat coordinates in \eqref{eq:flattwopole}.

 Let us explicitly derive the prepotential on the limit-stratum in the low-dimensional example $\hat{\ell}=4\,$, $\bar{k}=1$ and $p=m=2\,$. In other words, the top-dimensional stratum only has two additional simple poles and its boundary is only made up of the one component $\Hw_{4,1}(2)\,$. This is a case of a total collapse of all the simple poles into a higher-order pole. The prepotential on the top-dimensional stratum can be found in \cite[Example 4.1]{PS26}. We have,
 \[
 \begin{aligned}
      F_{\mathrm{lin}}(\underline{t},\underline{\alpha},\underline{\beta})&=\tfrac12 t_1^2t_2+e^{t_2}+\alpha_1 H(\beta_1,\underline{t})+\alpha_2 H(\beta_2,\underline{t})\,,
 \end{aligned}
 \]
 where
 \[
 H(z,\underline{t}):=e^z+t_1z-e^{t_2-z}\,,
 \]
 and by \cref{cor:generalpolecollapselinear}:
 \[
 \begin{aligned}
     \lim_{\Delta\to 0}F_{\mathrm{lin}}(\underline{t},\underline{\alpha},\underline{\beta})&=\tfrac12 t_1^2t_2+e^{t_2}+c_1H(\beta)+c_2H'(\beta)\\
     &=\tfrac12 t_1^2t_2+e^{t_2}+e^\beta(c_1+c_2)+t_1(c_1\beta+c_2)-e^{t_2-\beta}(c_1-c_2)\,.
 \end{aligned}
 \]
Therefore, the prepotential on the limit stratum in the collision chart is, using  \cref{prop:prepotentiallimit}, given by
 \[
 \begin{aligned}
     F(\underline{t},\underline{c},\beta)&=\tfrac12 t_1^2t_2+e^{t_2}+\tfrac12c_1^2\beta+\tfrac14 c_1^2\log c_2+\tfrac12 c_1c_2-\tfrac{1}{24}c_2^2+\\
     &\quad +e^\beta(c_1+c_2)+t_1(c_1\beta+c_2)-e^{t_2-\beta}(c_1-c_2)\,.
 \end{aligned}
 \]
 When written in the flat coordinates in \eqref{eq:flattwopole}, this coincides -- up to a rescaling of $t_{-2}$ and a relabelling of $t_2$ -- with the prepotential in \cite[Example 3.4]{MZ24}, as expected.
\end{example}
\begin{remark}\label{rem:partialcollapse}
    In a partial collapse, i.e. when $m<p$ and the limit stratum is not the deepest one in the boundary, the result changes due to  $F_{\mathrm{lin}}$ containing terms depending on the simple poles that are fixed in the collision. In particular, suppose $m=2$ and denote by $\gamma_1,\dots, \gamma_{p-2}$ and $\xi_1,\dots, \xi_{p-2}$ the flat coordinates corresponding to the simple poles that do not take part in the collision, then the linear term will also contain:
 \[
 F_{\mathrm{lin}}^{\mathrm{partial}}:=\alpha_1\sum_{i=1}^{p-2}\gamma_i\log(e^{\beta_1}-e^{\xi_i})+\alpha_2\sum_{i=1}^{p-2}\gamma_i\log(e^{\beta_2}-e^{\xi_i})\,.
 \]
 The point is that $F_{\mathrm{lin}}$ contains all the terms that are at most linear in the residue-coordinates of the poles we are colliding. In particular, this will have finite limit:
 \[
 \begin{aligned}
     \lim_{\Delta\to 0} F_{\mathrm{lin}}^{\mathrm{partial}}&=c_1\sum_{i=1}^{p-2}\gamma_i\log(e^{\beta}-e^{\xi_i})+e^\beta c_2\sum_{i=1}^{p-2}\frac{\gamma_i}{e^\beta-e^{\xi_i}}\,.
 \end{aligned}
 \]
 This explains the presence of rational terms in the difference between the position of two movable poles in the \cref{ex:A71234}, and why they are absent in the other examples in \cref{sec:Ex}, as they either correspond to the longest or shortest partition. This clearly generalises to the partial collision of a higher number of simple poles, the limit will in these cases contain higher-order derivatives of the logarithmic term.
\end{remark}

\begin{example}
For the collapse of $m=3$ simple poles, we start by determining the flat coordinates on the lower-dimensional stratum. The Frobenius metric in the collision chart is:
\[
\begin{aligned}
    \eta^\low=\dd{\beta}\dd{c_1}+\frac{2}{3}\frac{1}{c_3}\dd{c_2}\dd{c_3}-\frac{2}{9}\frac{c_2}{c_3^2}\dd{c_3}^2\,.
\end{aligned}
\]
This immediately gives the flat coordinates:
    \begin{equation}\label{eq:flatthreepoles}
            \begin{aligned}
        t_{-1}&:=\beta\,, &&& t_{-2}&:=c_1\,,&&& t_{-3}&:=\frac{c_2}{3\sqrt[3]{c_3}}\,, &&&t_{-4}&:=\sqrt[3]{c_3}\,.
    \end{aligned}
    \end{equation}
As for the prepotential, an explicit calculation using the asymptotics in \cref{cor:asymptoticinnerCP} shows that, unlike the previous example, here minimal regularisation of the logarithmic terms does not give renormalisation. Therefore, on top of subtracting a collision-shape dependent counterterm to make the limit finite, we also need to add a function on the lower-dimensional stratum to make the three-point function components match. The final result in collision coordinates is:
\begin{equation*}
    \begin{aligned}
F^{(m=3)}_{\mathrm{ren}}(\beta,c_1,c_2,c_3)&= 
\tfrac{1}{2}\,c_1^{2}\,\beta
+\tfrac{1}{2}\,c_1^{2}\,\log c_3
+\tfrac{1}{2}\,c_1c_2
+\tfrac{1}{24}\,c_1c_3\\[4pt]
&\quad+\tfrac{1}{6}\,\frac{c_1c_2^{2}}{c_3}
-\tfrac{1}{24}\,c_2^{2}
-\tfrac{1}{108}\,\frac{c_2^{4}}{c_3^{2}}
-\tfrac{1}{960}\,c_3^{2}\,.
\end{aligned}
\end{equation*}
Again, we explicitly consider the case of a total collapse $m=p=3$ with $\hat{\ell}=5$ and $\bar{k}=1\,$. The contribution to the top-dimensional-stratum prepotential that is linear in the residue coordinates of the collapsing poles is \cite{PS26}:
\[
\begin{aligned}
    F_{\mathrm{lin}}(\underline{t},\underline{\alpha},\underline{\beta})&=\tfrac12 t_1^2t_2+e^{t_2}+\alpha_1\,H(\beta_1,\underline{t})+\alpha_2\,H(\beta_2,\underline{t})+\alpha_3\, H(\beta_3,\underline{t})\,,
\end{aligned}
\]
with the same $H$ as in the previous example.
This has a finite limit given by, according to \cref{cor:generalpolecollapselinear}:
\[
\begin{aligned}
    \lim_{n\to\infty}F_{\mathrm{lin}}(\underline{t},\underline{c},\beta)&=\tfrac12 t_1^2t_2+e^{t_2}+c_1H(\beta,\underline{t})+c_2H'(\beta,\underline{t})+\tfrac12 c_3H''(\beta,\underline{t})\\
    &=\tfrac12 t_1^2t_2+e^{t_2}+e^{\beta}\bigl(c_1+c_2+\tfrac12 c_3\bigr)+\\&\quad+t_1\bigl(c_1\beta+c_2\bigr)-e^{t_2-\beta}\bigl(c_1-c_2+\tfrac12 c_3\bigr)\,.
\end{aligned}
\]
Therefore, in this case the prepotential on the limit stratum in the collision chart is:
\[
\begin{aligned}
    F(\underline{t},\underline{c},\beta)&=\tfrac12 t_1^2t_2+e^{t_2}+\tfrac{1}{2}\,c_1^{2}\,\beta
+\tfrac{1}{2}\,c_1^{2}\,\log c_3
+\tfrac{1}{2}\,c_1c_2
+\tfrac{1}{24}\,c_1c_3\\
&\quad+\tfrac{1}{6}\,\frac{c_1c_2^{2}}{c_3}
-\tfrac{1}{24}\,c_2^{2}
-\tfrac{1}{108}\,\frac{c_2^{4}}{c_3^{2}}
-\tfrac{1}{960}\,c_3^{2}+t_1\bigl(c_1\beta+c_2\bigr)+\\
 &\quad+e^{\beta}\bigl(c_1+c_2+\tfrac12 c_3\bigr)- e^{t_2-\beta}\bigl(c_1-c_2+\tfrac12 c_3\bigr)\,.
\end{aligned}
\]
As a function of its flat coordinates in \cref{eq:flatthreepoles}, this coincides with the prepotential in \cite[Example 3.5]{MZ24} up to a linear change of variables.
\end{example}

From the previous two examples, we notice that the renormalised limit of the logarithmic sector of the prepotential seems to exhibit a universal structure. Namely, it contains a logarithmic term in the leading principal Laurent coefficient $c_m\,$, and the other terms are Laurent polynomials in $c_m$ and ordinary polynomials in the other variables. For $m=2\,$, the dependence on $c_2$ is only polynomial, but this case is expected to be special due to \cref{cor:asymptoticinnerCP}. Thus, we anticipate that the cases $m>3$ have similar behaviour to $m=3\,$. Let us, therefore, pin down the universal structure of the renormalised term.
\begin{lem}\label{lem:regularisedstructural}
Upon subtraction of the minimal counterterm $CT_n\,$, the regularised limit of the logarithmic contribution in $F^\ttop_n:=F^\ttop\circ\iota_n$ is of the following form:
\[
   \lim_{n\to\infty}\bigl(F_{\log}-CT_n\bigr) =\tfrac12 c_1^2\beta+\tfrac12\gamma\,  c_1^2 \log c_m+P(\underline{c})\,, 
    \]
    with $P\in \CC[c_1,\dots, c_{m-1},c_m^{\pm1 }]$ and $\gamma\in\CC\,$.
\end{lem}
\begin{proof}
   The minimal counterterm is the one that precisely subtracts all the divergent terms of $F_{\log}\,$, without adding any artificial finite term at the regularisation step. Let us fix a collision shape $\underline{\kappa}\in(\CC^\ast)^m\smallsetminus\diag$ and consider the following expansions of the individual log terms:
   \[
   \begin{aligned}
       \log(e^{\beta_i}-e^{\beta_j})&=\beta+\log(e^{\delta_n\kappa_i}-e^{\delta_n\kappa_j})=\beta+\log\delta_n+\log(\kappa_i-\kappa_j)+\order{\delta_n}\,,\\
       \log\alpha_i&=-(m-1)\log\delta_n+\log c_m+\log (V_\kappa^{-1})_{im}+\order{\delta_n}\,.
   \end{aligned}
   \]
Now, since $\alpha_i=\sum_j\delta_n^{1-j} (V^{-1}_\kappa)_{ij}c_j\,$, it follows that in $F_{\log}\,$, the coefficient of $\beta$ is going to be:
\[
\begin{aligned}
    \tfrac12\sum_i\alpha_i^2+\sum_{i<j }\alpha_i\alpha_j=\tfrac12\biggl( \sum_i\alpha_i\biggr)^2=\tfrac12 c_1^2\,.
\end{aligned}
\]
Similarly, $\log c_m$ will have a coefficient $\tfrac12 \sum_i\alpha_i^2\,$. The finite part of such expression in the limit is clearly proportional to $c_1^2\,$, with a shape-dependent constant. Other finite terms may arise from multiplying the quadratic expressions in the $\alpha$-variables with the terms of order $\delta_n$ in the expansions of the logarithms. This will at most produce functions of the $c$-coordinates with a pole at $c_m=0$ coming from the positive-$\delta_n$-powers in $\log\alpha_i\,$, because otherwise $\alpha_i$ is polynomial in $c_1,\dots, c_m\,$.
\end{proof}
\begin{lem}\label{lem:flatcm}
    In the collision chart, the components of the three-point function in the collision block on the limit stratum lie in the ring $\mathcal{A}_{\underline{t}}[\beta,c_1,\dots, c_{m-1},c_m^{\pm 1}\bigr]\,$, where $\mathcal{A}_{\underline{t}}$ is a suitable ring generated by polynomials in the flat coordinates $\underline{t}$ that are left unchanged by the collision.
\end{lem}
\begin{proof}
Fix a collision shape $\underline{\kappa}\,$. The result is basically a consequence of the fact that:
\[
\alpha_i=\delta_n^{1-m} \gamma_i(\underline{\kappa})\,c_m\bigl(1+R_\kappa(\underline{c})\bigr)\,,
\]
with $R_\kappa\in \CC[c_1/c_m,\dots, c_{m-1}/c_m]\,$. Given the prepotential on the top-dimensional stratum, it follows immediately that the three-point function components of $\iota_n^\ast c^\ttop$ are Laurent polynomials in $c_m\,$, but only depend polynomially on any other collision coordinate. According to \cref{lem:expansioncritical}, we just need to check that the same applies to the corresponding coordinates of $c^\inn_n\,$. Clearly, the derivatives of $\lambda_n$ with respect to the collision coordinates only depend on the principal Laurent coefficients of the limit function polynomially. The only rational dependence comes from $\lambda_n''$ evaluated at the inner critical point in the denominator. From the expression of $\lambda_n''$ in the proof of \cref{lem:divergencectensorcomponents}, it immediately follows that one has an expansion analogous to the one of $\alpha_i$ for $1\,/\,\lambda_n''$ at the inner critical point, which proves the statement.
\end{proof}
\begin{lem}\label{lem:rootofcmflat}
    In the collision of $m$ poles, flat coordinates in the collision cluster are of the following form, as a function of the cluster collision coordinates:
    \[
    \begin{aligned}
        s&:=\beta\,,\\
        t_{-1}&:=c_1\,,\\
        t_{-k}&\in c_m^{k/m}\,\CC\bigl[\tfrac{c_1}{c_m},\dots, \tfrac{c_{m-1}}{c_m}\bigr]\,,&&& k=2,\dots, m-1\,,\\
        t_{-m}&:=c_m^{1/m}\,.
    \end{aligned}
    \] 
    Moreover, the Jacobian matrix of $(c_2,\dots, c_{m-1})\mapsto (t_{-2},\dots, t_{-(m-1)})$ is upper-triangular.
\end{lem}
\begin{proof}
The computations are similar to the ones in \cite[Lemma 4.2]{BvG22}, with $w_0$ replaced by $c_m$ and the ratios $c_j/c_m$ instead of the $w$-coordinates with positive indices.
The coordinates $s$ and $t_{-1}$ are proportional to the type-b and type-c flat coordinates in \cite[Theorem 5.1]{Dub96} respectively, see also \cref{Eq:5.1b,Eq:5.1c}. Equivalently, using \cref{lem:divergencectensorcomponents}, it is easy to see that the Saito metric components on the target stratum that contain at least one $\beta$ index do not require renormalisation, and are therefore simply given by the corresponding components of the Saito metric on the top-dimensional stratum. As we discussed in the examples, the latter is the degenerate metric $ \dd\beta\dd{c_1}\,$. As for $t_{-m}\,$, this is proportional to the type-a flat coordinate for $\infty_i$ being the collision point and $\alpha=1\,$ -- see \cref{Eq:5.1a} -- whereas the other ones correspond to the other choices of $\alpha\,$. Their form follows from expanding:
\[
\lambda(w)^{k/m}=\bigl(\tfrac{c_m}{w}\bigr)^{k/m}\,\bigl(1+R(\underline{c},w)\bigr)^{k/m}\,,
\]
where $R(\underline{c},w)\in \mathcal{A}_{\underline{t}}[c_1/c_m,\dots, c_{m-1}/c_m][w]$ and $w$ is a coordinate centered at the collision point.
\end{proof}
\begin{prop}\label{prop:renormalisedlimitexplicit}
   The renormalised limit of the logarithmic sector for the collision of $m$ poles is of the form:
    \[
    F^{(m)}_{\mathrm{ren}}(\beta,\underline{c})=\tfrac12 \gamma c_1^2\log c_m+\tfrac12 c_1^2\beta+G^{(m)}(\underline{c})\,,
    \]
    with $G^{(m)}\in \CC\bigl[c_1,\dots, c_{m-1},c_m^{\pm 1}\bigr]$ and $\gamma\in\CC\,$. For $m=2\,$, $G^{(2)}$ depends polynomially on $c_2\,$.
\end{prop}
\begin{proof}
This follows from the previous three Lemmas. In particular, since the three-point function components only depend rationally on the top Laurent coefficient, hence on the flat coordinate $t_{-m}$ from \cref{lem:rootofcmflat}, integrating will only produce at most logarithmic terms in $c_m\,$. We just need to check that there are no other such terms except for the one proportional to $c_1^2\log c_m\,$. Since the three-point function components are not logarithmic in $c_m\,$ (or in any other variable), $\log c_m$ must be at most proportional to a quadratic polynomial in the flat coordinates. No other quadratic contribution can be produced by the regularised limit, according to \cref{lem:regularisedstructural}. Hence, such terms, if they exist, can only be coming from the finite part of an inner three-point function components. In particular, they would give a non-zero residue at $c_m=0\,$, finite as $n\to\infty\,$, to a component $(c^\inn)_{ijm}\,$. From the asymptotics as $n\to\infty$ for $\Res_{s_n}\phi_{ijk}^{(n)}$ in \cref{lem:divergencectensorcomponents} and for $1/\lambda_n''\,$, as discussed in \cref{lem:flatcm}, however, it follows that a finite term proportional to $1/c_m$ in the inner three-point function component $(c^\inn)_{ijk}$ can only be produced if $m+2=i+j+k\,$. When $k=m\,$, the only solutions to such equation are $i=j=1\,$. Thanks to \cref{lem:rootofcmflat}, this proves the statement.
\end{proof} 

We conclude our study with a structure result for the prepotential of type $A$ generalised Dubrovin--Zhang Frobenius manifold.
\begin{thm}\label{thm:prepotential}
    Let $\hat{\ell},\bar{k},p\in\ZZ_{\geq 0}\,$. Fix $\pi\vdash p\,$, and let $(\lambda(\mu), \phi)$ be a Landau--Ginzburg model as in \cref{thm:FrobAmult} associated to the partition $\pi\,$. 
    Let $F^{\mathrm{DZ}(\hat{\ell},\bar{k})}(\underline{t})\in \mathcal{A}^{(\bar{k})}_{\hat{\ell}+\bar{k}}$ be the prepotential of the Dubrovin--Zhang Frobenius manifold associated to $\widetilde{W}^{(\bar{k})}\bigl(A_{\hat{\ell}+\bar{k}}\bigr)$, as in \cref{thm:DZ:reconstruction}, written in the flat coordinates $\underline{t}$ for its Saito metric. For any collision cluster $\alpha=1,\dots, \ell(\pi)\,$, we denote by $\beta_\alpha$ and $\underline{c}_\alpha:=\{c_{i\alpha}\}_{i=1}^{m_\alpha}$ its collision coordinates, as defined in \cref{eq:collisionchartexp}. 
    
  In the coordinates $\underline{t},\beta_1,\dots, \beta_{\ell(\pi)},\underline{c}_1,\dots,\underline{c}_{\ell(\pi)}\,$, the prepotential for the generalised Dubrovin--Zhang Frobenius manifold with Landau--Ginzburg model $(\lambda,\phi)$ takes the form:
    \begin{equation*}
        \begin{aligned}        F_\pi(\underline{t},\underline{\beta},\underline{c})&=F^{\mathrm{DZ}(\hat{\ell},\bar{k})}(\underline{t})+\sum_{\alpha=1}^{\ell(\pi)}F^{(m_\alpha)}_{\mathrm{ren}}(\beta_\alpha,\underline{c}_\alpha)+\\&\quad +\sum_{\alpha=1}^{\ell(\pi)}\sum_{i=1}^{m_\alpha}\Lambda_i(\underline{t},\beta_\alpha)\,c_{i\alpha}+\sum_{1\leq \alpha<\gamma\leq \ell(\pi)}F_{\mathrm{int}}^{\alpha,\gamma}(\beta_\alpha,\beta_\gamma,\underline{c}_\alpha,\underline{c}_\gamma)\,,
\end{aligned}
    \end{equation*}
  where $F^{(m)}_{\mathrm{ren}}$ is the renormalised limit of the logarithmic sector of the top-stratum prepotential for the collapse of $m$ simple poles, as introduced in \cref{prop:prepotentiallimit}, $\Lambda_i(\underline{t},\nu)\in \mathcal{A}^{(\bar{k})}_{\hat{\ell}+\bar{k}}[\nu,e^{\pm \nu}]$ and the pairwise interaction term between two collision clusters is of the form:
  \[
  \begin{aligned}
      F_{\mathrm{int}}^{\alpha,\gamma}(\beta_\alpha,\beta_\gamma,\underline{c}_\alpha,\underline{c}_\gamma)=c_{1\alpha}c_{1\gamma}\log(e^{\beta_\alpha}-e^{\beta_\gamma})+\sum_{k= 1}^{m_\alpha+m_\gamma-2}\frac{1}{(e^{\beta_\alpha}-e^{\beta_\gamma})^k}P^{\alpha,\gamma}_k(e^{\beta_\alpha},e^{\beta_\gamma},\underline{c}_\alpha,\underline{c}_\gamma)\,,
  \end{aligned}
  \]
 for some families of polynomials $\{P^{\alpha,\gamma}_k\}\,$.    In particular:
    \begin{enumerate}
        \item $F_\pi$ contains $\ell(\pi)$ logarithmic terms proportional to $c_{1\alpha}^2\log c_{m_\alpha\alpha}\,$, and else depends polynomially on $\{c_{1\alpha},\dots, c_{(m_\alpha-1)\alpha}\}_{1\leq \alpha\leq \ell(\pi)}\,$ and it is a Laurent polynomial in the leading principal Laurent coefficient $\{c_{m_\alpha\alpha}\}_{1\leq \alpha\leq \ell(\pi)}\,$ of any pole of order at least three.
        \item It is invariant under the action of the Young subgroup $\Sym_\pi$ of the partition $\pi$ acting by permutation of collision blocks giving poles of the same order.
    \end{enumerate}
\end{thm}
\begin{proof}
    This comes from recursively taking the renormalised limit of the top-dimensional stratum prepotential $F$ in its flat coordinates -- given in \cite{PS26} -- as described in \cref{prop:prepotentiallimit} for single pole collision. In particular, the renormalised limits $F_{\mathrm{ren}}^{(m_\alpha)}$ come from the terms in the logarithmic sector that only involve residue-pole coordinates in the same collision block, whereas the terms $F_{\mathrm{int}}^{\alpha,\gamma}$ are the finite limits of the logarithmic terms involving pairs $(\alpha_i,\beta_i)$ and $(\alpha_j,\beta_j)$ on the top-dimensional stratum lying in different collision clusters. Their explicit expression is a consequence of the phenomenon described in \cref{rem:partialcollapse} for a single collision. Finally, the terms involving the coefficient $\Lambda_i$ come from taking the limit of the terms in $F$ depending on the diagonal invariants that are linear in the residue coordinates. In particular, it is easy to see from \cref{cor:generalpolecollapselinear} that $\Lambda_i$ is proportional to the $(i-1)^{\text{th}}$ derivative of an exponential anti-derivative of the underlying Dubrovin--Zhang superpotential for the extended affine Weyl group $\widetilde{W}^{(\bar{k})}(A_{\hat{\ell}+\bar{k}})\,$. This proves the statement.
\end{proof}
\begin{remark}
Following \cite{Jon96}, by \textit{Young subgroup}, here, we mean the group of permutations of the equal parts of the partition. Explicitly, we split a given partition $\pi=(m_1,\dots,m_{\ell(\pi)})$ into pairwise disjoint blocks of size $\abs{I_k}$ for $k\in\ZZ_{\geq 0}\,$, where $I_k:=\{i=1,\dots, \ell(\pi)\,:\,m_i=k\}\,$. Then, we set:
\[
\Sym_\pi:=\prod_{k\in\ZZ_{\geq 0}}\Sym_{\abs{I_k}}\leq \Sym_{\ell(\pi)}\,.
\]
If $\pi$ is an ordered partition, this can be thought of as its stabiliser under the action of $\Sym_{\ell(\pi)}$ permuting its elements.
    Invariance under the action of the Young subgroup is the residual symmetry after breaking the full diagonal-invariance of the prepotential on the top-dimensional stratum \cite{PS26}, which is the case $\pi=(1^p)$ and, therefore, $\Sym_\pi\cong\Sym_p\,$.
\end{remark}
\begin{cor}\label{cor:prepotentialZuo}
On  the deepest boundary stratum, corresponding to the shortest partition $\pi=(p)\,$, the prepotential takes the form:
\[
\begin{aligned}
    F(\underline{t},\beta,\underline{c})&=F^{\mathrm{DZ}(\hat{\ell},\bar{k})}(\underline{t})+F^{(p)}_{\mathrm{ren}}(\beta,\underline{c})+\sum_{i=1}^{p}\Lambda_i(\underline{t},\beta)\,c_{i}\,.
\end{aligned}
\]
In particular, this is a Laurent polynomial in the leading Laurent coefficient $c_p$ whenever $p>2\,$, up to a logarithmic term of the form $ c_1^2\log c_p\,$, and it depends polynomially on the lower coefficients $c_1,\dots, c_{p-1}\,$. 
\end{cor}
\begin{remark}
        Notice, in particular, that this proves the conjecture about the form of the prepotential for the Frobenius manifold on orbit spaces of extended affine Weyl groups in type $A$ as constructed in \cite{MZ24}. In particular, up to the logarithmic term, $F$ depends exponentially on two variables, one being the standard generator of the extended direction in the Dubrovin--Zhang Frobenius manifold \cite{DZ98}, the other being the location of the collision point in the exponential chart. Furthermore, it depends rationally on one flat variable for $p>2\,$, which is the one given as a $p^{\text{th}}$ root of $c_p\,$, and polynomially on all the remaining flat variables, as described in \cref{lem:rootofcmflat}.
\end{remark}

\begin{example}
    Consider the case $\hat{\ell}=2\,,\bar{k}=1$ and $p=2\,$. The two partitions $\pi_1=(1,1)$ and $\pi_2=(2)$ of $p$ correspond to the cases where the superpotential has two movable simple poles and a double pole respectively, hence they correspond to the previous \cref{ex:A5123,ex:A513}. We, here, show that our pole collision procedure produces the same prepotentials, up to a linear change of variables.

    Firstly, we identify the underlying Dubrovin--Zhang type solution, which is the same for both Landau--Ginzburg models. This is the one associated to the extended affine Weyl group $\widetilde{W}^{(1)}(A_3)\,$, which can be found in \cite[Example 2.5]{DZ98}:
    \[
    \begin{aligned}
        F^{\mathrm{DZ}(2,1)}(\underline t)&=\tfrac12 t_1^2t_4+\tfrac13 t_1t_2t_3+\tfrac{1}{18}t_2^3-\tfrac{1}{36}t_2^2t_3^2+\tfrac{1}{648}t_2t_3^4-\tfrac{1}{19440}t_3^6+\bigl(t_2+\tfrac16 t_3^2\bigr)e^{t_4}\,.
    \end{aligned}
    \]

    Secondly, we construct the solution associated to the longest partition $\pi=(1,1)\,$. Since this is a case where the two movable poles are both simple, the solution is computed using the results of \cite{PS26}, and it is:
    \[
    \begin{aligned}
        F_{\pi_1}(\underline t, \underline{ \alpha},\underline{\beta})&= F^{\mathrm{DZ}(2,1)}(\underline t)+\tfrac12 \bigl(\alpha_1^2\log(\alpha_1e^{\beta_1})+\alpha_2^2\log(\alpha_2e^{\beta_2})\bigr)+\alpha_1\alpha_2\log(e^{\beta_1}-e^{\beta_2})+\\
        &\quad +\tfrac13\bigl(\alpha_1e^{3\beta_1}+\alpha_2e^{3\beta_2}\bigr)-\tfrac12 t_3\bigl(\alpha_1e^{2\beta_1}+\alpha_2e^{2\beta_2}\bigr)+\bigl(t_2+\tfrac16 t_3^2\bigr)\bigl(\alpha_1e^{\beta_1}+\alpha_2e^{\beta_2}\bigr)+\\
        &\quad-t_1\bigl(\alpha_1\beta_1+\alpha_2\beta_2\bigr)-e^{t_4}\bigl(\alpha_1e^{-\beta_1}+\alpha_2e^{-\beta_2}\bigr)-\tfrac13 t_2t_3(\alpha_1+\alpha_2)\,.
    \end{aligned}
    \]
    This coincides with the solution in \cref{ex:A5123}, under the following linear change of variables:
    \[
    \begin{aligned}
        t_1&\mapsto t_5\,,&
        t_2&\mapsto \tfrac16 t_2\,,&t_3&\mapsto t_1\,,&t_4&\mapsto\tfrac{10}{3}t_0-t_{-2}-t_{-1}\,,\\
        \alpha_1&\mapsto t_3-t_4\,,&\alpha_2&\mapsto -t_3\,,&\beta_1&\mapsto t_{-2}-\tfrac23 t_0+i\pi\,& \beta_2&\mapsto t_{-1}-\tfrac23 t_0+i\pi\,.
    \end{aligned}
    \]

 Thirdly, we compute the prepotential on the boundary stratum corresponding to the partition $\pi_2=(2)$ using the pole collision formalism, starting from the prepotential of the top-dimensional stratum given above. Since we have two simple poles colliding, we do not have to repeat the computations for the renormalised limit of the logarithmic sector, or its flat coordinates, since they have been carried out in \cref{ex:prepotentialtwopoles}, and the result is universal. We only need to work out the finite limit of the linear term in the residue variables $\alpha_1\,,\alpha_2$. In particular, the function $H$ defined in \cref{eq:Flin} is here given by:
    \[
    H(\beta,\underline t)=\tfrac13 e^{2\beta}-\tfrac12 t_3e^{2\beta}+\bigl(t_2+\tfrac16 t_3^2\bigr)e^\beta-t_1\beta-e^{t_4-\beta}-\tfrac13 t_2t_3\,.
    \]
    Therefore, according to \cref{prop:renormalisedlimitexplicit}, the limit of the linear part in the residue coordinates is
    \[
    \begin{aligned}
\lim_{n\to\infty}\bigl(\alpha_1H(\beta_1)+\alpha_2H(\beta_2)\bigr)&=e^{3\beta}\bigl(c_2+\tfrac13 c_1\bigr)-t_3e^{2\beta}\bigl(c_2+\tfrac12 c_1\bigr)+e^\beta\bigl(t_2+\tfrac16 t_3^2\bigr)(c_1+c_2)+\\
&\quad -t_1(c_1\beta+c_2)+e^{t_4-\beta}(c_2-c_1)-\tfrac13 c_1t_2t_3\,.
    \end{aligned}
    \]
    In the flat coordinates in \cref{eq:flattwopole}, the prepotential on the lower-dimensional stratum is therefore:
    \[
    \begin{aligned}
        F_{\pi_2}(\underline t, t_{-1},t_{-2},t_{-3})&=F^{\mathrm{DZ}(2,1)}(\underline t)+\tfrac12 t_{-2}^2t_{-1}+\tfrac12 t_{-2}^2\log t_{-3} +\tfrac12 t_{-2}t_{-3}^2 -\tfrac{1}{24}t_{-3}^4+\\
        &\quad+e^{3t_{-1}}\bigl(t_{-3}^2 +\tfrac13 t_{-2}\bigr)-t_3e^{2t_{-1}}\bigl(t_{-3}^2 +\tfrac12 t_{-2}\bigr)+e^{t_{-1}}\bigl(t_2+\tfrac16 t_3^2\bigr)(t_{-2}+t_{-3}^2 )+\\
&\quad -t_1(t_{-2}t_{-1}+t_{-3}^2 )+e^{t_4-t_{-1}}(t_{-3}^2 -t_{-2})-\tfrac13 t_{-2}t_2t_3\,.
    \end{aligned}
    \]
    This coincides with the prepotential in Example \ref{ex:A513}, up to a linear change of variables. In particular, $t_1,t_2$ and $t_3$ map as above, since their images do not depend on coordinates in the directions of the movable poles, whereas for the remaining coordinates we have:
    \[
    \begin{aligned}
        t_4&\mapsto \tfrac{10}{3}t_0-2\,t_{-1}\,,& t_{-1}&\mapsto t_{-1}-\tfrac23 t_0+i\pi\,,&t_{-2}&\mapsto -t_4\,,&t_{-3}&\mapsto-\tfrac{i}{2}t_3\,.
    \end{aligned}
    \]
\end{example}

  \section{Conclusion and Outlook}
  \label{sec:conc}
  We have classified Frobenius manifold structures associated with  Landau--Ginzburg models constructed by taking a degeneration of the family of spectral curves of the relativistic Toda chain in Dynkin type $A\,$, obtained by shifting arbitrary subsets of fundamental characters as in \cref{eq:shiftgen}. This generalises the result of \cite{BvG22}, where a B-model for Dubrovin--Zhang Frobenius manifolds \cite{DZ98}, of any fixed Dynkin type, was constructed by shifting a single fundamental character. For this reason, and due to their form, we refer to the resulting Landau--Ginzburg models as being of \textit{generalised Dubrovin--Zhang type}. 

  Given the classification result in \cref{thm:AMain}, for a fixed index set $I$ of size $p+1$, we noticed that the corresponding Landau--Ginzburg superpotentials are in one-to-one correspondence with partitions of $p\,$. Moreover, the dimension of the corresponding Frobenius manifold is an increasing function of the length of the partition. In the longest partition, the Landau--Ginzburg superpotential has a ramification profile over the marked point $\infty$ in the target that differs from that of the associated Dubrovin--Zhang Frobenius manifold by the number of simple ramification points. Coarsening the partition produces non-simple poles, which can be interpreted as a collision limit of the simple poles in the top-dimensional Frobenius manifold. With this picture in mind, in \cref{sec:polecollision} we studied the behaviour of the Frobenius structure in this limit. Upon developing a suitable complex analytic formalism, we described how the metric and three-point functions for different partitions are related to one another, and eventually, given the explicit description of the top-dimensional structure in \cite{PS26}, we gave an algorithmic renormalisation procedure to compute the prepotential, leading to the explicit formula in \cref{thm:prepotential} in a well-behaved coordinate system. 

  As a special case, we recover the Frobenius manifolds on orbit spaces of type $A$ extended affine Weyl groups, as constructed in \cites{Zuo20,MZ24}. In particular, we are able to prove the conjectural form for the prepotential in \cite{MZ24} as a special case of \cref{thm:prepotential}.\\

  We conclude this work with some open questions:
  \begin{enumerate}
      \item In \cite{BvG22}, Landau--Ginzburg models for Dubrovin--Zhang Frobenius manifolds are constructed for any Dynkin type. In an upcoming preprint \cite{PvGprep}, we will generalise some aspects of the current paper to the other Dynkin types.  
      \item The pole collision formalism can be directly employed to study limits of prepotentials in genus zero, more specifically to types $B, C, D\,$. 
      \item In higher genus, we expect that similar pole-collision algorithms might exist. As a direct application, one could consider the Landau--Ginzburg models for Frobenius manifolds on the orbit spaces of (extended) Jacobi groups \cites{Ber00b,Ber00,Alm22}. The formalism developed in \cites{Rej23,Rej25} might prove useful to this end, and to extensions beyond the Dubrovin--Zhang Frobenius manifold class. More on the differential geometric side, one could also study the behaviour of the horizontal and vertical components of the identity and Euler vector field near the boundary, in the spirit of \cites{Str01,Str04}.
      \item It might be interesting to rephrase the complex-analytic formalism we have introduced in a more algebro-geometric language, in relation to standard compactifications of Hurwitz schemes \cites{HM82,Kon94,ELSSV01,ACV03,Lan19}. 
      \item The Landau--Ginzburg model in the deepest boundary stratum, corresponding to $\pi=(p)\,$, is known to give the Frobenius manifold on the space of orbits of extended affine Weyl groups with two marked Dynkin nodes \cites{Zuo20,MZ24}. This is also true for the other classical Dynkin types \cite{MZ22}. An orbit space description for arbitrary partition is, at present, missing. Given that, in our language, the Landau--Ginzburg model for $\pi=(p)$ is also obtained by shifting two fundamental characters at distance $p$ from one another, it is natural to expect that intermediate cases would correspond to further extensions of affine Weyl groups. However, this leads to  accidental isomorphisms, suggesting that the question would require a deeper analysis.
      \item The discriminant locus of a Hurwitz Frobenius manifold also has a natural stratification, and is conjectured, in general, to be a natural submanifold in the sense of \cite{Str04}. The behaviour of the Saito determinant on such a subspace has been  studied in \cites{AFS20,BvG00}, for orbit spaces of Weyl groups, and extended affine Weyl groups, respectively. Given our construction, one could generalise these results to the discriminant of generalised Dubrovin--Zhang Frobenius manifolds. 
\item On the A-model side, the canonical Dubrovin--Zhang Frobenius manifolds give the quantum cohomology of the projective line with two (type $A$) or three (types $D$ and $E$) orbifold points \cites{MT08,Ros10}. At least in type $A\,$, it is natural to expect that our generalised Dubrovin--Zhang Landau--Ginzburg models could give the quantum cohomology of $\RS$ with a greater number of orbifold points.
\item From the integrable system point of view, it is known that the principal hierarchy of the Dubrovin--Zhang Frobenius manifold coincides with the dispersionless limit of the Hirota integrable hierarchies constructed by Milanov–Shen–Tseng in \cites{MST16,BvG22}. It would be interesting to investigate the integrable hierarchies associated to our generalised Dubrovin--Zhang Frobenius manifolds with this in mind. This could be linked to the growing interest in the study of principal hierarchies of infinite-dimensional Frobenius manifolds and the relation to their finite-dimensional submanifolds, see e.g. \cites{WX12,MWZ24}. A pole-collision renormalised limit can probably also be introduced for the isomonodromic tau-function \cite{KS05}.
\item While a collision of movable poles modifies the Frobenius manifold structure in quite a non-trivial way, the change in the Dubrovin dual structure, as constructed in \cite{Dub04}, is expected to be more easily controlled \cites{RS07,Ril07}. In particular, both the flat coordinates of the intersection form and the dual solution to the WDVV equations are easier to predict. This is an interesting phenomenon on its own, and  could be straightforward to generalise to higher genus, see e.g. \cite{RS06} for the Jacobi group case. Finally, given the explicit form of the superpotentials and of the WDVV solutions, one can construct pairs of families of trigonometric solutions to the open WDVV equations related by (generalised) Dubrovin duality \cites{Alm25,PS25}.
\item Chekhov-Eynard-Orantin topological recursion (TR) is a recursive formalism recovering higher-genus information from genus-zero data \cite{EO09}. Its input is closely related to the Landau--Ginzburg model of a semisimple Frobenius manifold, and for a Frobenius manifold with a quantum-cohomological interpretation,  TR recovers higher-genus Gromov--Witten invariants \cites{DOSS14,DNOPS18,DNOPS19}. It would be interesting to connect the Frobenius manifolds described here to enumerative geometry by applying TR. The main distinction between the B-model of a Frobenius manifold and the TR input data lies in the presence of a bi-differential encoding first-order deformations of the Frobenius principal hierarchy. For B-models of Dubrovin--Zhang Frobenius manifolds, this bi-differential appears in the Dubrovin–Krichever data (along with a projection) and corresponds to the Chern–Simons planar two-point function \cite{Bri20}. Hence, one may attempt to find TR input data for Frobenius manifolds arising from the relativistic Toda chain  by studying variations of the Dubrovin–Krichever data or via the methods of \cites{KS05,KK06}.
  \end{enumerate}
We will address (some of) these problems in future work.

\appendix

\begin{acknowledgements}
The authors would like to thank Paolo Lorenzoni and Ian Strachan for discussions, suggestions and comments on earlier versions of this paper.
 A. P. was supported by a Ph.D. studentship of the EPSRC Doctoral Training Partnership (EP/W524359/1). K. v. G. was supported by funds of INFN (Istituto  Nazionale di Fisica Nucleare) through IS-CSN4 Mathematical Methods of Nonlinear Physics.
\end{acknowledgements}

\phantomsection

\bibliography{refs} 

@article{Bri20,
 AUTHOR = {Brini, Andrea},
 TITLE = {{$\rm E_8$} spectral curves},
 JOURNAL = {Proc. Lond. Math. Soc. (3)},
 FJOURNAL = {Proceedings of the London Mathematical Society. Third Series},
 VOLUME = {121},
 NUMBER = {4},
 YEAR = {2020},
 PAGES = {954--1032},
 DOI = {10.1112/plms.12331},
 URL = {https://doi.org/10.1112/plms.12331},
 ISSN = {0024-6115,1460-244X},
}

@incollection{Dub96,
 AUTHOR = {Dubrovin, Boris},
 TITLE = {Geometry of {$2$}D topological field theories},
 BOOKTITLE = {Integrable systems and quantum groups ({M}ontecatini {T}erme,
1993)},
 SERIES = {Lecture Notes in Math.},
 VOLUME = {1620},
 YEAR = {1996},
 PAGES = {120--348},
 PUBLISHER = {Springer, Berlin},
 DOI = {10.1007/BFb0094793},
}

@article{DSZZ19,
 AUTHOR = {Dubrovin, Boris and Strachan, Ian A. B. and Zhang, Youjin and
Zuo, Dafeng},
 TITLE = {Extended affine {W}eyl groups of {BCD}-type: their {F}robenius
manifolds and {L}andau-{G}inzburg superpotentials},
 JOURNAL = {Adv. Math.},
 FJOURNAL = {Advances in Mathematics},
 VOLUME = {351},
 YEAR = {2019},
 PAGES = {897--946},
 DOI = {10.1016/j.aim.2019.05.030},
 URL = {https://doi.org/10.1016/j.aim.2019.05.030},
 ISSN = {0001-8708,1090-2082},
 MRCLASS = {53D45},
}

@article{MZ22,
 AUTHOR = {Ma, Shilin and Zuo, Dafeng},
 TITLE = {Frobenius manifolds and a new class of extended affine {W}eyl
groups of {BCD}-type},
 JOURNAL = {J. Geom. Phys.},
 FJOURNAL = {Journal of Geometry and Physics},
 VOLUME = {180},
 YEAR = {2022},
 PAGES = {Paper No. 104622, 26},
 DOI = {10.1016/j.geomphys.2022.104622},
 URL = {https://doi.org/10.1016/j.geomphys.2022.104622},
 ISSN = {0393-0440,1879-1662},
}

@article{Ros10,
 AUTHOR = {Rossi, Paolo},
 TITLE = {Gromov-{W}itten theory of orbicurves, the space of
tri-polynomials and symplectic field theory of {S}eifert
fibrations},
 JOURNAL = {Math. Ann.},
 FJOURNAL = {Mathematische Annalen},
 VOLUME = {348},
 NUMBER = {2},
 YEAR = {2010},
 PAGES = {265--287},
 DOI = {10.1007/s00208-009-0471-0},
 URL = {https://doi.org/10.1007/s00208-009-0471-0},
 ISSN = {0025-5831,1432-1807},
 MRCLASS = {53D42 (14N35 53D45)},
 MRREVIEWER = {Hsian-Hua\ Tseng},
}

@incollection{Fer81,
 AUTHOR = {Ferus, Dirk},
 EDITOR = {Ferus, Dirk
and K{\"u}hnel, Wolfgang
and Simon, Udo
and Wegner, Bernd},
 TITLE = {A remark on {C}odazzi tensors in constant curvature spaces},
 BOOKTITLE = {Global {D}ifferential {G}eometry and {G}lobal {A}nalysis},
 YEAR = {1981},
 PAGES = {257--257},
 PUBLISHER = {Springer Berlin Heidelberg},
 ADDRESS = {Berlin, Heidelberg},
 ISBN = {"978-3-540-38419-9"
}}

@article{Str04,
 AUTHOR = {Strachan, I. A. B.},
 TITLE = {Frobenius manifolds: natural submanifolds and induced
bi-{H}amiltonian structures},
 JOURNAL = {Differential Geometry and its Applications},
 FJOURNAL = {Differential Geometry and its Applications},
 VOLUME = {20},
 NUMBER = {1},
 YEAR = {2004},
 PAGES = {67--99},
 DOI = {10.1016/j.difgeo.2003.10.001},
 URL = {https://doi.org/10.1016/j.difgeo.2003.10.001},
 ISSN = {0926-2245,1872-6984},
}

@article{Alm25,
 AUTHOR = {Almeida, Guilherme F.},
 TITLE = {{O}pen {H}urwitz {F}lat {F} manifolds},
 YEAR = {2025},
 EPRINT = {arXiv:2503.09258}}

@incollection{Dub04,
 AUTHOR = {Dubrovin, Boris},
 TITLE = {On almost duality for {F}robenius manifolds},
 BOOKTITLE = {Geometry, topology, and mathematical physics},
 SERIES = {Amer. Math. Soc. Transl. Ser. 2},
 VOLUME = {212},
 YEAR = {2004},
 PAGES = {75--132},
 PUBLISHER = {Amer. Math. Soc., Providence, RI},
}

@phdthesis{Ril07,
 AUTHOR = {Riley, A.},
 TITLE = {Frobenius manifolds: caustic submanifolds and discriminant almost duality},
 YEAR = {2007},
 MONTH = {01},
 NOTE = {University of Hull},
 SCHOOL = {University of Hull},
}

@incollection{Dub99,
 AUTHOR = {Dubrovin, Boris},
 TITLE = {Painlev\'e{} transcendents in two-dimensional topological
field theory},
 BOOKTITLE = {The {P}ainlev\'e{} property},
 SERIES = {CRM Ser. Math. Phys.},
 YEAR = {1999},
 PAGES = {287--412},
 PUBLISHER = {Springer, New York},
}

@article{HM99,
 AUTHOR = {Hertling, C. and Manin, Yuri I.},
 TITLE = {Weak {F}robenius manifolds},
 JOURNAL = {Internat. Math. Res. Notices},
 FJOURNAL = {International Mathematics Research Notices},
 NUMBER = {6},
 YEAR = {1999},
 PAGES = {277--286},
 ISSN = {1073-7928,1687-0247},
 MRCLASS = {53D45},
}

@article{Man05,
 AUTHOR = {Manin, Yuri I.},
 TITLE = {{F}-manifolds with flat structure and {D}ubrovin's duality},
 JOURNAL = {Adv. Math.},
 FJOURNAL = {Advances in Mathematics},
 VOLUME = {198},
 NUMBER = {1},
 YEAR = {2005},
 PAGES = {5--26},
 DOI = {10.1016/j.aim.2004.12.003},
 URL = {https://doi.org/10.1016/j.aim.2004.12.003},
 ISSN = {0001-8708,1090-2082},
 MRCLASS = {53D45 (14D20 14H10)},
}

@article{DZ98,
 AUTHOR = {Dubrovin, Boris and Zhang, Youjin},
 TITLE = {Extended affine {W}eyl groups and {F}robenius manifolds},
 JOURNAL = {Compositio Math.},
 FJOURNAL = {Compositio Mathematica},
 VOLUME = {111},
 NUMBER = {2},
 YEAR = {1998},
 PAGES = {167--219},
}

@article{RS06,
 AUTHOR = {Riley, Andrew and Strachan, Ian A. B.},
 TITLE = {Duality for {J}acobi group orbit spaces and elliptic solutions
of the {WDVV} equations},
 JOURNAL = {Lett. Math. Phys.},
 FJOURNAL = {Letters in Mathematical Physics},
 VOLUME = {77},
 NUMBER = {3},
 YEAR = {2006},
 PAGES = {221--234},
 DOI = {10.1007/s11005-006-0096-0},
 ISSN = {0377-9017,1573-0530},
 MRCLASS = {53D45 (11G55 33B30)},
}

@article{RS07,
 AUTHOR = {Riley, Andrew and Strachan, Ian A. B.},
 TITLE = {A note on the relationship between rational and trigonometric
solutions of the {WDVV} equations},
 JOURNAL = {J. Nonlinear Math. Phys.},
 FJOURNAL = {Journal of Nonlinear Mathematical Physics},
 VOLUME = {14},
 NUMBER = {1},
 YEAR = {2007},
 PAGES = {82--94},
 ISSN = {1402-9251,1776-0852},
 MRCLASS = {53D45},
}

@article{Ber00,
 AUTHOR = {Bertola, Marco},
 TITLE = {Frobenius manifold structure on orbit space of {J}acobi
groups. {II}},
 JOURNAL = {Differential Geom. Appl.},
 FJOURNAL = {Differential Geometry and its Applications},
 VOLUME = {13},
 NUMBER = {3},
 YEAR = {2000},
 PAGES = {213--233},
 DOI = {10.1016/S0926-2245(00)00027-9},
 URL = {https://doi.org/10.1016/S0926-2245(00)00027-9},
 ISSN = {0926-2245,1872-6984},
 MRCLASS = {11F50 (53D45)},
}

@article{Ber00b,
 AUTHOR = {Bertola, Marco},
 TITLE = {Frobenius manifold structure on orbit space of {J}acobi
groups. {I}},
 JOURNAL = {Differential Geom. Appl.},
 FJOURNAL = {Differential Geometry and its Applications},
 VOLUME = {13},
 NUMBER = {1},
 YEAR = {2000},
 PAGES = {19--41},
 DOI = {10.1016/S0926-2245(00)00026-7},
 URL = {https://doi.org/10.1016/S0926-2245(00)00026-7},
 ISSN = {0926-2245,1872-6984},
 MRCLASS = {11F50 (11F22 53D45)},
}

@article{Zuo20,
 AUTHOR = {Zuo, Dafeng},
 TITLE = {Frobenius manifolds and a new class of extended affine {W}eyl
groups of {A}-type},
 JOURNAL = {Lett. Math. Phys.},
 FJOURNAL = {Letters in Mathematical Physics},
 VOLUME = {110},
 NUMBER = {7},
 YEAR = {2020},
 PAGES = {1903--1940},
}

@article{MZ24,
 AUTHOR = {Ma, Shilin and Zuo, Dafeng},
 TITLE = {Frobenius {M}anifolds and a {N}ew {C}lass of {E}xtended
{A}ffine {W}eyl {G}roups of {A}-type ({II})},
 JOURNAL = {Commun. Math. Stat.},
 FJOURNAL = {Communications in Mathematics and Statistics},
 VOLUME = {12},
 NUMBER = {4},
 YEAR = {2024},
 PAGES = {617--632},
}

@article{PS26,
 AUTHOR = {Alessandro Proserpio and Ian A. B. Strachan},
 TITLE = {Diagonal invariants and genus-zero {H}urwitz {F}robenius manifolds},
 JOURNAL = {Selecta Math. (N.S.)},
 FJOURNAL = {Selecta Mathematica. New Series},
 VOLUME = {32},
 NUMBER = {13},
 YEAR = {2026},
}

@incollection{Dub98,
 AUTHOR = {Dubrovin, Boris},
 TITLE = {Differential geometry of the space of orbits of a {C}oxeter
group},
 BOOKTITLE = {Surveys in differential geometry: integral systems},
 SERIES = {Surv. Differ. Geom.},
 VOLUME = {4},
 YEAR = {1998},
 PAGES = {181--211},
 PUBLISHER = {Int. Press, Boston, MA},
 ISBN = {1-57146-066-7},
}

@article{BvG00,
 AUTHOR = {Andrea Brini and Karoline van Gemst},
 TITLE = {{S}aito discriminant strata for extended affine {W}eyl groups},
 JOURNAL = {in preparation},
}

@article{MT08,
 AUTHOR = {Milanov, Todor E. and Tseng, Hsian-Hua},
 TITLE = {The spaces of {L}aurent polynomials, {G}romov-{W}itten theory
of {$\Bbb P^1$}-orbifolds, and integrable hierarchies},
 JOURNAL = {J. Reine Angew. Math.},
 FJOURNAL = {Journal f\"ur die Reine und Angewandte Mathematik. [Crelle's
Journal]},
 VOLUME = {622},
 YEAR = {2008},
 PAGES = {189--235},
 DOI = {10.1515/CRELLE.2008.069},
 URL = {https://doi.org/10.1515/CRELLE.2008.069},
 ISSN = {0075-4102,1435-5345},
}

@article{EO09,
 AUTHOR = {Eynard, B. and Orantin, N.},
 TITLE = {{Topological recursion in enumerative geometry and random
matrices}},
 JOURNAL = {J. Phys.},
 VOLUME = {A42},
 NUMBER = {29},
 YEAR = {2009},
 PAGES = {293001},
 DOI = {10.1088/1751-8113/42/29/293001},
 SLACCITATION = {"%%CITATION = JPAGA,A42,293001;%%"
}}

@article{AL13,
 AUTHOR = {Arsie, Alessandro and Lorenzoni, Paolo},
 TITLE = {From the {D}arboux-{E}gorov system to bi-flat {F}-manifolds},
 JOURNAL = {J. Geom. Phys.},
 FJOURNAL = {Journal of Geometry and Physics},
 VOLUME = {70},
 YEAR = {2013},
 PAGES = {98--116},
 DOI = {10.1016/j.geomphys.2013.03.023},
 URL = {https://doi.org/10.1016/j.geomphys.2013.03.023},
 ISSN = {0393-0440,1879-1662},
 MRREVIEWER = {Dafeng\ Zuo},
}

@article{BvG22,
 AUTHOR = {Brini, Andrea and van Gemst, Karoline},
 TITLE = {Mirror symmetry for extended affine {W}eyl groups},
 JOURNAL = {J. \'Ec. polytech. Math.},
 FJOURNAL = {Journal de l'\'Ecole polytechnique. Math\'ematiques},
 VOLUME = {9},
 YEAR = {2022},
 PAGES = {907--957},
 DOI = {10.5802/jep.197},
 URL = {https://doi.org/10.5802/jep.197},
 ISSN = {2429-7100,2270-518X},
 MRCLASS = {53D45},
 MRREVIEWER = {Hsian-Hua\ Tseng},
 SLACCITATION = {"%%CITATION = ARXIV:1510.08690;%%"
}}

@article{MWZ24,
 AUTHOR = {Shilin Ma and Chao-Zhong Wu and Dafeng Zuo},
 TITLE = {Infinite-dimensional {F}robenius {M}anifolds and {E}xtensions of {G}enus-{Z}ero {W}hitham {H}ierarchies},
 YEAR = {2024},
 URL = {https://arxiv.org/abs/2406.08239},
 EPRINT = {arXiv:2406.08239},
}

@article{Sai81,
 AUTHOR = {Saito, Kyoji},
 TITLE = {Primitive forms for a universal unfolding of a function with
an isolated critical point},
 JOURNAL = {J. Fac. Sci. Univ. Tokyo Sect. IA Math.},
 FJOURNAL = {Journal of the Faculty of Science. University of Tokyo.
Section IA. Mathematics},
 VOLUME = {28},
 NUMBER = {3},
 YEAR = {1981},
 PAGES = {775--792 (1982)},
 ISSN = {0040-8980},
}

@article{Sai83,
 AUTHOR = {Saito, Kyoji},
 TITLE = {Period mapping associated to a primitive form},
 JOURNAL = {Publ. Res. Inst. Math. Sci.},
 FJOURNAL = {Kyoto University. Research Institute for Mathematical
Sciences. Publications},
 VOLUME = {19},
 NUMBER = {3},
 YEAR = {1983},
 PAGES = {1231--1264},
 DOI = {10.2977/prims/1195182028},
 URL = {https://doi.org/10.2977/prims/1195182028},
 ISSN = {0034-5318,1663-4926},
 MRCLASS = {32G11 (32B30)},
 MRREVIEWER = {Helmut\ Hamm},
}

@article{PS25,
 AUTHOR = {Alessandro Proserpio and Ian A. B. Strachan},
 TITLE = {{D}ubrovin duality for open {H}urwitz flat {F}-manifolds},
 YEAR = {2025},
 EPRINT = {arXiv:2512.08795},
}

@article{LQZ25,
 AUTHOR = {Liu, Si-Qi and Qu, Haonan and Zhang, Youjin},
 TITLE = {Generalized {F}robenius manifolds with non-flat unity and
integrable hierarchies},
 JOURNAL = {Comm. Math. Phys.},
 FJOURNAL = {Communications in Mathematical Physics},
 VOLUME = {406},
 NUMBER = {4},
 YEAR = {2025},
 PAGES = {Paper No. 77, 92},
 ISSN = {0010-3616,1432-0916},
 MRCLASS = {53D45 (37K10 70G45 70H06)},
}

@incollection{Dub98b,
 AUTHOR = {Dubrovin, Boris},
 TITLE = {Flat pencils of metrics and {F}robenius manifolds},
 BOOKTITLE = {Integrable systems and algebraic geometry ({K}obe/{K}yoto,
1997)},
 YEAR = {1998},
 PAGES = {47--72},
 PUBLISHER = {World Sci. Publ., River Edge, NJ},
}

@article{Rej23,
    AUTHOR = {Rejeb, Chaabane},
 TITLE = {New formula for the prepotentials associated with {H}urwitz-{F}robenius manifolds and generalized {WDVV} equations},
 YEAR = {2023},
 EPRINT = {arXiv:2312.00317},
}

@article{AFS20,
 AUTHOR = {Antoniou, G. and Feigin, M. V.  and Strachan, I. A. B.},
 TITLE = {The {S}aito determinant for {C}oxeter discriminant strata},
 JOURNAL = {arXiv:2008.10133v1},
 YEAR = {2020}}

@article{Alm22,
 AUTHOR = {Almeida, Guilherme F.},
 TITLE = {The differential geometry of the orbit space of extended
affine {J}acobi group {$A_ n$}},
 JOURNAL = {J. Geom. Phys.},
 FJOURNAL = {Journal of Geometry and Physics},
 VOLUME = {171},
 YEAR = {2022},
 PAGES = {Paper No. 104409, 52},
 DOI = {10.1016/j.geomphys.2021.104409},
 ISSN = {0393-0440,1879-1662},
 MRCLASS = {53D45},
}

@article {Jon96,
    AUTHOR = {Jones, Andrew R.},
     TITLE = {A combinatorial approach to the double cosets of the symmetric
              group with respect to {Y}oung subgroups},
   JOURNAL = {European J. Combin.},
  FJOURNAL = {European Journal of Combinatorics},
    VOLUME = {17},
      YEAR = {1996},
    NUMBER = {7},
     PAGES = {647--655},
      ISSN = {0195-6698,1095-9971},
   MRCLASS = {20B30 (05A15)},
MRREVIEWER = {G.\ D.\ James},
       DOI = {10.1006/eujc.1996.0056},
       URL = {https://doi.org/10.1006/eujc.1996.0056},
}

@article {ACV03,
    AUTHOR = {Abramovich, Dan and Corti, Alessio and Vistoli, Angelo},
     TITLE = {Twisted bundles and admissible covers},
   JOURNAL = {Comm. Algebra},
  FJOURNAL = {Communications in Algebra},
    VOLUME = {31},
      YEAR = {2003},
    NUMBER = {8},
     PAGES = {3547--3618},
      ISSN = {0092-7872,1532-4125},
MRREVIEWER = {Andrew\ Kresch},
       DOI = {10.1081/AGB-120022434},
       URL = {https://doi.org/10.1081/AGB-120022434},
}

@article {HM82,
    AUTHOR = {Harris, Joe and Mumford, David},
     TITLE = {On the {K}odaira dimension of the moduli space of curves},
   JOURNAL = {Invent. Math.},
  FJOURNAL = {Inventiones Mathematicae},
    VOLUME = {67},
      YEAR = {1982},
    NUMBER = {1},
     PAGES = {23--88},
      ISSN = {0020-9910,1432-1297},
MRREVIEWER = {Yujiro\ Kawamata},
       DOI = {10.1007/BF01393371},
       URL = {https://doi.org/10.1007/BF01393371},
}

@incollection {DNOPS18,
    AUTHOR = {Dunin-Barkowski, P. and Norbury, P. and Orantin, N. and
              Popolitov, A. and Shadrin, S.},
     TITLE = {Primary invariants of {H}urwitz {F}robenius manifolds},
 BOOKTITLE = {Topological recursion and its influence in analysis, geometry,
              and topology},
    SERIES = {Proc. Sympos. Pure Math.},
    VOLUME = {100},
     PAGES = {297--331},
 PUBLISHER = {Amer. Math. Soc., Providence, RI},
      YEAR = {2018},
      ISBN = {978-1-4704-3541-7},
   MRCLASS = {53D45 (32G15 81T45)},
MRREVIEWER = {Hsian-Hua\ Tseng},
       DOI = {10.1090/pspum/100/01768},
       URL = {https://doi.org/10.1090/pspum/100/01768},
}

@article {Ful69,
    AUTHOR = {Fulton, William},
     TITLE = {Hurwitz schemes and irreducibility of moduli of algebraic
              curves},
   JOURNAL = {Ann. of Math. (2)},
  FJOURNAL = {Annals of Mathematics. Second Series},
    VOLUME = {90},
      YEAR = {1969},
     PAGES = {542--575},
      ISSN = {0003-486X},
   MRCLASS = {14.20},
MRREVIEWER = {S.\ L.\ Kleiman},
       DOI = {10.2307/1970748},
       URL = {https://doi.org/10.2307/1970748},
}

@incollection {Lan19,
    AUTHOR = {Lando, Sergei K.},
     TITLE = {On {D}ubrovin's {F}robenius structures on {H}urwitz spaces},
 BOOKTITLE = {Primitive forms and related subjects---{K}avli {IPMU} 2014},
    SERIES = {Adv. Stud. Pure Math.},
    VOLUME = {83},
     PAGES = {221--236},
 PUBLISHER = {Math. Soc. Japan, [Tokyo]},
      YEAR = {2019},
      ISBN = {978-4-86497-085-3},
   MRCLASS = {14H15 (14H70 30F20)},
}

@article{PvGprep,
      title={Frobenius manifolds from the relativistic {T}oda chain: beyond {D}ynkin type {A}}, 
      author={Proserpio, Alessandro and van Gemst, Karoline},
      journal={in preparation},
}

@incollection {Kon94,
    AUTHOR = {Kontsevich, Maxim},
     TITLE = {Enumeration of rational curves via torus actions},
 BOOKTITLE = {The moduli space of curves ({T}exel {I}sland, 1994)},
    SERIES = {Progr. Math.},
    VOLUME = {129},
     PAGES = {335--368},
 PUBLISHER = {Birkh\"auser Boston, Boston, MA},
      YEAR = {1995},
      ISBN = {0-8176-3784-2},
   MRCLASS = {14N10 (14D22 14L30)},
MRREVIEWER = {Anatoly\ Libgober},
       DOI = {10.1007/978-1-4612-4264-2\_12},
       URL = {https://doi.org/10.1007/978-1-4612-4264-2_12},
}

@article {ELSSV01,
    AUTHOR = {Ekedahl, Torsten and Lando, Sergei and Shapiro, Michael and
              Vainshtein, Alek},
     TITLE = {Hurwitz numbers and intersections on moduli spaces of curves},
   JOURNAL = {Invent. Math.},
  FJOURNAL = {Inventiones Mathematicae},
    VOLUME = {146},
      YEAR = {2001},
    NUMBER = {2},
     PAGES = {297--327},
      ISSN = {0020-9910,1432-1297},
   MRCLASS = {14H30 (14H10 14N35)},
MRREVIEWER = {Ravi\ D.\ Vakil},
       DOI = {10.1007/s002220100164},
       URL = {https://doi.org/10.1007/s002220100164},
}

@article {MST16,
    AUTHOR = {Milanov, Todor and Shen, Yefeng and Tseng, Hsian-Hua},
     TITLE = {Gromov-{W}itten theory of {F}ano orbifold curves, gamma
              integral structures and {ADE}-{T}oda hierarchies},
   JOURNAL = {Geom. Topol.},
  FJOURNAL = {Geometry \& Topology},
    VOLUME = {20},
      YEAR = {2016},
    NUMBER = {4},
     PAGES = {2135--2218},
      ISSN = {1465-3060,1364-0380},
   MRCLASS = {14N35 (14J45 17B67 17B69)},
MRREVIEWER = {Pierre-Emmanuel\ Chaput},
       DOI = {10.2140/gt.2016.20.2135},
       URL = {https://doi.org/10.2140/gt.2016.20.2135},
}

@article {WX12,
    AUTHOR = {Wu, Chao-Zhong and Xu, Dingdian},
     TITLE = {A class of infinite-dimensional {F}robenius manifolds and
              their submanifolds},
   JOURNAL = {Int. Math. Res. Not. IMRN},
  FJOURNAL = {International Mathematics Research Notices. IMRN},
      YEAR = {2012},
    NUMBER = {19},
     PAGES = {4520--4562},
      ISSN = {1073-7928,1687-0247},
   MRCLASS = {58D15 (30D30 53D45 58B20)},
MRREVIEWER = {Dafeng\ Zuo},
       DOI = {10.1093/imrn/rnr192},
       URL = {https://doi.org/10.1093/imrn/rnr192},
}

@article {DNOPS19,
    AUTHOR = {Dunin-Barkowski, P. and Norbury, P. and Orantin, N. and
              Popolitov, A. and Shadrin, S.},
     TITLE = {Dubrovin's superpotential as a global spectral curve},
   JOURNAL = {J. Inst. Math. Jussieu},
  FJOURNAL = {Journal of the Institute of Mathematics of Jussieu. JIMJ.
              Journal de l'Institut de Math\'ematiques de Jussieu},
    VOLUME = {18},
      YEAR = {2019},
    NUMBER = {3},
     PAGES = {449--497},
      ISSN = {1474-7480,1475-3030},
   MRCLASS = {53D45 (14H81 32G15 81T45)},
MRREVIEWER = {Felix\ Janda},
       DOI = {10.1017/s147474801700007x},
       URL = {https://doi.org/10.1017/s147474801700007x},
}

@article {DOSS14,
    AUTHOR = {Dunin-Barkowski, P. and Orantin, N. and Shadrin, S. and Spitz,
              L.},
     TITLE = {Identification of the {G}ivental formula with the spectral
              curve topological recursion procedure},
   JOURNAL = {Comm. Math. Phys.},
  FJOURNAL = {Communications in Mathematical Physics},
    VOLUME = {328},
      YEAR = {2014},
    NUMBER = {2},
     PAGES = {669--700},
      ISSN = {0010-3616,1432-0916},
   MRCLASS = {81T45 (14N35 53D45)},
MRREVIEWER = {Wan\ Keng\ Cheong},
       DOI = {10.1007/s00220-014-1887-2},
       URL = {https://doi.org/10.1007/s00220-014-1887-2},
}

@article {KK06,
    AUTHOR = {Kokotov, A. and Korotkin, D.},
     TITLE = {Isomonodromic tau-function of {H}urwitz {F}robenius manifolds
              and its applications},
   JOURNAL = {Int. Math. Res. Not.},
  FJOURNAL = {International Mathematics Research Notices},
      YEAR = {2006},
     PAGES = {Art. ID 18746, 34},
      ISSN = {1073-7928,1687-0247},
   MRCLASS = {53D45 (14H30)},
MRREVIEWER = {Hsian-Hua\ Tseng},
       DOI = {10.1155/IMRN/2006/18746},
       URL = {https://doi.org/10.1155/IMRN/2006/18746},
}

@article {KS05,
    AUTHOR = {Kokotov, Alexey and Strachan, Ian A. B.},
     TITLE = {On the isomonodromic tau-function for the {H}urwitz spaces of
              branched coverings of genus zero and one},
   JOURNAL = {Math. Res. Lett.},
  FJOURNAL = {Mathematical Research Letters},
    VOLUME = {12},
      YEAR = {2005},
    NUMBER = {5-6},
     PAGES = {857--875},
      ISSN = {1073-2780},
   MRCLASS = {53D45 (34M35 34M55 37J35 37K20)},
       DOI = {10.4310/MRL.2005.v12.n6.a7},
       URL = {https://doi.org/10.4310/MRL.2005.v12.n6.a7},
}

@article {Str99,
    AUTHOR = {Strachan, I. A. B.},
     TITLE = {Degenerate {F}robenius manifolds and the bi-{H}amiltonian
              structure of rational {L}ax equations},
   JOURNAL = {J. Math. Phys.},
  FJOURNAL = {Journal of Mathematical Physics},
    VOLUME = {40},
      YEAR = {1999},
    NUMBER = {10},
     PAGES = {5058--5079},
      ISSN = {0022-2488,1089-7658},
   MRCLASS = {37K10 (37K05 53D45)},
       DOI = {10.1063/1.533015},
       URL = {https://doi.org/10.1063/1.533015},
}

@article {Tel12,
    AUTHOR = {Teleman, Constantin},
     TITLE = {The structure of 2{D} semi-simple field theories},
   JOURNAL = {Invent. Math.},
  FJOURNAL = {Inventiones Mathematicae},
    VOLUME = {188},
      YEAR = {2012},
    NUMBER = {3},
     PAGES = {525--588},
      ISSN = {0020-9910,1432-1297},
   MRCLASS = {57R56 (18D10 53D45)},
MRREVIEWER = {Julia\ Bergner},
       DOI = {10.1007/s00222-011-0352-5},
       URL = {https://doi.org/10.1007/s00222-011-0352-5},
}

@article {Giv01,
    AUTHOR = {Givental, Alexander B.},
     TITLE = {Semisimple {F}robenius structures at higher genus},
   JOURNAL = {Internat. Math. Res. Notices},
  FJOURNAL = {International Mathematics Research Notices},
      YEAR = {2001},
    NUMBER = {23},
     PAGES = {1265--1286},
      ISSN = {1073-7928,1687-0247},
   MRCLASS = {53D45 (14N35)},
MRREVIEWER = {Gilberto\ Bini},
       DOI = {10.1155/S1073792801000605},
       URL = {https://doi.org/10.1155/S1073792801000605},
}

@article {Str01,
    AUTHOR = {Strachan, I. A. B.},
     TITLE = {Frobenius submanifolds},
   JOURNAL = {J. Geom. Phys.},
  FJOURNAL = {Journal of Geometry and Physics},
    VOLUME = {38},
      YEAR = {2001},
    NUMBER = {3-4},
     PAGES = {285--307},
      ISSN = {0393-0440,1879-1662},
   MRCLASS = {53D45},
MRREVIEWER = {Andreas\ Gathmann},
       DOI = {10.1016/S0393-0440(00)00064-4},
       URL = {https://doi.org/10.1016/S0393-0440(00)00064-4},
}

@article {Rej25,
    AUTHOR = {Rejeb, Chaabane},
     TITLE = {W{DVV} solutions associated with the genus one holomorphic
              differential},
   JOURNAL = {J. Geom. Phys.},
  FJOURNAL = {Journal of Geometry and Physics},
    VOLUME = {210},
      YEAR = {2025},
     PAGES = {Paper No. 105432, 70},
      ISSN = {0393-0440,1879-1662},
   MRCLASS = {53D45 (35C05 35G20 35Q53 58J60)},
MRREVIEWER = {Iskander\ A.\ Taimanov},
       DOI = {10.1016/j.geomphys.2025.105432},
       URL = {https://doi.org/10.1016/j.geomphys.2025.105432},
}

\end{document}